\documentclass[11pt]{article}

\usepackage{etex}
\usepackage{verbatim}
\usepackage{xspace,enumerate}
\usepackage[dvipsnames]{xcolor}
\usepackage[T1]{fontenc}
\usepackage[full]{textcomp}
\usepackage[american]{babel}
\usepackage{mathtools}
\usepackage{amsthm}
\usepackage{algorithm}
\usepackage[noend]{algpseudocode}
\usepackage{nicefrac}

\usepackage{fullpage}
\usepackage{tcolorbox}
\usepackage{todonotes}

\usepackage{mathpazo}
\usepackage{multicol}
\usepackage{graphicx}
\usepackage{booktabs}
\usepackage{thmtools}
\usepackage{thm-restate}

\usepackage[
letterpaper,
top=1in,
bottom=1in,
left=1in,
right=1in]{geometry}
\usepackage{newpxtext} %
\usepackage{textcomp} %
\usepackage[varg,bigdelims]{newpxmath}
\usepackage[scr=rsfso]{mathalfa}%
\usepackage{bm} %
\let\mathbb\varmathbb
\usepackage{microtype}
\usepackage[pagebackref,bookmarksnumbered,colorlinks=true,urlcolor=blue,linkcolor=blue,citecolor=OliveGreen]{hyperref}
\usepackage[capitalise,nameinlink]{cleveref}
\crefname{lemma}{Lemma}{Lemmas}
\crefname{fact}{Fact}{Facts}
\crefname{theorem}{Theorem}{Theorems}
\crefname{corollary}{Corollary}{Corollaries}
\crefname{claim}{Claim}{Claims}
\crefname{example}{Example}{Examples}
\crefname{problem}{Problem}{Problems}
\crefname{definition}{Definition}{Definitions}
\crefname{exercise}{Exercise}{Exercises}
\usepackage{amsthm}

\newtheorem{theorem}{Theorem}[section]
\newtheorem*{theorem*}{Theorem}
\newtheorem{lemma}[theorem]{Lemma}
\newtheorem*{lemma*}{Lemma}
\newtheorem{fact}[theorem]{Fact}
\newtheorem*{fact*}{Fact}
\newtheorem{proposition}[theorem]{Proposition}
\newtheorem*{proposition*}{Proposition}
\newtheorem{corollary}[theorem]{Corollary}
\newtheorem*{corollary*}{Corollary}

\newtheorem*{hypothesis*}{Hypothesis}
\newtheorem{conjecture}[theorem]{Conjecture}
\newtheorem*{conjecture*}{Conjecture}
\theoremstyle{definition}
\newtheorem{definition}[theorem]{Definition}
\newtheorem*{definition*}{Definition}

\newtheorem*{construction*}{Construction}

\newtheorem*{example*}{Example}

\newtheorem*{question*}{Question}

\newtheorem*{assumption*}{Assumption}

\newtheorem*{problem*}{Problem}

\newtheorem*{openquestion*}{Open Question}
\theoremstyle{remark}
\newtheorem{claim}[theorem]{Claim}
\newtheorem*{claim*}{Claim}
\newtheorem{remark}[theorem]{Remark}
\newtheorem*{remark*}{Remark}

\newtheorem*{observation*}{Observation}
\usepackage{paralist}
\let\originalleft\left
\let\originalright\right
\renewcommand{\left}{\mathopen{}\mathclose\bgroup\originalleft}
\renewcommand{\right}{\aftergroup\egroup\originalright}
\usepackage{turnstile}
\usepackage{mdframed}
\usepackage{tikz}
\usetikzlibrary{positioning}
\usepackage{caption}
\usepackage{newfloat}
\usepackage{array}
\usepackage{subfig}
\usepackage{bbm}
\usepackage{xparse}
\usepackage{amsthm} %
\makeatletter
\let\latexparagraph\paragraph
\RenewDocumentCommand{\paragraph}{som}{%
  \IfBooleanTF{#1}
    {\latexparagraph*{#3}}
    {\IfNoValueTF{#2}
       {\latexparagraph{\maybe@addperiod{#3}}}
       {\latexparagraph[#2]{\maybe@addperiod{#3}}}%
  }%
}
\newcommand{\maybe@addperiod}[1]{%
  #1\@addpunct{.}%
}
\makeatother

\usepackage{boxedminipage}

\newcommand{\ind}[1]{\mathbf{1}_{\Brac{#1}}}

\DeclareMathOperator*{\E}{\Esymb}

\renewcommand{\Pr}{\ProbOp}

\RequirePackage[outline]{contour} 
\contourlength{0.065pt}
\contournumber{10}%

\newcommand{\mper}{\,.}

\newcommand\bdot\bullet

\newcommand{\cA}{\mathcal A}

\renewcommand{\leq}{\leqslant}
\renewcommand{\le}{\leqslant}
\renewcommand{\geq}{\geqslant}
\renewcommand{\ge}{\geqslant}

\let\epsilon=\varepsilon
\numberwithin{equation}{section}
\newcommand\MYcurrentlabel{xxx}
\newcommand{\MYstore}[2]{%
  \global\expandafter \def \csname MYMEMORY #1 \endcsname{#2}%
}
\newcommand{\MYload}[1]{%
  \csname MYMEMORY #1 \endcsname%
}
\newcommand{\MYnewlabel}[1]{%
  \renewcommand\MYcurrentlabel{#1}%
  \MYoldlabel{#1}%
}
\newcommand{\MYdummylabel}[1]{}
\newcommand{\torestate}[1]{%
  \let\MYoldlabel\label%
  \let\label\MYnewlabel%
  #1%
  \MYstore{\MYcurrentlabel}{#1}%
  \let\label\MYoldlabel%
}
\newcommand{\restatedef}[1]{%
  \let\MYoldlabel\label
  \let\label\MYdummylabel
  \begin{definition*}[Restatement of \cref{#1}]
    \MYload{#1}
  \end{definition*}
  \let\label\MYoldlabel
}
\newcommand{\restatetheorem}[1]{%
  \let\MYoldlabel\label
  \let\label\MYdummylabel
  \begin{theorem*}[Restatement of \cref{#1}]
    \MYload{#1}
  \end{theorem*}
  \let\label\MYoldlabel
}
\newcommand{\restatelemma}[1]{%
  \let\MYoldlabel\label
  \let\label\MYdummylabel
  \begin{lemma*}[Restatement of \cref{#1}]
    \MYload{#1}
  \end{lemma*}
  \let\label\MYoldlabel
}
\newcommand{\restateprop}[1]{%
  \let\MYoldlabel\label
  \let\label\MYdummylabel
  \begin{proposition*}[Restatement of \cref{#1}]
    \MYload{#1}
  \end{proposition*}
  \let\label\MYoldlabel
}
\newcommand{\restatefact}[1]{%
  \let\MYoldlabel\label
  \let\label\MYdummylabel
  \begin{fact*}[Restatement of \cref{#1}]
    \MYload{#1}
  \end{fact*}
  \let\label\MYoldlabel
}
\newcommand{\restate}[1]{%
  \let\MYoldlabel\label
  \let\label\MYdummylabel
  \MYload{#1}
  \let\label\MYoldlabel
}

\newcommand{\eps}{\epsilon}

\allowdisplaybreaks
\newcommand{\om}{\om}

\newcommand{\matA}{\mathbf{A}}

\newcommand{\ignore}[1]{}

\usepackage{cleveref}
\usepackage{dsfont}
\usepackage{turnstile}
\usepackage{algorithm}
\usepackage{xfrac}
\usepackage{thm-restate}

\newcommand{\RIGone}{{\textsc{RIG}(n, d, p, q)}}

\renewcommand{\E}{\ensuremath{\mathbb{E}}}
\renewcommand{\Pr}[1]{\ensuremath{\mathbb{P}\left(#1\right)}}
\newcommand{\Prnop}[1]{\ensuremath{\mathbb{P}(#1)}}
\newcommand{\Expected}[1]{\ensuremath{\mathbb{E}\left[#1\right]}}
\newcommand{\Expectedsub}[2]{\ensuremath{\mathbb{E}_{#1}\left[#2\right]}}

\newcommand{\Expectednop}[1]{\ensuremath{\mathbb{E}[#1]}}

\newcommand{\Expectedtildesub}[2]{\ensuremath{\widetilde{\mathbb{E}}_{#1}\left[#2\right]}}
\newcommand{\Expectedtildesubtwo}[2]{\ensuremath{\widetilde{\mathbb{E}}_{#1}\bigg[#2\bigg]}}
\newcommand{\Expectedtildesubnop}[2]{\ensuremath{\widetilde{\mathbb{E}}_{#1}[#2]}}
\newcommand{\Expectedsubnop}[2]{\ensuremath{\mathbb{E}_{#1}[#2]}}

\newcommand{\Vartildesub}[2]{\ensuremath{\widetilde{\mathbb{V}\text{ar}}_{#1}\left[#2\right]}}

\newcommand{\kbip}{k_{(\text{bip})}}
\newcommand{\Gnp}[1]{G(n,{#1})}

\newcommand{\USR}{U_{S,R}}
\newcommand{\USRgood}{U_{S,R}^{(\text{good})}}
\newcommand{\USRbad}{U_{S,R}^{(\text{bad})}}

\newcommand{\modified}{\ensuremath{\epsedge k}}

\newcommand{\neigh}[1]{N_G({#1})}

\newcommand{\assumptionoldslct}{\ensuremath{ \max \big\{ \log\big( \frac{1-q}{1-p} \big), q \big\} \le \frac{1}{2} - \varepsilon }}

\newcommand{\cliqueset}[1]{\mathcal{B}_{#1}}

\newcommand{\mkfoundsizetwo}{4\epsnode / \smallc}
\newcommand{\mkfoundsize}{5\epsnode / \smallc}
\newcommand{\kfoundsizetwo}{(1 - \mkfoundsizetwo )}
\newcommand{\kfoundsize}{(1-\mkfoundsize)}

\newcommand{\adv}[1]{\widetilde{#1}}
\newcommand{\maxbal}[1]{ \max \left\{ \frac{1 - p}{p}, \frac{p}{1-p} \right\}^{#1} }

\newcommand{\maxbalall}[1]{ \left( \frac{1 - p}{p} \right)^{#1} }
\newcommand{\maxbalallcompact}[1]{ \left( (1 - p)/p \right)^{#1} }
\newcommand{\maxbalcomp}[1]{ \max \left\{ (1 - p)/p, p/(1-p) \right\}^{#1} }

\newcommand{\Hg}[1]{H_{#1}(T, \ell)}
\newcommand{\Hgtil}[1]{\widetilde{H}_{#1}(T, \ell, \modified)}

\newcommand{\Hgtilprime}[1]{H_{#1}(T, \ell, \modified)}
\newcommand{\Hgtilprimeprime}[1]{H'_{#1}(T, \ell, \modified)}

\title{Robust Algorithms for Finding Cliques in Random Intersection Graphs via Sum-of-Squares}
\author{Andreas Göbel \and Janosch Ruff \and Leon Schiller}

\date{\small Hasso Plattner Institute, University of Potsdam\\ \texttt{\{firstname.lastname\}@hpi.de}}

\newcommand{\qplus}{q \!\! \uparrow }
\newcommand{\qminus}{q \!\! \downarrow }

\newcommand{\kapprox}{(1 - 2\epsnode)k}

\newcommand{\ktil}{\adv{k}}

\newcommand{\sbm}{\text{MMSB}(n, s, \varepsilon, d, \beta)}
\newcommand{\dirichlet}{\text{Dir}(\beta)}

\newcommand{\overlap}[1]{\mathcal{R}(T)}
\newcommand{\leftt}{S^{\text{(in)}}_\ell(T)}
\newcommand{\rightt}{S^{\text{(out)}}_\ell(T)}

\newcommand{\lefttadv}{\adv{S}^{\text{(in)}}_\ell(T)}
\newcommand{\righttadv}{\adv{S}^{\text{(out)}}_\ell(T)}
\newcommand{\Madv}{M_{\ge \epsedge}}

\renewcommand{\Gnp}{G(n,p)}

\newcommand{\ksplit}{k_{\text{(bip)}}}
\newcommand{\bicliqueaxioms}[1]{\mathcal{B}(H(T), #1)}
\newcommand{\bicliqueaxiomsalgo}{\mathcal{B}(H(T), \ksplit)}
\newcommand{\bicliqueaxiomsalgocompact}{\mathcal{B}(H(T), \ksplit)}

\newcommand{\bicliqueaxiomsalgoapprox}{\widetilde{\mathcal{B}}(G, A, B, k, \gamma)}
\newcommand{\bicliqueaxiomsalgoapproxconcrete}{\widetilde{\mathcal{B}}(\adv{G}, U, V, \ktil, \epsedge k) \cup \{w_T = 1\}}

\newcommand{\bicliqueaxiomsalgoapproxconcretereduced}{\widetilde{\mathcal{B}}(\adv{G}, U, V, \ktil, \epsedge k)}

\newcommand{\degreeapprox}{O(t)}
\newcommand{\degreeexact}{12}

\newcommand{\epsedge}{\varepsilon^{\text{(deg)}}}
\newcommand{\epsnode}{\varepsilon^{\text{(node)}}}

\newcommand{\Asub}{A_{\text{(sub)}}}
\newcommand{\smallc}{\rho}

\newcommand{\Ssize}{\frac{1}{2}k}
\newcommand{\deltap}{\Delta_p}

\newcommand{\pspc}{-0.2cm}

\newcommand{\spiraks}{Nikoletseas, Raptopoulos, and Spirakis \cite{Spirakis-MFCS-20212} }
\newcommand{\spiraksj}{Nikoletseas, Raptopoulos, and Spirakis \cite{nikoletseas2021maximum} }

\newcommand{\BKSl}{Buhai--Kothari--Steurer (BKS) }
\newcommand{\BKSs}{BKS }
\newcommand{\BKSss}{BKS}
\newcommand{\BKS}{Buhai, Kothari, and Steurer \cite{Kothari-STOC-2023} }
\newcommand{\ERgraph}{Erdős--Rényi graph}
\newcommand{\ERgraphs}{Erdős--Rényi graphs}

\begin{document}
\maketitle
\begin{abstract}

We study efficient algorithms for recovering cliques in dense random intersection graphs (RIGs). In this model, $d = n^{\Omega(1)}$ cliques of size approximately $k$ are randomly planted by choosing the vertices to participate in each clique independently with probability $\delta$. While there has been extensive work on recovering one, or multiple disjointly planted cliques in random graphs, the natural extension of this question to recovering overlapping cliques has been, surprisingly, largely unexplored. Moreover, because every vertex can be part of polynomially many cliques, this task is significantly harder than in case of disjointly planted cliques (as recently studied by Kothari, Vempala, Wein and Xu [COLT'23]).

In this work we obtain the first efficient algorithms for recovering the community structure of RIGs both from the perspective of exact and approximate recovery. Our algorithms are further robust to noise, monotone adversaries, a certain, optimal number of edge corruptions, and work whenever $k \gg \sqrt{n \log(n)}$. Our techniques follow the proofs-to-algorithms framework utilizing the sum-of-squares hierarchy. An essential component are certificates for the absence of large cliques outside of the ground-truth. Instead of spectral certificates, a central ingredient are modified versions of the biclique certificates, recently used for semi-random planted clique by Buhai, Kothari and Steurer [STOC'23]. To turn these certificates into robust and efficient algorithms that do not produce ``false positives'', we rely on an extremely sharp concentration property for pseudo-distributions which might be of independent interest.

Our techniques further extend to the related task of efficient \emph{refutation}, and lead to algorithms that can not only recover the ground-truth, but also certify the optimality of this clustering.
\end{abstract}

\thispagestyle{empty} 

\newpage
\thispagestyle{empty} 

\tableofcontents
\thispagestyle{empty} 

\newpage

\section{Introduction}\label{sec:intro}
\setcounter{page}{1}
Finding cliques in graphs is a notoriously difficult and intensely studied problem in theoretical computer science. It was one of the first problems shown to be NP-complete in the worst-case \cite{karp1975computational}. This hardness extends even to the problem of finding $n^\varepsilon$-sized cliques in graphs containing $n^{1-\varepsilon}$ sized cliques \cite{hastad99, zuckerman2006linear}, which has led to the emergence of \emph{average-case} variants of the problem, the most famous of which is   \emph{planted clique} \cite{jerrum1992large, kuvcera1995expected} where a clique of size $k$ is planted in an \ERgraph{} $G(n,p)$. 

While the planted clique can be uniquely identified whenever $k \gg 2\log(n)$, all efficient algorithms require $k \gg \sqrt{n}$. 
In fact, if $k \gg \sqrt{n\log n}$, then the problem can be solved by simply taking the vertices of largest degree, while a series of algorithms based on spectral and semidefinite programming methods achieve the conjecturally optimal threshold $k=\Omega(\sqrt{n})$ \cite{Alon_Krivelevich_Sudakov_1998,Feige_Krauthgamer_2000,mcsherry2001spectral,ames2011nuclear}.\footnote{There are also algorithms based on other techniques see e.g.~\cite{feige2010finding,dekel2014finding,deshpande2015finding,chen2016statistical}.} Lower bounds against several  classes of algorithms such as low-degree polynomials \cite{schramm2022computational}, statistical-query algorithms \cite{feldman2017statistical} and in particular the sum-of-squares hierarchy \cite{meka2015sum,deshpande2015improved,hopkins2018integrality,Barak_Hopkins_Kelner_Kothari_Moitra_Potechin_2019} have led to the belief that the problem exhibits a \emph{computational-statistical gap} in the regime $k \ll \sqrt{n}$.
This planted clique conjecture has further proved crucial for establishing computational-statistical gaps in other statistical inference problems \cite{hajek2015computational,brennan2018reducibility,brennan2020reducibility,bresler2023algorithmic,kothari2023planted,lee2025fundamental}.

While generalizations of planted clique have e.g. focused on ``semirandom'' settings \cite{Kothari-STOC-2023, blasiok2024semirandom, blasiok2024semirandom}, recovering dense-subgraphs \cite{brennan2018reducibility, brennan2020reducibility}, and recovering multiple, yet disjoint cliques \cite{chen2016statistical,kothari2023planted}, another very natural extension of planted clique remains poorly understood: what if many (i.e. $n^{\Omega(1)}$) cliques of size $k$ are planted in a random graph \emph{without} forcing them to be disjoint? 
Even though we might not expect this setting to be fundamentally different from the case of disjoint cliques as considered in recent work \cite{kothari2023planted}, a closer look reveals that all simple, efficient algorithms considered in previous work actually \emph{fail to carry over} to this novel regime. Hence, a natural (and surprisingly un-answered) question that arises is whether recovery in this model is \emph{strictly harder} than planted clique, or whether \emph{new algorithms} are required.

\paragraph{Random Intersection Graphs}
It turns out that this setting of randomly planting many, overlapping cliques is captured by a random graph model called \emph{random intersection graphs} (RIGs) that has received some attention in the literature. We define a noisy version of RIGs as follows. 
\begin{definition}[Random Intersection Graphs with Noise]
    Define a random graph $\textsc{RIG}(n, d, p, q)$ on vertex set $[n]$ as follows. Every vertex $v$ draws a set $M_v \subseteq [d]$ by including each $i \in [d]$ independently with probability $\delta = \delta(d, p, q)$. Then, for every pair $u, v \in [n]$, we include an edge with probability \begin{align*}
        \Pr{u \sim v} = \begin{cases}
            1 & \text{ if } M_v \cap M_u \neq \emptyset\\
            q & \text{ otherwise.}
        \end{cases}
    \end{align*} The probability $\delta = \delta(d, p, q)$ is chosen such that $\Pr{u \sim v} = p$ for all $u,v\in [n], u \neq v$. Notice that $p > q$. 
\end{definition}
We refer to the elements in $[d]$ as \emph{labels}. For every label $\ell \in [d]$, it is easy to see that the set $S_\ell \coloneqq \{v \in [n] \mid \ell \in M_v\}$ is a clique of (expected) size $k \coloneqq \delta n$. Hence, $\RIGone$ can equivalently be seen as a graph in which we plant $d$ cliques $S_1, \ldots, S_d$ of size roughly $k$ uniformly at random, while edges outside of these ``ground-truth'' communities appear independently with probability $q$ while the \emph{overall edge-density} is $p$. 

A natural algorithmic task consists in recovering the sets $S_1, \ldots, S_d$ given only $G \sim \RIGone$. The main focus of this work is the so-called \emph{high-dimensional dense case} where $p$ and $p-q$ are both constants, and where $d \ge n^{\Omega(1)}$. Here, every vertex is part of roughly $\delta d \approx \sqrt{d}$ sets $S_\ell$, each of which has size $\approx n/\sqrt{d}$. Moreover, each $S_\ell$, $S_{\ell'}$ intersects with high probability (w.h.p.)\footnote{Here w.h.p. refers to probability at least $1-O(1/\mathrm{poly(n)})$.} in a polynomial number of vertices. While simple efficient algorithms for recovery in the case $p \le n^{-c}$ for sufficiently large $c$ exist \cite{Behrisch_Taraz_2006}, the dense case turns out to be very challenging\footnote{This is in line with classical planted clique, which is only really challenging if $p$ is a constant, while for $p \le n^{-\varepsilon}$, simple algorithms exist, even if $k \ll \sqrt{n}$. See e.g. the discussion in \cite{Feige_Grinberg_2024}.}, even in the noiseless ($q = 0$) case, and no efficient algorithms are known, even if $k \gg \sqrt{n} \ \text{polylog}(n)$. 

\paragraph{``Fragility'' of simple algorithms}
While somewhat simple algorithms that succeed at recovery for $k \gtrsim \sqrt{n \log(n)}$ can be constructed without the need of advanced techniques such as semidefinite programming\footnote{One such algorithm can for example be built by using the properties that our SoS-certificates rely on in a more direct, combinatorial way. However, it is still much more involved than simple algorithms known for planted clique and planted colouring, which are typically based on taking the vertices of highest degree or counting common neighbours.}, these have the important drawback of \emph{not being robust}. Instead, they can be considered ``fragile'' in the sense that it is easy to deceive them by only introducing a small number of adversarial changes, even if these changes are ``helpful'' in the sense that they only concern edges that are present due to noise. An adversary that is allowed to introduce such ``helpful'' changes is called a \emph{monotone adversary}, and designing algorithms that are robust against its actions turns out to be a much more challenging task. This has also been observed in previous works on planted clique and its variants, and it initiated the search for new algorithms that remain \emph{robust} in semirandom settings as described above, see for example \cite{feige2000finding, Moitra_Perry_Wein_2016, Kothari-STOC-2023}. Moitra, Perry, and Wein \cite{Moitra_Perry_Wein_2016} have even shown that a monotone adversary can \emph{shift} the information-theoretic threshold for weak recovery in the stochastic block model, thereby rendering the problem \emph{strictly harder} even though the adversary only makes ``helpful'' changes.



\paragraph{Prior Work}
From an algorithmic viewpoint, our setting was previously studied by \spiraksj who gave a simple \emph{inefficient} algorithm for recovering the community structure of a dense RIG, thus showing that exact recovery is \emph{information-theoretically} feasible. Subsequently, Christodoulou, Nikoletseas, Raptopoulos, and Spirakis~\cite{Christodoulou_Nikoletseas_Raptopoulos_Spirakis_2023} experimentally studied a spectral heuristic for this problem, but they were not able to give provable performance guarantees. Surprisingly, our results in \Cref{sec:spectralfailure} show that this approach \emph{does provably not work}, at least if $d$ is sufficiently large. In~\cite{nikoletseas2021maximum} and~\cite{Christodoulou_Nikoletseas_Raptopoulos_Spirakis_2023}, finding an efficient algorithm for exact recovery with provable performance guarantees was therefore left as an open problem. It was further suggested as a future direction in~\cite[Section 3, (2)]{bresler2025partial}.
In the sufficiently sparse regime (where $p = n^{-c}$ for sufficiently large exponent $c$), some simple efficient combinatorial algorithms are known \cite{behrisch2009coloring,Behrisch_Taraz_2006}. However, all these approaches completely break down in the dense case. 

\subsection{Results}

In this work, we resolve the open question in \cite{nikoletseas2021maximum} and give the first polynomial time algorithm for recovering the community structure of a dense RIG.
We further put emphasis on \emph{robustness} of our algorithms and show that they continue to work in a natural ``semi-random'' version of our model\footnote{Similar to the planted-clique semi-random model of Feige and Krauthgamer \cite{feige2001heuristics}.} where a \emph{monotone adversary} can delete an arbitrary number of edges that appeared ``due to noise''. We further give optimal algorithms for approximate recovery under the presence of up to $\varepsilon k^2$ additional \emph{edge corruptions} (i.e. arbitrary additions and deletions). In both settings, we focus on the following notions of recovery.
\begin{definition}[Exact and approximate recovery in RIGs]\label{def:recovery}
    Given a $G \sim \RIGone$, possibly after being modified by an adversary. We consider an algorithm $\textsc{Alg}$ that on input $G$ outputs a list $\mathcal{L}$ that consists of exactly $d$ subsets of $[n]$. In this setting, we consider the following recovery guarantees.
    \begin{itemize}
        \setlength\itemsep{.0001em}
        \item We say that $\textsc{Alg}$ achieves \textbf{exact recovery} if for all $\ell \in [d]$ there exists \emph{exactly one} $S \in \mathcal{L}$ such that 
        $S = S_\ell$.
        \item Given some $\rho \in \mathbb{N}$, we say that $\textsc{Alg}$ achieves $\rho$-\textbf{approximate recovery} if for all $\ell \in [d]$, there is \emph{exactly one} $S \in \mathcal{L}$ such that 
        $
            |S \triangle S_\ell| 
            \le \rho
        $ where $A \triangle B$ denotes the \emph{symmetric difference} of two sets $A, B$.
    \end{itemize}
\end{definition}

\vspace{-.7cm}
\paragraph{Exact recovery against a monotone adversary}
Our first main result is a polynomial time algorithm for exact recovery under presence of a \emph{monotone adversary}. Formally, in this ``semirandom'' setting, the adversary is allowed to delete an arbitrary number of edges $\{u,v\}$ such that $M_u \cap M_v = \emptyset$ from $G$ (i.e. edges due to noise). 
Our first main result states that we achieve exact recovery in this setting whenever the expected set size of any $S_\ell$ is $k \gg \sqrt{n \log(n)}$.

\begin{theorem}[Exact recovery against a monotone adversary]\label{thm:exactrecovery}
    There exists a polynomial time algorithm that on input $G \sim \RIGone$, possibly modified by a monotone adversary achieves exact recovery whenever the parameters $p,d,q$ are such that $d \ge n^{\Omega(1)}$\footnote{In this work, the notation $d \gg n^{\Omega(1)}$ means that the results hold whenever $d \ge n^{\alpha}$ for \emph{every}, arbitrarily small constant $\alpha > 0$.}, $k \gg \sqrt{n \log(n)}$ and $p$ is bounded away from $1$. This holds with high probability over the draw of $G$.
\end{theorem}

\begin{remark}[Optimality of the recovery guarantees] We remark that the assumption $k \gg \sqrt{n \log(n)}$ is optimal up to a factor of $O(\sqrt{\log(n)})$ since we can show that efficient recovery is impossible if $k \ll \sqrt{n}$ under the planted clique hypothesis (\Cref{thm:hardness}). 
\end{remark}

\vspace{\pspc}
\paragraph{Approximate recovery under edge corruptions}

In addition to a monotone adversary, we consider a so-called \emph{bounded adversary} that is allowed to make up to $\varepsilon k^2$ arbitrary edge insertions or deletions (called \emph{edge corruptions}) in $G$.
\noindent In this setting, we achieve approximate recovery under the same parameter assumptions as before. 
\begin{theorem}[Approximate recovery under edge corruptions]\label{thm:approxrecovery}
    There exists a polynomial time algorithm that on input $G \sim \RIGone$ possibly modified by a monotone adversary and with up to $\varepsilon k^2$ edge corruptions, achieves $O(\varepsilon k)$-approximate recovery whenever the parameters $p,d,q$ are such that $d \ge n^{\Omega(1)}$, $k \gg \sqrt{n \log(n)}$ and $p$ is bounded away from $1$. This holds with high probability over the draw of $G$ and for all $\varepsilon < \varepsilon_0$ where $\varepsilon_0$ is a constant that depends only on $p$.
\end{theorem}

\begin{remark}[Optimality of the recovery guarantees]
It is not hard to see that exact recovery in this stronger adversarial setting is impossible since the adversary can arbitrarily modify up to a constant (depending on $\varepsilon$) fraction of any ground-truth clique such that the exact original labelling becomes unrecognizable. For example, given any fixed $S_\ell$ it can add all edges between $S_\ell$ and up to $\Omega(\varepsilon k)$ vertices outside of $S_\ell$. Then, it becomes impossible to tell which vertices in this ``extended clique'' are adversarially added. Therefore, $O(\varepsilon k)$-approximate recovery is the best we can hope for. The fact that we require $\varepsilon < \varepsilon_0$ where $\varepsilon_0$ depends on $p$ is also necessary since in the example just mentioned, more vertices can be added to $S_\ell$ the closer $p$ gets to $1$. Moreover, as before, our assumption on $k$ is optimal up to a factor of $O(\sqrt{\log(n)})$.
\end{remark}

\vspace{\pspc}
\paragraph{Improving the ``single label clique theorem''}
Apart from our algorithmic contributions, we also improve upon the understanding of (noisy) random intersection graphs in a more structural sense. Concretely, we prove the natural property that every clique of size at least $(1-\varepsilon)k$ only appears within some ground-truth clique $S_\ell$, i.e., is spanned by a single label. This is called the ``single label clique theorem'' as introduced in \cite{nikoletseas2021maximum} (without efficient algorithms). In fact, our algorithms are in a sense built upon proving this theorem \emph{within the constant-degree sum-of-squares} proof system. This \emph{proof of identifiability} is not only the underpinning of our algorithms, but it also holds in a more general range of parameters as considered in previous work. Concretely, our version extends to all relevant noise levels, and much higher dimensions $d$ as considered in \cite{nikoletseas2021maximum}.

\begin{restatable}[Single label clique theorem]{theorem}{slct}{Single label clique theorem}\label{thm:slct}
    Consider a $\RIGone$. Recall that $k \coloneqq n\delta$ is the expected clique size. Assume further that there are arbitrarily small constants $\alpha, \varepsilon > 0$ such that one of the following parameter assumptions holds.
    \begin{enumerate}
        \setlength\itemsep{.0001em}
        \item[(i)] We have $n^{\alpha} \ll d \ll n^{1-\alpha}$ while $p, q, d$ are such that $k \gg n^{\varepsilon}$.
        \item[(ii)] We have $n^{\alpha} \ll d \ll n^{2-\alpha}$ while $p, q, d$ are such that $k \gg n^{\varepsilon}$ and 
        $
            \assumptionoldslct.
        $
    \end{enumerate}
    Then, the following holds with high probability over the draw of $G \sim \RIGone$. For any clique $K$ of size at least $(1-\varepsilon)k$, there is some $\ell \in [d]$ such that $K \subseteq S_\ell$.
\end{restatable}
While in case $k \lesssim \sqrt{n \log(n)}$, the proof uses a similar (yet strengthened) combinatorial approach as in \cite{Spirakis-MFCS-20212}\footnote{This is the conference version of~\cite{nikoletseas2021maximum} which contains a different proof of the single label clique theorem than the journal version.}, our proof for $k \gtrsim \sqrt{n \log(n)}$ uses an entirely different approach based on certifiable pseudo-random properties in suitably chosen sub-graphs of $\RIGone$.

\vspace{\pspc}
\paragraph{Computational-statistical and detection-recovery gaps} Another interesting consequence of \Cref{thm:slct} is that exact recovery remains \emph{information-theoretically feasible} if $n^\varepsilon \ll k \ll \sqrt{n}$, even though \emph{efficient recovery} is impossible. This uncovers a \emph{computational-statistical gap} as recently shown in case of disjoint communities in \cite{kothari2023planted}. While this is not surprising given the planted clique hypothesis, it also uncovers a similar \emph{detection-recovery gap} as observed in \cite{kothari2023planted} for planted coloring if $n^{1 + \alpha} \ll d \ll n^{2-\alpha}$. This is because in this regime, distinguishing $\RIGone$ from the Erd\H{o}s--R\'enyi model $G(n,p)$ is easy and can be accomplished using a simple statistic such as counting \emph{signed triangles} as shown in \cite{Brennan_Bresler_Nagaraj_2020}.\footnote{After adapting their proof to our noisy case.}

\vspace{\pspc}
\paragraph{Efficient refutation of large cliques} 

What is further remarkable about our new proof is that in case $k \gtrsim \sqrt{n \log(n)}$, it can fully be phrased as a constant-degree sum-of-squares proof. Besides the fact that this is interesting in its own right, this further yields a polynomial time \emph{refutation algorithm} to certify almost tight bounds on the clique number of a given $G \sim \RIGone$. Precisely, for every $\varepsilon > 0$ and parameters such that $k \gg n^{\frac{1}{2} + \varepsilon}$, we can certify in degree-$O(1/\varepsilon)$ SoS that $G$ contains no clique of size larger than $(1 + o(1))k$. This is particularly interesting in light of the fact that simple spectral certificates provably fail at this task as we show in \Cref{sec:spectralfailure}. Moreover, given the ground truth communities as input, there is a refutation algorithm that for some sufficiently small $\varepsilon' > 0$  can certify in degree-$O(1)$ SoS that \emph{all cliques} of size at least $(1-\varepsilon')k$ are \emph{fully} contained in some ground truth community $S_\ell$. This has the further consequence that our algorithm for exact recovery from \Cref{thm:exactrecovery} can not only find every ground-truth $S_\ell$, but also efficiently \emph{certify its own correctness}, i.e. that no further communities besides the ones that were found exist. Details are found in \Cref{sec:refutation}.

\vspace{\pspc}
\paragraph{Related Work}

With regard to community detection, it is further important to point out that RIGs have some resemblance with \emph{mixed membership stochastic block models} (MMSBs) which have been studied as a model for overlapping community detection before, see in particular \cite{Hopkins_Steurer_2017, Anandkumar_Ge_Hsu_Kakade_2013}. We discuss the similarities and differences between these two models in more detail in \Cref{sec:mmsb}. For now, we emphasize that the main differences between community detection in MMSBs and our setting are that (i) we achieve stronger recovery guarantees as previous work (exact/approximate recovery for \emph{every} community instead of only weak recovery in \cite{Anandkumar_Ge_Hsu_Kakade_2013, Hopkins_Steurer_2017}), (ii) we focus on the high-dimensional setting where the number of communities grows polynomially in $n$, and (iii) we put emphasis on the robustness of our algorithms, while the known algorithms for MMSBs are somewhat fragile and cannot easily handle adversaries. 

We would further like to point out that RIGs have also been studied from rather different perspectives and serve as an increasingly popular model for random graphs with underlying latent structure in various settings. For examples, we refer to the remarkable work of Brennan, Bresler, and Nagaraj \cite{Brennan_Bresler_Nagaraj_2020} who characterize the thresholds for distinguishability of RIGs from Erdős–Rényi graphs as $d \rightarrow \infty$. As a further example, we refer to Liu and Austern \cite{Liu_Austern_2025} who recently studied a noisy version of RIGs under the lens of \emph{graph matching}. In the context of recovering potentially overlapping cliques random graphs, we further find it important to point to the recent works \cite{bresler2025partial, bresler2024thresholds} where the re-constructibility of a random hypergraph from its graph projection is studied. As the authors note, this model and the reconstruction problem considered there is highly similar to our setting. However, in \cite{bresler2025partial, bresler2024thresholds} this problem is studied in a very different regime of parameters since the focus is on the case where $k$ is \emph{constant}, while we consider large $k$ in the order of at least $\sqrt{n}$. 

\section{Techniques}\label{sec:techniques}


Our algorithms are an instantiation of the so-called \emph{proofs-to-algorithms framework} (for an excelent exposition see for example \cite{Fleming_Kothari_Pitassi_2019}) that utilizes the \emph{sum-of-squares} hierarchy. The general idea is that an algorithm follows from two ingredients. First, a mechanism for \emph{certifying} ``goodness'' of a solution candidate, and second, a \emph{proof of identifiability} which states that any solution that passes our certification step is close to the underlying ground-truth set of parameters we are trying to estimate. 

While in general, this only immediately yields an \emph{inefficient} algorithm (simply enumerate all solution candidates and check whether they pass our test of goodness), it turns out that this approach can be made efficient if our certification and proof of identifiability can be phrased as a \emph{low-degree sum-of-squares (SoS) proof}.\footnote{See \Cref{sec:sosprelims} for a more detailed exposition.}

\subsection{Simple, inefficient algorithm}\label{sec:inefficientalgo}

Our approach towards a simple, correct (yet inefficient) algorithm follows a similar high-level idea as employed in previous work for average-case clustering problems (see for example \cite{kothariclustering, Fleming_Kothari_Pitassi_2019}). Specifically, we enumerate over all possible clusters (in our case subsets of roughly $k$ vertices) and check if they ``seem like one of the clusters we are looking for''.
Then, we wish to find a proof of identifiability which states that every cluster that passes our test is close to one of the ground-truth cliques we wish to find.

While in the case of clustering data in $\mathbb{R}^d$ the certification of ``goodness'' typically involves more complex properties like higher-order moments or anti-concentration properties, in our case we simply need to check if our solution candidate is a clique. 
Given this certificate, our proof of identifiability is given precisely by the \emph{single label clique theorem} (\Cref{thm:slct}). Hence, given that we can show \Cref{thm:slct}, a simple correct (yet inefficient) algorithm consists in simply enumerating all maximal cliques of size at least $(1 - \varepsilon)k$.



\paragraph{Combinatorial proof of identifiability}

Let us start by describing the approach of \spiraks towards proving the single label clique theorem, as it appeared in the conference version of the paper. To show that any clique $K$ of size $(1-\varepsilon)k$ is contained in some ground-truth set $S_\ell$ (i.e. spanned by a single label), we assume that $K$ is not spanned by a single label and proceed in two natural steps to arrive at a contradiction. The first step consists in proving that for every $S_\ell$, the clique $K$ cannot intersect $S_\ell$ in more than $\approx k/\sqrt{d}$ vertices. This can be done using relatively standard arguments for bounding the probability that a biclique forms within the boundary of $S_\ell$. These arguments are presented in \Cref{sec:cliqueintersections} and imply that $K$ must be ``scattered'' across the entire graph in the sense that no $S_\ell$ contains more than $\approx k/\sqrt{d}$ vertices of $K$. In a second step, one could hope to show that no such ``scattered'' clique can exist since it would imply a decently large clique in which every edge is formed by a distinct label. 

While the proof of this second step in the conference version \cite{Spirakis-MFCS-20212} contained an error,\footnote{This error was somewhat addressed in \cite{nikoletseas2021maximum} but it does not easily translate to the noisy case.} the above approach works if $p, q$ are sufficiently small and yields a proof of identifiability, which we present in a strengthened and corrected form in \Cref{sec:slctalpha02}. However, translating this very approach into an SoS-proof seems almost impossible due to its combinatorial character and the many case distinctions involved. 

\subsection{Towards efficient algorithms}

Hence, to obtain an efficient algorithm, the main challenge in making the above ideas work is to find new ways of ``certifying'' that \emph{any sufficiently large clique in $G$ has large overlap with some set $S_\ell$}. This will be the main focus for the rest of this section.

\subsubsection{ Failure of spectral certificates }\label{sec:spectralfailure}

A natural approach towards certifying the absence of cliques in random graphs since the seminal work of Alon, Krivelevich and Sudakov \cite{Alon_Krivelevich_Sudakov_1998} uses the \emph{spectral norm} $\| \cdot \|$ of its centred adjacency matrix $\mathbf{A}$. Focussing on the case of $p = 1/2$ for now, this matrix is defined as $\mathbf{A}_{u,v} = 1$ if $\{u, v\} \in E(G)$ and $\mathbf{A}_{u,v} = -1$ otherwise. By standard results from random matrix theory, it is well known that for an \ERgraph{} $\Gnp$, we have $\| \mathbf{A} \| \le O(\sqrt{n})$ w.h.p. Moreover, it is easy to see that in any graph $G$ that contains a clique of size $k$, we have $\|\mathbf{A}\| \ge k$. These two simple facts directly yield a SoS-certificate for the absence of a clique of size $k \gg \sqrt{n}$ in a $G(n,1/2)$.

In our case, we know that there are cliques of size roughly $k = n\delta$, so the spectral norm of $\RIGone$ will be at least $k$, but we could hope that a spectral certificate yields the absence of cliques of size $ck$ for some constant $c$. However, as we show in \cref{sec:spectrallowerbound}, in dense RIGs this is very far from the truth, even in the noiseless case. 
\begin{restatable}{theorem}{SpectralLB}
    Let $\mathbf{A}$ be the centered adjacency matrix of a $G \sim \RIGone$ with $p = 1/2$ and $q = 0$. Then there is a constant $C$ such that, w.h.p.,
    $\| \mathbf{A}\| \ge C n d^{-1/4} = \Omega(d^{{1/4}} k) \gg k$.
\end{restatable}
This shows that the spectral norm of $\mathbf{A}$ exceeds $k$ by a factor of $\Omega(d^{1/4})$ with high probability. The reason is that the eigenspace of $\mathbf{A}$ is dominated by a somewhat \emph{de-localized} eigenvector supported on $\Omega(n)$ vertices that have chosen an above-average number of labels. We remark that even after a neighbourhood reduction (which is discussed in the next sub-section), a similar argument shows that the spectral norm still exceeds $k$ by a polynomial factor, so a simple spectral certificate fails nonetheless.




While prior work (see for example \cite[Chapter 2]{Le_2016} or \cite{Chen_Ding_Hua_Tiegel_2025}) has employed so-called regularization techniques to deal with degree fluctuations in sparse graphs with independent edges, we are dealing with dense graphs and significant dependence where no regularization techniques known to us are easily applicable. Furthermore, the known approaches for obtaining spectral norm bounds in dependent random graphs (see e.g. \cite{Li_Schramm_2024, Bangachev_Bresler_2024, Baguley_Goebel_Pappik_Schiller_2025}) are typically too lossy for our setting. In addition, even if a simple regularization methods such as removing high-degree vertices worked, then the obtained certificates would likely only apply certain sub-graphs of $G$, which gives us no hope of further \emph{refuting} the presence of cliques globally. 


\subsubsection{ Biclique certificates via pseudo-randomness }

Instead of relying on spectral methods, we use certificates that rely on different \emph{pseudo-random} properties in suitably chosen bipartite sub-graphs of a $G \sim \RIGone$.

\vspace{\pspc}
\paragraph{ Certificates based on balancedness by \BKSl }
Recent work by \BKS introduced a $O(1/\varepsilon)$-degree SoS-certificate for the absence of $n^{\varepsilon} \times k$ bicliques in random bipartite graphs. This certificate relies on a pseudo-random property of bipartite graphs called \emph{balancedness}.

Concretely, a bipartite graph $H = (A \uplus B, E)$ with $\{\pm 1\}$ adjacency matrix is said to have $r$-fold balancedness $\Delta$ if for every $S \subseteq A$ with $|S| = 2r$, we have that $|\sum_{u\in B} \prod_{v \in B} \mathbf{A}(u, v)| \le \Delta$ (the formal, more general definition is found in \Cref{sec:balancedness}). It follows from standard concentration bounds that a random bipartite (Erdős–Rényi) graph with $p = 1/2$ has $r$-fold balancedness $O(\sqrt{n\log(n)})$ for every fixed $r$ w.h.p. 

\BKS have shown that in a balanced bipartite graph with such $r$-fold balancedness, if there were a biclique of size $\Omega(n/k)^{1/r} \times k$, one could extract $\Omega(n/k)$ vectors $v_1, \ldots, v_m$ in $\mathbb{R}^n$ with pairwise negative correlation $ \langle v_i, v_j \rangle \le - ck$ for some constant $c > 0$. However, a simple argument shows that the number of such vectors in $n$ dimensions is bounded by $O(n/k)$, which yields a contradiction. All of this can be done in degree $O(r)$-SoS and gives meaningful certificates for the absence of cliques in any bipartite graph with $r$-fold balancedness $\le ck$ where $c$ is any constant strictly smaller than one.

We remark that our approach in its simplest form uses essentially these very certificates as introduced in \cite{Kothari-STOC-2023}, however, a notable difference is that the graphs we consider do provably not have the same (strong) balancedness properties as plain \ERgraphs{} as considered in \cite{Kothari-STOC-2023}. Therefore, we actually use a generalized version of these \BKSss-certificates that work in a setting where not every pair of distinct vectors $v_1, \ldots, v_m$ has negative correlation, but only most of them. This comes at the cost of slightly weaker certificates. 

\vspace{\pspc}
\paragraph{When do balanced subgraphs appear in $\RIGone$?}

To use \BKSss-inspired certificates, in RIGs, we wish to understand where balanced sub-graphs appear in a $G \sim \RIGone$.
A natural candidate as considered in \cite{Kothari-STOC-2023} would be the bipartite graph $G[S_\ell \uplus ([n] \setminus S_\ell)]$ for any fixed $S_\ell$. However, it turns out that this graph is far from being sufficiently well balanced. The reason for this is that \emph{all other} ground-truth cliques $S_{\ell'}$ intersect this bipartition with high probability, which prevents good balancedness properties to emerge because large bicliques actually \emph{do exist}.


Another approach would be to ``flip around'' the bipartition by considering the graph $G[([n] \setminus S_\ell) \uplus S_\ell]$.\footnote{Note that balancedness is not a symmetric property, the graph $G[A \uplus B]$ can have poor balancedness, while $G[B \uplus A]$ has significantly better balancedness properties.} In this case, we get $r$-fold balancedness roughly $k/\sqrt{d}$. While this would be sufficient for ruling out bicliques of size $n^{\varepsilon} \times k/\sqrt{d}$, it is impossible to go below this ``barrier'' of $k/\sqrt{d}$ since $k \times k/\sqrt{d}$ bicliques actually do exist in this bipartition as well.\footnote{This is because $k/\sqrt{d}$ the expected intersection between any two ground-truth cliques, so any $S_{\ell'}$ will actually cause a $k \times k /\sqrt{d}$ sized biclique in $G[[n]\setminus S_\ell \uplus S_\ell]$. } This allows us to conclude that every clique $K$ of size $k$ can only intersect every $S_\ell$ in at most $k/\sqrt{d}$ vertices, however, it does not rule out the case that $K$ is ``scattered'' across many $S_\ell$, so we end up in the same situation we encountered when discussing the combinatorial proof of identifiability in \Cref{sec:inefficientalgo}. 

\vspace{\pspc}
\paragraph{Using balancedness after a neighbourhood reduction}

This suggests that we need a different way of ruling out ``scattered'' cliques in $G$. 
To overcome this challenge, we observe that the main issue with the approaches considered so far is that we only considered bipartite sub-graphs $G[A \uplus B]$ where $A \cup B = [n]$. In this situation, as seen above, it can easily happen that one of the ground-truth cliques intersects our bipartition and thus makes good balancedness properties impossible. 

However, this is very different after applying a \emph{neighbourhood reduction} to $G$. 
That is, given a suitable constant $t \in \mathbb{N}$ and a $t$-tuple $T \in \binom{[n]}{t}$, what can we say about where a clique $K$ can appear \emph{if we force it to contain $T$}? Assuming for now that $T$ is a ``typical''$t$-tuple, then each individual label $\ell \in [d]$ should not appear in too many vertices of $T$. Hence, for a given vertex $v$ outside of $T$, the probability of being adjacent to all $u \in T$ should decay exponentially in $t$, so we can expect all sets $S_\ell$ to shrink by a constant factor (which grows as a function of $t$) when restricting our attention to the neighbourhood $\neigh{T}$. Bearing in mind that all we need is $r$-fold balancedness at most $ck$ for some constant $c < 1$, we can now expect to rule out the existence of sufficiently large cliques in $\neigh{T}$ entirely!  

\vspace{\pspc}
\paragraph{Concentration to few cliques in $\neigh{T}$}
Concretely, we can use the following balancedness-based argument to show identifiability. We define the set $\cliqueset{T}$ as the set of (bad) labels that appear in at least three vertices of $T$, and we wish to prove that only very few vertices in any clique $K$ of size $(1-\varepsilon)k$ that contains $T$ are placed outside of $V(\cliqueset{T})$. By simple concentration bounds, we can show that $|\cliqueset{T}| = O(1)$ w.h.p. for all $T$, so our statement implies that $K$ ``concentrates'' only on constantly many ground-truth cliques with vertex set $V(\cliqueset{T}) = \bigcup_{\ell \in \cliqueset{T}} S_\ell$, as depicted in \Cref{fig:algorithm}. This can afterwards be strengthened into a full proof of identifiability.

To show our ``concentration'' statement, we simply fix a partition $V_1, V_2, \ldots, V_m$ of $[n]$ such that each $V_i$ has size at most $k/2$. A somewhat technical proof then shows that for every $t$-tuple $T$, the graph $G[A_i \uplus B_i]$, where $A_i \coloneqq V_i \setminus V(\cliqueset{T})$ and $B_i = \neigh{T} \cap ( \ [n] \setminus V_i \ )$, has $r$-fold balancedness at most $c k$ for some $c  < \frac{1}{2}$ provided $t$ is sufficiently large. 
 Then, every clique of size $k$ must put mass at least $k/2$ on $[n] \setminus V_i$, because $|V_i| \le k/2$ and placing $\adv{k}$ vertices in $V_i$ would imply a $\adv{k} \times k/2$ biclique in $G[A_i \uplus B_i]$. For any $\adv{k} \ge n^{\Omega(2/r)}$, this can be ruled out using a (modified version) of the BKS-certificate \cite{Kothari-STOC-2023} due to balancedness. Summing over all $V_i$ yields that the number of vertices outside of $V(\cliqueset{T})$ is $O\big(n^{1 + 2/r} / k \big) = o(k)$ whenever $k \gg n^{\frac{1}{2} + \varepsilon}$ and $r$ are chosen sufficiently large (in the order of $1/\varepsilon$).

\paragraph{Balancedness within $V(\cliqueset{T})$ yields concentration to one $S_\ell$.}

After having established concentration to $V(\cliqueset{T})$, we can use balancedness properties among the (constantly many) cliques in $V(\cliqueset{T})$ to ``bootstrap'' this concentration into the significantly stronger statement that $K$ \emph{is entirely contained} in one of the $S_\ell$ for $\ell \in \cliqueset{T}$. Fixing some $\ell \in \cliqueset{T}$ and denoting by $\{x_i\}_{i \in S_\ell}$ the $0/1$-valued variables indicating membership in $S_\ell$, and by $\{y_i\}_{i \in [n] \setminus S_\ell}$ the variables indicating membership in $[n] \setminus S_\ell$ while $|x|_\ell = \sum_{i \in S_\ell} x_i$ and $|y|_\ell = \sum_{i \in [n] \setminus S_\ell} y_i$ is the number of vertices in $S_\ell$ and outside of $S_\ell$, respectively, this is achieved using an SoS-certificate of the form\footnote{See \Cref{sec:sos} and \Cref{lem:slctimproved} for the precise statements and definitions.} 
$$
     \cA(G, (1-\varepsilon)k) \sststile{x,y}{O(t)} \big\{ |x|_\ell^{2t} |y|_{\ell} \le O(k/\sqrt{d})^{2t} \big\}
$$
for any $i \in [n]\setminus S_\ell$. Intuitively, this tells us that $K$ cannot simultaneously intersect $S_\ell$ and $[n] \setminus S_\ell$ by a large amount. Enforcing that $K$ is not a proper subset of any $S_\ell$ amounts to requiring $|y|_\ell \ge 1$ for all $\ell \in \cliqueset{T}$, and under this assumption, the above certificate yields that the total mass on $V(\cliqueset{T})$ is $O(k /\sqrt{d}) = o(k)$, which (combined with the statement from step one) contradicts the fact that $K$ contains $\Omega(k)$ vertices. 

Given the constraints $|y|_\ell \ge 1$ for all $\ell \in [d]$, this entire proof can be phrased as a constant degree \emph{SoS-refutation}, i.e., a derivation of the contradiction $-1 \ge 0$. Given the ground-truth labels, this can directly be used for a polynomial-time refutation algorithm to certify that every clique of size $(1-\varepsilon)k$ is contained in some $S_\ell$.
For recovering $S_\ell$, we clearly cannot use the constraints $|y|_\ell \ge 1$ as input, but the ideas used up to this point are nonetheless useful algorithmically, as will be the focus of the following sections.

\subsection{ Reaching $k \gtrsim \sqrt{n\log(n)} $ }

While the ideas presented so far can be turned into an algorithm for recovery, we can only approach the barrier $k = \sqrt{n}$ up to a factor of $n^\varepsilon$ while using degree-$O(1/\varepsilon)$ SoS-proofs, which makes the ideas used so far less efficient as $\varepsilon \rightarrow 0$. The reason for this is that in step one, we had to rule out $n^{\varepsilon}$-sized cliques in each $V_i$, so summing over $O(n/k)$ many such sets only rules out cliques of size $\Omega(k)$ if $k$ exceeds $\sqrt{n}$ by a factor of $n^\varepsilon$. This is because the \BKSs-inspired certificates are only really effective in bipartite graph of balancedness $\Delta$ if the clique we wish to refute contains more than $\Delta$ vertices on the right, which in our case is $\Omega(k)$. However, a clique of size $k$ could very well put much less than $k$ vertices into $V(\cliqueset{T})$ while ``hiding'' $\Omega(k)$ vertices on the left. The way we solved this issue so far was to simply partition $[n] \setminus V(\cliqueset{t})$ into sufficiently small sets $V_i$ such that $\Omega(k)$ vertices are always forced outside of each $V_i$, which, however, comes at the cost of many $V_i$ and leads to the problem described above.

\paragraph{Proof of identifiability for a fixed bipartition}
If we are only interested in recovering the $S_\ell$ instead of entirely \emph{refuting} the existence of further cliques, there is an easy way to avoid this issue: we simply split the entire graph into two equally sized parts $U$ and $V$, and search for a clique that puts at least $\frac{1-\varepsilon}{2} k$ vertices into both $U$ and $V$. Since for a randomly chosen bipartition $U,V$, this is met for all ground-truth $S_\ell$ w.h.p., this is indeed sufficient for recovery.

A proof of identifiability can then be accomplished using similar ideas as before: take a $t$-tuple $T$ and restrict your attention to the bipartite graph $G[A_i \uplus B_i]$ where now $A_i = U \setminus V(\cliqueset{T})$ and $B_i = \neigh{T} \cap V$ as shown in \Cref{fig:algorithm}.
Again, this graph is sufficiently well balanced to show that any $\frac{1-\varepsilon}{2}k\times \frac{1-\varepsilon}{2}k$-sized biclique in $G[U \uplus V]$ can only intersect $[n]\setminus V(\cliqueset{T})$ in $o(k)$ many vertices. However, since $\Omega(k)$ vertices are forced into both $U$ and $V$, the \BKSs-inspired certificates are now significantly more powerful and work all the way down to $k \gtrsim \sqrt{n\log(n)}$, while degree $\degreeexact$-SoS suffices for all such $k$.

\vspace{\pspc}
\paragraph{Algortihm for exact recovery}

This approach also directly translates into a simple algorithm for exact recovery: for a sufficiently large constant $t$, sample a $t$-tuple $T \in \binom{[n]}{t}$ uniformly at random and compute a degree-$\degreeexact$ pseudo-distribution in $H(T) \coloneqq G[U \uplus V] \cap \neigh{T}$ that respects the biclique axioms $\bicliqueaxioms{\adv{k}}$ (defined in \Cref{sec:axioms}), which force at least $\adv{k} = \frac{1-\varepsilon}{2}k$ vertices into both $U$, and $V$.
Then, for any such pseudo-distribution $\mu$ and sufficiently small $\varepsilon > 0$, construct the set $S_T \coloneqq \{i \in [n] \mid \Expectedtildesub{\mu}{w_i} \ge 1-\varepsilon \}$ as a solution candidate.

We show that for each $S_\ell$ and most $T\subseteq S_\ell$ with $|T| = t$, the set $S_T$ is almost identical to $S_\ell$. Moreover, we show that for \emph{every} $t$-tuple $T \subseteq [n]$, the set $S_T$ is mostly contained in $V(\cliqueset{T})$ by virtue of our proof of identifiability and the duality between SoS-proofs and pseudo-distributions. A combinatorial clean-up step based on the observations in \Cref{sec:cliqueintersections} and a check if $S_T$ is a clique then finally yields the correctness of this approach. A simple argument further shows that this algorithm is robust against a monotone adversary. The details are found in \Cref{sec:analysisexact}.

\subsection{ Identifiability under corruptions and recovery via pseudo-concentration }

A further upshot of our techniques is that they are not limited to recovering cliques, but even yield optimal recovery guarantees under up to $\varepsilon k^2$ edge corruptions whenever $\varepsilon < \varepsilon_0$ where $\varepsilon_0$ is some small constant.\footnote{This means that $\varepsilon$ is actually allowed to depend on $n$, we just require that $\varepsilon < \varepsilon_0$ for all sufficiently large $n$. This allows in particular all $\varepsilon = o(1)$ and all sufficiently small constants.} In this setting, up to a constant fraction of edges within any fixed clique can be missing and the goal is still to output \emph{exactly} $d$ sets of vertices $\mathcal{L} = \{\overline{S}_1, \ldots, \overline{S}_d\}$ such that for every $S_\ell$ there is exactly one $\overline{S} \in \mathcal{L}$ such that $|S_\ell \triangle \overline{S}| \le C\varepsilon k$ for some fixed constant $C$ (that does not depend on $\varepsilon$).

\vspace{\pspc}
\paragraph{A more powerful adversary}

To show this, we actually consider a more general adversarial framework that contains $\varepsilon k^2$ edge corruptions as a special case. Concretely, we allow two types of adversarial modifications. First, there can be up to $\epsedge k$ corruptions incident to \emph{every} vertex, and second, we allow up to $\epsnode k$ \emph{corrupted vertices} whose neighbourhood is \emph{entirely arbitrary}. If we fix some value $\epsedge$, then after $\varepsilon k^2$ edge corruptions, it is easy to see that there are at most $\frac{2\varepsilon}{\epsedge} k$ vertices with more than $\epsedge k$ incident corruptions, so edge corruptions are a special case of this adversarial model for suitable parametrization. 

\vspace{\pspc}
\paragraph{Condition of goodness} To adapt our proof of identifiability to this adversarial setting, we first need to change our condition of ``goodness'' since the clusters we are looking for are not necessarily cliques any longer. Instead, our goal is to recover the \emph{non-corrupted vertices} of all ground-truth communities\footnote{This is the best we can hope for since the neighborhood of corrupted vertices is completely arbitrary.} $S_{\ell}$. Therefore our new condition for goodness is met on a set of vertices if we can ``add back'' up to $\epsedge k$ edges incident to every vertex in order to obtain a clique. To model this in SoS, we introduce an additional $0/1$-variable $z_{u,v}$ for every $u, v$,  which encodes what edges to ``add back'' into our graph. Details are found in \Cref{sec:stronger-adv}.

\vspace{\pspc}
\paragraph{Proof of identifiability}
While the high level approach towards our proof of identifiability remains the same (use a neighbourhood reduction and certify the absence of ``cliques'' in the reduced graph), there are several further challenges to be overcome.

Concretely, a particular problem in light of the ideas underlying our algorithm for exact recovery is that after sampling a random $t$-tuple $T$ and computing a pseudo-distribution purely on the neighbourhood $\neigh{T}$, we have hardly any control over how this distribution behaves due to corruptions incident to $T$. At the very least, we cannot hope to get a similar ``concentration to few cliques'' phenomenon as used in our previous analysis since the neighbourhood of $T$ can be entirely arbitrary. To overcome this issue, it turns out to be important not to enforce the neighbourhood reduction a priori, but rather ``within SoS''. To achieve this, we now compute a pseudo-distribution on \emph{all of} $G$, require enough mass on each side of our bipartition $U, V$, and simply \emph{force the variables corresponding to vertices in $T$} to be one.

This, enables a ``soft'' neighbourhood reduction that allows for a certain mass outside of $\neigh{T}$, and it crucially leads to a pseudo-distribution that factors in more ``global'' information in $G$, which is important for evading the corrupted vertices. The key observation in our analysis is that whenever we have a $t$-tuple $T$ in which no vertex is corrupted, the graph $G[(U \setminus V(\cliqueset{T})) \uplus (V \cap \neigh{T})]$ is still sufficiently well balanced, even if we allow for a certain, adversarial number of edge insertions. To show that this is captured by low-degree SoS, we adapt the \BKSs-inspired certificates to this adversarial setting. Moreover, we exploit the fact that SoS ``knows'' that every sufficiently dense sub-graph $K$ containing $T$ must put $\Omega(k)$ mass on $V \cap \neigh{T}$, which together with the balancedness of $G[(U \setminus V(\cliqueset{T})) \uplus (V \cap \neigh{T})]$ leads to a similar certificate as before for asserting that only few vertices of $K$ are outside of $V(\cliqueset{T})$ (provided not too many corruptions happen incident to $T$). Remarkably, this statement now applies to the whole set $[n] \setminus \cliqueset{T}$ and not just to its restriction to $\neigh{T}$, as before. 

Of course, a concentration statement like this can no longer be true if many corruptions happen incident to $T$. However, in this case we use the fact that our pseudo-distribution has a sufficiently ``global'' view on $G$, so even if forcing $T$ to be contained in our sub-graph $K$, it still ``knows'' that $K$ also has to contain some $t$-tuple which only consists of non-corrupted vertices (even without knowing which vertices are corrupted). For such ``good'' $t$-tuples, a similar concentration phenomenon as described above applies and this yields a proof of identifiability.

\vspace{\pspc}
\paragraph{Algorithm for approximate recovery}
It still remains to exploit the above observations algorithmically. As mentioned above, our algorithm samples $t$-tuples $T$ at random and computes a pseudo-distribution $\mu_T$ in ``relaxed'' clique axioms on $G$ that ensures enough mass on both $U$ and $V$ while adding the constraint $\{ w_T = 1 \}$. Again, we can show that at least a constant fraction of all $t$-tuples in $S_\ell$ are such that this procedure forces concentration on $S_\ell$ implying that we can approximately recover $S_\ell$ by computing $S_T \coloneq \{i \in [n] \mid \Expectedtildesub{\mu_T}{w_i} \ge 1-\smallc \}$ for a sufficiently small constant $\smallc$. All details are found in \Cref{sec:stronger-adv}. The main challenge lies in analysing this procedure for the case that $T$ is not one of these ``good'' tuples in order to show that it does not yield ``false positives'', i.e., that $S_T$ never fails to capture one of the ground-truth $S_\ell$ up to $O(\varepsilon k)$ vertices.

\vspace{\pspc}
\paragraph{Recovery via pseudo-concentration}
To this end, we rely on our proof of identifiability by showing that it implies an \emph{ extremely sharp concentration property} for pseudo-distributions. In particular, if $M$ denotes the set of corrupted vertices, then we can show that \emph{for any} $t$-tuple $T$ (even if $T$ contains corrupted vertices) and sufficiently small $\smallc$, that if $|S_T| = \Omega(k)$, \emph{it must be a subset} of $S_\ell \cup M$ for some $\ell \in [d]$. This implies that every sufficiently large $S_T$ we find at any point during the algorithm is a good approximation for some ground-truth $S_\ell$. Moreover, the guarantee $S_T \in S_\ell \cup M$ is the best we can hope for since the edges incident to $M$ are entirely arbitrary. 

We prove this \emph{pseudo-concentration property} by relying on our proof of identifiability. Concretely, we use the fact that every sufficiently large $S_T$ must always contain some $t$-tuple $R$ in which no vertex is corrupted and in which further $\Expectedtildesub{\mu_T}{w_{R}} \ge \Omega(1)$. Relying on the duality between SoS-proofs and pseudo distributions, this implies that $S_T$ can only put a small number of vertices outside of $V(\cliqueset{R})$, which is a preliminary concentration statement that---in a second step---can be ``bootstrapped'' into the much stronger statement that $S_T$ \emph{must be entirely contained} in $S_\ell \cup M$ for some $\ell \in \cliqueset{R}$. The certificates needed for the analysis of both steps rely on our ``robustified'' version of the \BKSs-certificates based on balancedness properties.

\vspace{\pspc}
\paragraph{Exact recovery under a degree-restricted adversary}

Another interesting consequence of the above analysis is that in case we allow no corrupted vertices (i.e. $M = \emptyset$), but continue to allow an adversary to arbitrarily corrupt up to $\epsedge k$ edges incident to \emph{every} vertex in $[n]$, then our algorithm still achieves \emph{exact recovery} and our pseudo-concentration implies that $S_T$ is always a subset of some $S_\ell$.\footnote{We can further guarantee that some $S_T$ is always \emph{exactly equal} to $S_\ell$ for all $S_\ell$ by computing a pseudo-distribution that maximizes the expectation over all variables $w_i$.} This is interesting not only because it yields a more robust algorithm for exact recovery as described previously, but also because the used techniques might prove useful for finding algorithms in other noisy settings, including a more noisy version of RIGs where not only additional edges outside of the ground-truth cliques can form, but also some edges within the $S_\ell$ might be missing due to noise. We leave studying this case open for future work, but believe that our techniques form a good starting point. We refer to \Cref{sec:discussion} for more details.



\begin{figure}[t]
\centering
\includegraphics[height=0.35\textheight]{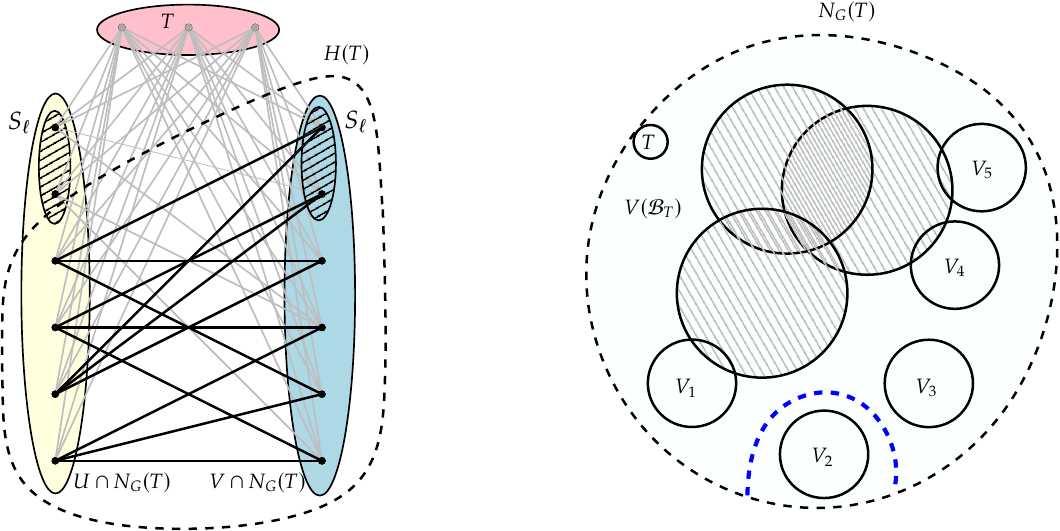}

\caption{Sketch for \Cref{alg:splitting-one-sided}: (left) The Bipartite subgraph $H(T) = G[(U\setminus S_{\ell}) \uplus V] \cap N_G(T)$ is \emph{balanced} shown in \Cref{lem:general-balancedness}. (right) Sketch of step one of the proof of identifiability for our SoS-based proof of \Cref{thm:slct} and our refutation algorithm (\Cref{thm:algorithmicrefutation}). } 
    \label{fig:algorithm}

\end{figure}

\newcommand{\pFK}{p_{\text{FK}}}
\subsection{Comparison to known algorithms for the FK-model}\label{sec:for-reviwer-2}
Finally, we discuss how our model relates to the semi-random setting considered by \cite{Kothari-STOC-2023}, where algorithms robust agains a monotone adversary for the so-called \emph{Feige-Kilian model} (FK-model) are introduced. Here, a clique $S$ of size $k$ is planted in an initially empty graph, while the edges in $S \times [n] \setminus S$ are present independently with probability $\pFK$, and edges in $[n] \setminus S \times [n] \setminus S$ are entirely arbitrary. The goal is to find a list of $(1+o(1))n/k$ cliques that includes $S$. 

One may wonder if our model can be seen as a special case of the FK-model if the edges in $[n] \setminus S \times [n] \setminus S$ are chosen appropriately and $\pFK$ is sufficiently large. While the models are certainly related, the biggest difference to our setting is that in the FK model for $\pFK$ being bounded away from $1$, any clique can intersect the ground-truth $S$ in at most $O(\log(n))$ vertices. In $\RIGone$, the typical overlap between any two $S_\ell, S_{\ell'}$ is $\approx k\delta$, which is of order $n^{\Omega(1)}$ if $k \ge n^{\frac{1}{2} + \varepsilon}$. To potentially be able to map between FK and $\RIGone$, we would therefore need to choose $\pFK$ quite close to one in order to achieve $\delta k \times k$-sized bicliques in $G[S \uplus ([n] \setminus S)]$. Specifically, this would require $(1 - \pFK)^{-1} \gtrsim \delta k /\log(n)$, as a simple calculation reveals. 

However, the algorithms in \cite{Kothari-STOC-2023} require $k \ge n^{\frac{1}{2} + \varepsilon} (1-\pFK)^{-\frac{1}{\varepsilon}}$, while the run time is $n^{O(1/\varepsilon)}$ and $\varepsilon > 0$ is a parameter, see \cite[Theorem 5.6]{Kothari-STOC-2023}. To apply their results to our setting, we would hence need $k \ge n^{\frac{1}{2} + \varepsilon} (\delta k /\log(n))^{\frac{1}{\varepsilon}}$ or equivalently (using the canonical choice $p, p-q =\Omega(1)$ such that $k = \Theta(n/\sqrt{d})$ and $\delta = \Theta(1/\sqrt{d})$), \vspace{-0.1cm}
$$
    1 \gtrsim n^{\varepsilon} \log(n)^{-1/\varepsilon} \left(n/d\right)^{\frac{1}{\varepsilon} - \frac{1}{2}}.
$$\vspace{-0.1cm}
Now,  the restriction $k \gtrsim \sqrt{n \log(n)}$ corresponds to $d \lesssim n/\log(n)$, so the RHS above is at least of order $n^{\varepsilon} \log(n)^{-\frac{1}{2}} = n^{\Omega(1)}$ and the above inequality cannot be satisfied for any constant $\varepsilon > 0$.
This means that whenever we are in the regime where one may potentially represent $\RIGone$ in the FK-model, the algorithms in \cite{Kothari-STOC-2023} will not be applicable. Even if they were, they would require $k \ge n^{\frac{1}{2}+\varepsilon}$, while the run time is $n^{\Omega(1/\varepsilon)}$. Our algorithms overcome these restrictions while seamlessly extending all the way down to $k \gtrsim \sqrt{n\log(n)}$, while run time is $n^{O(1)}$ where the exponent does not grow as a function of $\varepsilon$. Moreover, our techniques extend to the task of refutation (if $k \ge n^{1+\varepsilon}$ ), and our recovery guarantees extend to a stronger adversarial model on top of monotone deletions. 






\section{Preliminaries}

We use $G, H, F$ to denote graphs while $V(G)$ is the set of vertices and $E(G)$ is the set of edges of a given graph $G$. A bipartite graph $G$ with bipartition $A, B$ and set of edges $E$ is denoted as $G = (A \uplus B, E)$. Moreover, given a set $S \subseteq V(G)$, we denote by $G[S]$ the graph induced by the vertices in $S$. The induced graph $G[S]$ is further denoted by $G \cap S$, and these two notations refer to the same object. For sets $A, B \subseteq [n]$, the graph $G[A \uplus B]$ is the \emph{bipartite graph} induced by the bipartition $A, B$, i.e., the graph $G' = (A \uplus B, E')$ with $E' = \{ \{u, v\} \in E(G) \mid u \in A, v\in B \}$. Given a set $T \subseteq V(G)$, we denote by $\neigh{T}$ the neighbourhood of $T$ in $G$, i.e., the set of vertices $v \in V(G) \setminus T$ such that for all $u \in T$, the edge $\{u, v\}$ is part of $G$.

For any $w \in \mathbb{R}^n$ with coordinates $x_i$ for $i \in [n]$, and for any $S \subseteq [n]$, we define $w_S \coloneqq \prod_{i \in S}w_i$. Logarithms are to be understood w.r.t. base $e$ unless specified otherwise. Furthermore the following version of Bernstein's inequality will be indispensable for this work. 
\begin{lemma}[Bernstein's Concentration Bound]\label{lem:bernstein}
    Let $X_1, \ldots, X_n$ be independent, zero-mean random variables such that $|X_i| \le C$ for all $i$, almost surely. Then for every $t > 0$, \begin{align*}
        \Pr{\left|\sum_{i=1}^n X_i\right| \ge t} \le 2\exp \left( \frac{-t^2}{2\left(\sum_{i=1}^n \Expected{X_i^2} + \frac{1}{3}Ct\right)} \right)
    \end{align*}
\end{lemma}

\subsection{Parameter assumptions}\label{sec:parameterassumptions}

An important assumption with regard to the parameters in all $G \sim \RIGone$ we consider is that $d \ge n^{\alpha}$ for some (arbitrarily small) constant $\alpha > 0$, and that $p$ is bounded away from $1$, i.e., that there is some (arbitrarily small) $\varepsilon > 0$ such that $p \le 1 - \varepsilon$. Constants hidden in Landau-notation typically depend linearly on $\alpha$ and $\varepsilon$ and generally get worse as $\alpha, \varepsilon \rightarrow 0$. With regard to the assumption $p \le 1 - \varepsilon$, this is at many places necessary as distinguishing a clique from the rest of the graph generally becomes harder the denser the overall graph becomes. On the other hand, we generally believe that the assumption $\alpha > 0$ is not necessary, as our problem should generally become easier as $k$ gets larger, while the hard regime is reached when $k \ll \sqrt{n}$ (as it is the case for other variants of planted clique as well, see e.g. \cite{Kothari_Vempala_Wein_Xu_2023, Barak_Hopkins_Kelner_Kothari_Moitra_Potechin_2019}).

We assume $\alpha > 0$ (which is also a common assumption in existing literature on RIGs, see e.g. \cite{Spirakis-MFCS-20212}) for the sake of simplicity and readability, since this assumption simplifies some of our proofs. Changing to the case of $\alpha = o(1)$ would require adapting some statements that in our case hold for all $t$-tuples $T \subseteq [n]$ to statements that now only hold for ``most'' $t$-tuples. On the other hand, for small $d$, the number of planted cliques is quite small and $k$ is quite large, which simplifies other steps. While we believe an extension of our analysis to the case of any $d$ or at least sufficiently large $d$ to be possible, we do not expect this to yield significant new insights. 

\subsection{ From dense to sparse graphs } 

Throughout the rest of this paper, we generally make the assumption that the overall edge-density $p$ is \emph{a constant}, while both our algorithms and \Cref{thm:slct} continue to hold if $p = o(1)$ as well. The reason we can still assume $p = \Omega(1)$ is that these \emph{dense} graphs form the hard instances for our problem, and everything we do in the dense regime actually \emph{carries over} directly to the sparse case. 

Concretely, we achieve this because our algorithms are robust to a \emph{monotone adversary}, and since we have a simple coupling argument showing that any RIG with small $p$ can be transformed into an RIG with $p = 1/2$ (or any other constant) simply by ``ramping up'' the noise levels given by $q$ while all ground-truth cliques are kept intact. Our algorithms and theorems then apply to this ``densified'' graph, which can be turned back into the original one using a monotone adversary. This is why it is not a restriction to assume $p \ge 1/2$ throughout the rest of this work. Our coupling is formally captured in the following lemma, which is proved in \Cref{sec:I-have-a-small-p}.

\begin{restatable}{lemma}{CouplingSmallP}\label{lem:couling-subset}
     Consider $G \sim \RIGone$ with $n^\alpha \ll d \ll n^{2-\alpha}$ for some arbitrarily small $\alpha > 0$ and arbitrary $p, q$, $q < p$. Then, for every $\frac{1}{2} \ge  p' > p$, there exists some $q' > q$ such that $G$ can be coupled to $G' \sim RIG(n,d,p',q')$.
     such that $G \subseteq G'$ while for every $\ell \in [d]$, the ground-truth cliques $S_\ell$ are identical in both $G$ and $G'$.
\end{restatable}

\subsection{Comparison with Mixed Membership Stochastic Block Models (MMSBs)}\label{sec:mmsb}

In this section, we provide a comparison between our model and algorithms to the results obtained in previous work on overlapping community detection.

Detecting overlapping communities has received less attention (at least from a rigorous mathematical standpoint), however, most existing work has focused on so called \emph{mixed membership stochastic block models} (MMSBs) that share some similarities with RIGs (see e.g. \cite{coja2010graph,krzakala2013spectral,Anandkumar_Ge_Hsu_Kakade_2013,bordenave2015non,abbe2015community,mossel2015reconstruction,guedon2016community,montanari2016semidefinite,yun2016optimal,Hopkins_Steurer_2017,10.1214/17-AOS1545,Mao02102021} but also the survey by Abbe~\cite{abbe2017community} and references therein). The purpose of this section is to discuss the similarities and differences of these models and to contrast the algorithmic guarantees considered in this and previous works. We start with the definition of MMSBs.

\begin{definition}
   We define $\sbm$ as the distribution over graphs on vertex set $n$ obtained using the following procedure. 
    \begin{enumerate}
\setlength\itemsep{.0001em}
        \item For each $v \in [n]$, independently draw a $d$-dimensional vector $\sigma_v \sim \dirichlet$ where $\dirichlet$ is the symmetric $d$-dimensional Dirichlet distribution with parameter $\beta$.
        \item For every pair $u, v \in [n], u \neq v$, draw the edge $\{u, v\}$ with probability $\frac{s}{n} \cdot \left( 1 + (\langle \sigma_u, \sigma_v \rangle - \frac{1}{d})\varepsilon \right)$.
    \end{enumerate}
\end{definition}
Here, the vectors $\sigma_v$ represent the \emph{fractional} community memberships of $v$. The probability that two vertices $u, v$ are adjacent is larger if the two vectors $\sigma_i, \sigma_j$ share similar features. In total, this gives rise to $d$ communities where community $\ell \in [d]$ is defined by all vertices $v$ for which $\sigma_v(\ell)$ is  ``large''. The parameter $\beta$ controls the amount of overlap between communities. Specifically, in the limit $\beta \rightarrow 0$, the vectors $\sigma_v$ will only be supported on a single coordinate giving rise to \emph{disjoint} communities like in the classical stochastic block model. On the other hand, for non-zero $\beta$ the  ``large'' values of $\sigma_v$ concentrate on roughly $\beta$ uniformly random coordinates of each $\sigma_v$. This has the intuitive interpretation that each vertex is a member of approximately $\beta$ communities.

\paragraph{Stronger notions of recovery}
RIGs are conceptually similar if we think of $\{0, 1\}^d$-valued community membership vectors $\sigma_v$ that arise as the \emph{characteristic vectors} of the sets $M_v$. Two vertices are then adjacent if and only if $\langle \sigma_u, \sigma_v\rangle \neq 0$ (in the noiseless case). Hence, the main difference between RIGs and MMSBs is that RIGs do not allow for fractional community memberships.
This allows a much less ambiguous and combinatorial definition of which vertices participate in a given community and thus enables us to study stronger notions of recovery (formalized in \Cref{def:recovery}): we require our algorithms to find exactly $d$ communities such that for \emph{every} ground-truth community $K$ we recover almost all vertices that participate in $K$. This is a key difference to previous work \cite{Anandkumar_Ge_Hsu_Kakade_2013, Hopkins_Steurer_2017}, where the goal is \emph{weak recovery} that merely requires finding a community assignment which is non-trivially correlated with the ground-truth. Here, a community assignment can have high correlation with the ground-truth even if entire communities are completely ``ignored''. This is especially true if the number of communities tends to infinity with $n$. Our algorithms achieve performance guarantees that explicitly exclude such outcomes.

\paragraph{Limitations of tensor-based approaches in RIGs}
Another important difference between MMSBs and RIGs is that RIGs do not make a distinction between edges that appear in only one community versus edges that appear in two or more communities. Specifically, if we consider a set of vertices $S$, then in a MMSB, the graph induced by $S$ will become denser the more shared communities the vertices in $S$ participate in. In RIGs on the other hand, there is no distinction between the vertices in $S$ having one versus many communities in common. Mathematically, this amounts to a certain ``non-linearity'' when estimating moments of the hidden variables given the posterior (i.e. the graph arising from them). 

To understand why this matters, it is important to point out the key idea underlying the algorithms in \cite{Anandkumar_Ge_Hsu_Kakade_2013} and \cite{Hopkins_Steurer_2017}. Here, the high-level approach is to evaluate a suitable low-degree polynomial on $G$ that serves as an estimate for the moments of the hidden variables, which is afterwards decomposed into a suitable instantiation of the hidden variables. Concretely, if we define centered community indicator vectors $y_1, \ldots, y_d \in \mathbb{R}^n$ with $y_\ell(v) = (\sigma_v(\ell) - \frac{1}{d})$, then in an MMSB, a centered \emph{3-star count} $\sum_{v \in [n]}(\mathds{1}_{i, v} - \frac{s}{n})(\mathds{1}_{j, v} - \frac{s}{n})(\mathds{1}_{k, v} - \frac{s}{n})$ serves as an \emph{unbiased estimator} for $\sum_{\ell \in [d]} y_\ell(i)y_\ell(j)y_\ell(k)$. Doing this for all $i, j, k \in [n]$, this means that we can estimate the tensor $T = \sum_{\ell \in [d]} y_\ell^{\otimes 3}$\footnote{in order to reduce the variance, this is in fact done for length-$O(\log(n))$-armed 3-stars. However, the estimated tensor remains $\sum_{\ell \in [d]}y_\ell^{\otimes 3}$. }. Now, using a similar estimator for the second moment matrix, this can be approximately orthogonalized after which a sufficiently strong tensor decomposition algorithm can be used to estimate the vectors $y_\ell$.

Concretely, in $\sbm$, we have
\begin{align*}
    &\Expectedsub{\sigma_v}{ \left. \prod_{u \in \{i, j, k\}}\left(\mathds{1}_{u, v} - \frac{s}{n}\right) \right| \sigma_i, \sigma_j, \sigma_k } = \left( \frac{s\varepsilon}{n} \right)^3 \Expectedsub{\sigma_v}{  \prod_{u \in \{i, j, k\}}\left(\langle \sigma_v, \sigma_u \rangle - \frac{1}{d}\right) }\\
    &\hspace{4cm}= \left( \frac{s\varepsilon}{n} \right)^3 \Expectedsub{\sigma_v}{  \prod_{u \in \{i, j, k\}}\left\langle \sigma_v - \frac{1}{d}\mathds{1}, \sigma_u  - \frac{1}{d}\mathds{1} \right\rangle} = C \left( \frac{s\varepsilon}{n} \right)^3 \sum_{\ell \in [d]} y_\ell(i)y_\ell(j)y_\ell(k),
\end{align*} 
so the estimator is indeed unbiased. The reason this works is that the model is \emph{linear}, i.e., the probability that two vertices are adjacent scales linearly with the inner product $\langle \sigma_v, \sigma_u \rangle$. This is crucially not the case in $\RIGone$, where instead of making $u, v$ adjacent with probability $p + (q-p)\langle\sigma_v, \sigma_u \rangle$, we use something like $p + (q-p)\mathds{1}(\sigma_v, \sigma_u)$ (if $\sigma_u, \sigma_v$ now denote the characteristic vectors of $M_v$, $M_u$, respectively). This replacement of the inner product by an indicator function destroys the unbiasedness of the the 3-star (and more generally the length-$\ell$-armed 3-star), so the quantity being estimated no longer nicely resembles the moments of our hidden variables.

\paragraph{Handling many communities and achieving robustness}

In addition to that, even if the algorithms from \cite{Anandkumar_Ge_Hsu_Kakade_2013, Hopkins_Steurer_2017} were applicable in our setting, they suffer from several other drawbacks. Specifically, besides the fact that they only achieve weak recovery, they have trouble scaling to the so-called ``high-dimensional'' case where the number $d$ of communities scales polynomially in $n$ (i.e. $d \ge n^{\Omega(1)}$). Specifically, the analysis in \cite{Hopkins_Steurer_2017} is only conducted under the restriction that $d \le n^{o(1)}$, while the algortithm of Anandkumar et al. \cite{Anandkumar_Ge_Hsu_Kakade_2013} needs the assumption $n \gg d^2 (\beta + 1)^2$. Considering for example our model $\RIGone$ with constant $p$ and $p - q = \Omega(1)$ (the dense case), which corresponds most closely to $\sbm$ with $s, \varepsilon= \Theta(n)$ and $\beta = \Theta(\sqrt{d})$, this translates roughly to $d \ll n^{1/3}$, which is much much more restrictive than the guarantees of our \Cref{thm:exactrecovery,thm:approxrecovery} that work for any $d \ll n/\log(n)$. 


In addition to that, our algorithms are provably robust against both monotone and bounded adversaries, which previous algorithms based on polynomial estimators and tensor decompositions cannot easily guarantee, even if they are a poly-logarithmic factor away from the generalized Kesten--Stigum (KS) threshold. 
While in the case of disjoint communities, robust algorithms achieving weak recovery that use a similar high-level approach (low-degree estimators + decomposition algorithm) exist \cite{Ding_DOrsi_Nasser_Steurer_2022, Mohanty_Raghavendra_Wu_2024}, the analysis of these algorithms only allows for an adversary that can modify a sufficiently small constant fraction of edges, which on its own is too little for handling a monotone adversary. Our algorithms rely on different techniques which are robust against monotone adversaries, while even withstanding optimal\footnote{optimal if we require our stronger notions of approximate recovery from \Cref{def:recovery}} bounded adversaries on top. 

\paragraph{Refutation and certifiable absence of cliques}

Another advantage of our technique over the approach in \cite{Anandkumar_Ge_Hsu_Kakade_2013, Hopkins_Steurer_2017} is that it uses the proofs-to-algorithms framework and thus relies on the \emph{certification} of the absence of cliques in certain subgraphs of RIGs. This has the consequence that our algorithms not only recover the community structure of a RIG, but they also \emph{prove} that there exists no further clique of size at least $\varepsilon k$ that is not contained in one of the recovered communities. Previous, tensor-based algorithms are much more heuristic in nature and---while they might succeed at (weakly) finding communities---the used techniques cannot refute the existence of further cliques.

\section{Combinatorial Properties of Noisy RIGs}

\subsection{Parameters and their relationship}\label{sec:parametersandtheirrelationship}

We start by introducing some auxiliary lemmas describing the size and relationship of the parameters used in our models.
\begin{lemma}\label{lem:sizeofdelta}
    For a graph sampled from $\RIGone$ with $d = \omega(1)$, we have that $$\delta = \sqrt{\log\left(\frac{1-q}{1-p}\right)\frac{1}{d}}\left( 1 \pm O\left( \frac{1}{d} \right) \right).$$
\end{lemma}
\begin{proof}
    By definition, we have that \begin{align}\label{eq:delta}
        p = 1 - (1-q)(1 - \delta^2)^{d}, \text{ so } \delta = \sqrt{1 - \exp\left( -\log\left(\frac{1-q}{1-p}\right)\frac{1}{d} \right)}.
    \end{align}
    Moreover, by a Taylor series expansion, we get that \begin{align*}
        \exp\left( -\log\left(\frac{1-q}{1-p}\right)\frac{1}{d} \right) &= 1 - \log\left(\frac{1-q}{1-p}\right)\frac{1}{d} \pm O\left(\frac{1}{d^2}\right)\\
        &=  1 - \left(\log\left(\frac{1-q}{1-p}\right)\frac{1}{d} \right)\left( 1 \pm O\left(\frac{1}{d}\right)\right).
    \end{align*}
    The lemma follows after noting that $\sqrt{ 1 \pm O\left(\frac{1}{d}\right)} = 1 \pm O\left(\frac{1}{d}\right)$.
\end{proof}

We further use the following lemma stating that the probability that two vertices are adjacent remains close to $p$ if the labels of one vertex are fixed and their number if sufficiently close to its expectation $\delta d$.
\begin{lemma}\label{lem:pconcentration}
        Assume that  $u, v$ are vertices with fixed $M_u$ such that $|M_u| = \delta d \pm a$ and $a = o(\delta d)$. Then, over the randomness of $M_v$, $\left|\Pr{u \sim v} - p\right| = O(\delta a).$
\end{lemma}\begin{proof}
        Using that $|M_u| = \delta d \pm a$, the probability that $u \sim v$ is 
        $$
            \Pr{u \sim v} = 1 - (1-q)(1-\delta)^{\delta d \pm a}, \text{ while } p = 1 - (1-q)(1-\delta^2)^{d}.
        $$
        To compare these two terms, we use a Taylor series and $\delta = o(1)$ to get that 
        $
            \exp\left(-\delta\right) = 1 - \delta \pm O(\delta^2) = (1 - \delta)(1 \pm O(\delta^2)) \text{, so in fact } 1 - \delta = \exp(-\delta)(1 \pm O(\delta^2)). 
        $
        Hence,
        \begin{align*}
            (1-\delta)^{\delta d \pm a} &= \exp\big(-\delta^2d \big)\exp\big(\pm \delta a\big)\big(1 \pm O(\delta^2)\big)^{\delta d \pm a}.
        \end{align*} We handle the three factors separately. First, note that by \Cref{lem:sizeofdelta}, 
    \begin{align*}
             \exp(-\delta^2 d) = \exp\left(-\log\left( \frac{1-q}{1-p} \right) \left( 1 \pm O\left(\frac{1}{d}\right)\right)\right) = \frac{1-p}{1-q}\left( 1 \pm O\left(\frac{1}{d}\right)\right).
        \end{align*}    
        Moreover, $\delta a = o(1)$, so 
        $$
            \exp\big(\pm \delta a\big) = \big( 1 \pm O(\delta a)\big).
        $$
        Finally, since $\delta d \pm a = \Theta(\delta d)$, 
        $$
            (1 \pm O(\delta^2))^{\delta d \pm a} = (1 \pm O(\delta^3d)) = \left(1 \pm O\left(\frac{1}{\sqrt{d}}\right)\right).
        $$
        Combining all the error terms then yields the lemma.
\end{proof}
To apply the above, we frequently use the following lemma that yields concentration for the number of labels chosen by any $v \in [n]$.

\begin{lemma}\label{lem:labelconcentration}
    Given a subset $M \subseteq [n]$, after drawing the labels in $M$, it holds for any vertex $v \in [n]$ with probability at least $1 - n^{-\omega(1)}$ that $|M_v| = \delta d \pm \log(n)\sqrt{\delta d}$.
\end{lemma}
\begin{proof}
To bound the number of labels per vertex, we note that the expected number of labels is $\delta d$, and the probability that the number of labels differs from this expectation by at least $t = \log(n) \sqrt{\delta d}$ is at most \begin{align*}
            \Pr{|M_v - \delta d| \ge t} \le 2 \exp\left( -\frac{t^2}{2 d\delta + \frac{2}{3}t} \right) \le 2 \exp\left( -\frac{\log^2(n) \delta d}{4 d\delta } \right) \le n^{-\omega(1)}
        \end{align*} if $n$ is large enough. A union bound over all $d \le n$ labels and $|M| \le n$ vertices shows that for all $v \in M$, we have $|M_v| = \delta d \pm \log(n)\sqrt{\delta d}$ with probability $\ge 1 - n^{-\omega(1)}$. 
\end{proof}

\subsection{Balancedness of bipartite subgraphs}\label{sec:balancedness}

A crucial property of dense RIGs that our algorithms rely on is that they contain many cuts with the following \emph{balancedness} property. It resembles \cite[Definition 4.4]{Kothari-STOC-2023} with the only difference being that we have two parameters $s$ and $r$ where $r$ makes a statement about the actual quality of the balancedness (like in \cite[Definition 4.4]{Kothari-STOC-2023} ) and $s$ controls the size of the left-sided set, which we need additionally since we do not get the same balancedness criteria for all $s$.

\begin{definition}[Balancedness]\label{def:balancedness}
    Let $H = (A \uplus B, E)$ be a bipartite graph. Given a set $S \subseteq A$ and a vertex $v \in B$, define \begin{align*}
        u_{S, p}(v) \coloneqq \prod_{u \in S} H_p(u, v) \text{ and } H_p(u, v) \coloneqq \begin{cases}
            \sqrt{\frac{1-p}{p}} &\text{if } \{u, v\} \in E(H)\\
            -\sqrt{\frac{p}{1-p}} &\text{otherwise.}
        \end{cases}
    \end{align*} We say that $H$ has \emph{$r$-fold balancedness} $\Delta$ if for all $S, R \subseteq A$ with $|S| = |R| = r$ and $|S \triangle R| \ge 3$, we have $\left|\sum_{v \in B} u_{S,p}(v)u_{R,p}(v)\right| \le \Delta$. 
\end{definition}
In the simplest case of $p = \frac{1}{2}$, $H_p(u, v)$ and $u_{S,p}(v)$ take values in $\{+1, -1\}$ and $H$ is balanced if the alternating sum $\sum_{v \in B} u_{S,p}(v)$ is sufficiently close to $0$ for all small $S \subseteq A$ (see also \Cref{fig:algorithm}b).

We introduce multiple criteria under which a bipartite subgraph of $G$ is balanced. All of them are in their core based on the following \Cref{lem:balancedness-core} stating how balancedness arises if the left side labels are fixed and meet certain criteria, while most of the right-sided vertices are yet to be uncovered. Before introducing said lemma, we must formalize some of the criteria the left sided labels must meet.

\begin{definition}\label{def:event-phase4}
    Given set $A \subseteq [n]$,
    let $\mathcal{E}(A, t, b)$ be the event that the following occurs.
    \begin{enumerate}
        \item For all $v \in A$, we have that $|M_v| \le \delta d \pm O(\log(n)\sqrt{\delta d})$.
        \item For all $2 \le r \le t$ and all $S, R \subseteq A$ with $|S| = |R| = r$, the number of labels $\ell \in [d]$ such that at least $2$ vertices in $S \cup R$ have chosen $\ell$ is at most $Ct\log(n)$ where $C > 0$ is an absolute constant. For any set $Q \subseteq A$, we call such labels\footnote{i.e. the labels appearing in at least two vertices in $Q$} ``duplicates'' and denote the set of such labels by $\mathcal{D}_{Q}$.
        \item For all $2 \le r \le t$ and all $S, R \subseteq A$ with $|S| = |R| = r$, the number of labels $\ell \in [d]$ such that at least $3$ vertices in $S \cup R$ have chosen $\ell$ is at most $b$. For any set $Q \subseteq A$, we call such labels\footnote{i.e. the labels appearing in at least three vertices in $Q$} ``bad labels'' and denote the set of such labels by $\mathcal{B}_{Q}$.
    \end{enumerate}
\end{definition}

To show that the above is likely, we start with the following lemma bounding the number of bad labels.
\begin{lemma}\label{lem:boubdbadlabels}
    Given any fixed $t \in \mathbb{N}$, then for any $b \ge \frac{4t}{\alpha}$, the probability that there are at least $b$ bad labels in any fixed $T \in \binom{[n]}{t}$ (i.e. labels that appear in at least three vertices in $T$) is $O(n^{-2t})$.
\end{lemma}
\begin{proof}
    Note that for fixed $T$, and a fixed label $\ell \in [d]$, the probability that at least three vertices in $S \cup R$ choose $\ell$ is at most 
        $$
            \binom{|T|}{3} \delta^3 = O\left( \left( \frac{1}{\sqrt{d}} \right)^3 \right).
        $$
        Hence, the probability that this happens for at least $a \in \mathbb{N}$ labels at the same time is at most 
        $$
            \sum_{i = a}^d \binom{d}{i} O\left(\frac{1}{\sqrt{d}}\right)^{3i} \le \sum_{i = a}^d O\left(\frac{1}{\sqrt{d}}\right)^{i} \le O\left(\frac{1}{\sqrt{d}}\right)^{a}.
        $$
        Since $d \ge n^{\alpha}$, the probability that $T$ contains at least $b$ bad labels is at most 
        $$
            O\left(\left(\frac{1}{\sqrt{d}}\right)^{b} \right) \le O \big( n^{- \frac{b\alpha}{2}} \big) \le O( n^{-2t} )
        $$ if $b \ge \frac{4t}{\alpha}$.
\end{proof}

It now follows that $\mathcal{E}(A, t, \frac{8t}{\alpha})$ occurs with high probability when choosing the left sided labels.
\begin{lemma}\label{lem:E-is-likely}
    Given a fixed set $A \subseteq [n]$ and a fixed constant $t$, the event $\mathcal{E}(A, t, \frac{8t}{\alpha})$ occurs with probability $1 - O(n^{-2t})$ over the draw of labels in $A$. 
\end{lemma}
\begin{proof}
 \textbf{Number of labels} 
Using \Cref{lem:labelconcentration} it immediately follows that $|M_v| = \delta d \pm \log(n)\sqrt{\delta d}$ with probability $1 -n^{-\omega(1)}$.

        \textbf{Duplicate labels}
        Given fixed $S, R \subseteq A$ with $|S| = |R| = r$, what is the probability that a label $\ell \in [d]$ is a duplicate in $S \cup R$? Certainly, using a union bound, at most $|S \cup R|^2 \delta^2 \le 4t^2\delta^2$, which is $O(t^2/d)$ by \Cref{lem:sizeofdelta}. Hence, the expected number of duplicates in $S$ is $O(t^2)$ and by an application of Bernstein's inequality (\Cref{lem:bernstein}), the probability that there are more than $\Expected{|\mathcal{D}_{S \cup R}|} + z$ duplicate labels in $S \cup R$ is at most \begin{align*}
            \Pr{|\mathcal{D}_{S \cup R}| \ge \Expected{|\mathcal{D}_{S \cup R}|} + z} \le 2\exp\left( \frac{-z^2}{8t^2\delta^2d + \frac{2}{3}z} \right).
        \end{align*} Setting $z = Ct\log(n)$ for a suitable constant $C > 0$ yields that the above is at most $n^{-4t}$. A union bound over all $S, R \subseteq A$ with $|S| = |R| \le t$ then yields that the number of duplicate labels is bounded by $Ct\log(n)$ for all such pairs $S,  R$ with probability $1 - n^{-2t}$.

        \textbf{Bad labels}
        Note that for any fixed $S, R$, we get from \Cref{lem:boubdbadlabels} and a union bound that the probability that there are at least $b \ge \frac{8t}{\alpha}$ bad labels in some $S \cup R$ is at most \begin{align*}
            O\left(n^{2t} \left(\frac{1}{\sqrt{d}}\right)^{b} \right) \le O \left( n^{2t - \frac{b\alpha}{2}} \right) \le O( n^{-2t} ).
        \end{align*}
\end{proof}

Given this, our core balancedness lemma is stated as follows.
\begin{lemma}[Balancedness of neighbourhood-restricted bipartite subgraphs with fixed labels on the left]\label{lem:balancedness-core}
    Fix three integers $r, t, \eta \in \mathbb{N}, r\le t$.
    Assume we are given disjoint subsets $A, B, T$ of $[n]$ such that $|T| = t$. We further make the following assumptions:
    \begin{enumerate}
    \setlength\itemsep{.01em}  
        \item All labels in $A$ and $T$ have been revealed, while for vertices in $B$, only a subset $L \subseteq [d]$ of labels has been revealed, with $|L| = O(1)$.
        \item No vertex in $A$ has chosen a label in $L$.
        \item For every $\ell \in [d] \setminus L$, there are at most $\eta$ vertices in $T$ that have chosen $\ell$.
        \item The labels in $A$ and $T$ are such that $\mathcal{E}(A, t, b)$ and $\mathcal{E}(T, t, b)$ are met, respectively.
        \item For every $S, R \subseteq A$ with $|S| = |R| \le r$, the number of duplicates in $ S \cup R \cup T$ is at most $2\log(n)^2$.
        \item For every $\ell \in L$, the number of vertices in $B$ that have chosen $\ell$ is at most $2k$.
    \end{enumerate}
    Then, after revealing the missing labels in $B$, the graph $H = G[ A \uplus ( B \cap N_G(T) ) ] $ has $r$-fold balancedness 
    $$
    (1+o(1))b \maxbal{r} \cdot  p^{t-\eta}k + o(k)
    $$
    for every $3 \le r \le t$ with probability $1 - n^{-2t}$.
\end{lemma} 
\begin{proof}
    Fix an arbitrary pair $S, R \subseteq A$ with $|S| = |R| = r$. We wish to bound $\USR \coloneqq \sum_{v\in B} \mathds{1}(v \in N_G(T)) \cdot u_{S,p}(v)u_{R,p}(v)$. At this point, recall from \Cref{def:event-phase4} that $\mathcal{B}_{S \cup R}$ denotes the set of ``bad'' labels, i.e., those that appear in at least $3$ vertices in $S \cup R$. With this in mind, we split $\USR$ into the contribution of vertices in $B$ that have chosen a bad label and those that have not. 
    \begin{align*}
        \USR = \underbrace{\sum_{\substack{v \in B \\ M_v \cap \mathcal{B}_{S \cup R} \neq \emptyset}} \mathds{1}(v \in N_G(T)) u_{S,p}(v)u_{R,p}(v) }_{\eqqcolon \USRbad} + \underbrace{\sum_{\substack{v \in B \\ M_v \cap \mathcal{B}_{S \cup R} = \emptyset}} \mathds{1}(v \in N_G(T)) u_{S,p}(v)u_{R,p}(v) }_{\eqqcolon \USRgood}.
    \end{align*}
    Now, we show separately that both sums are small with high probability.

    \paragraph{Estimating $\USRbad$.}
    For the first sum, we use that the number of bad labels in $T$ is small due to the occurrence of $\mathcal{E}(T, t, b)$ combined with the fact that for every label $\ell \in [d] \setminus L$, the number of vertices that have chosen $\ell$ while still being in $N_G(T)$ drops exponentially in $t = |T|$. Precisely, using that $|u_{S,p}(v)u_{R,p}(v)| \le \maxbal{r}$ by definition of $u_{S,p}, u_{R,p}$ (\Cref{def:balancedness}), we bound \begin{align*}
        |\USRbad| \le \maxbal{r} \sum_{\ell \in \mathcal{B}_{S \cup R}} | S_{\ell} \cap N_G(T) |
    \end{align*} and proceed by estimating $| S_{\ell} \cap N_G(T) |$ given a fixed $\ell \in \mathcal{B}_{S \cup R}$. To this end, we consider the probability that a given $v \in B$ is in $S_{\ell} \cap N_G(T)$. 
    
    To this end, denote by $B_L$ the set of all $v \in B$ that have chosen a label in $L$, by $p_L$ a bound on the probability that a $v \in B_L$ is in $S_{\ell} \cap N_G(T)$, and by $p_{\overline{L}}$ a bound on the probability that a $v \in B \setminus B_L$ is in $S_{\ell} \cap N_G(T)$. Then,
    $$
        \Expected{ | S_{\ell} \cap N_G(T) | } \le |B_L| p_L + |B \setminus B_L|  p_{\overline{L}}.
    $$ 
    To estimate the first term, we use $p_L \le \Pr{v \in S_\ell} \le \delta$, and $|B_L| \le 2|L|k$ (which holds by assumption 6). Thus, $\Expected{ | S_{\ell} \cap N_G(T) | } \le 2\delta k |L| + |B \setminus B_L|  p_{\overline{L}}$, so we can focus on estimating $p_{\overline{L}}$.
    
    To this end, observe that for any $v \in B \setminus B_L$, we have $\Prnop{v \in S_{\ell} \cap N_G(T)} = \delta \Prnop{v \in N_G(T) \mid v \in S_{\ell}}$. When conditioning on $v \in S_{\ell}$, this determines at most $\eta$ edges between $v$ and $T$ because of assumption 3. To estimate the probability that the remaining edges between $v$ and $T$ are also present, we use the fact that the number of duplicate labels $\mathcal{D}_T$ is small so it is unlikely for $v$ to choose such a label and the contribution of this event to $p_{\overline{L}}$ is negligible. If no label in $\mathcal{D}_T$ is chosen, then (over the remaining randomness) the remaining edges between $v$ and $T$ are independent since each label in $[d] \setminus \mathcal{D}_T$ is chosen by at most one vertex in $T$.
    
    Formally, we distinguish the following two cases. In the first case, $v$ choses a label in $\mathcal{D}_T$. By the bound on $|\mathcal{D}_T|$ due to the occurrence of $\mathcal{E}(A, t, b)$, this happens with probability at most $Ct \log(n) \delta$. In the second case, no label in $\mathcal{D}_T$ is chosen, and thus (using \Cref{lem:pconcentration}), the probability that $v \sim T$ is at most $\left(p\bigl(1 + O\bigl(\log(n) d^{-1/4} \bigr) \bigr)\right)^{t - \eta}$. Both cases together yield the bound
    $$
        \Prnop{v \in N_G(T) \mid v \in S_{\ell}} \le Ct\log(n)\delta + \Bigl(p\bigl(1 + O\bigl(\log(n) d^{-1/4} \bigr) \bigr)\Bigr)^{t - \eta} \le (1+o(1)) p^{t-\eta} + o(1).
    $$
    In total (recalling that $k = n\delta$), this yields that 
    $$
       \Expected{ | S_{\ell} \cap N_G(T) | } \le 2\delta k |L| +  n\delta \bigl( (1+o(1)) p^{t-\eta} + o(1) \bigr) \le (1+o(1)) p^{t-\eta}k + o(k).
    $$
    Further, we get from a Chernoff/Bernstein bound and a subsequent union bound over all $\ell \in [d] \setminus L$ that $|S_{\ell} \cap N_G(T)| \le (1+o(1)) p^{t-\eta} k + o(k)$ holds for all such labels with probability $1 - n^{-\omega(1)}$. 
    
    Finally, since $|\mathcal{B}_{S\cup R}| \le b$ by the occurrence of $\mathcal{E}(A, t, b)$, this implies that \begin{align*}
        |\USRbad| &\le \maxbal{r} \sum_{\ell \in \mathcal{B}_{S \cup R}} | S_{\ell} \cap N_G(T) |\\ &\le (1+o(1))b \maxbal{r} \cdot  p^{t-\eta}k + o(k)
    \end{align*}
    with probability $1 - n^{-\omega(1)}$. Using a union bound, the above holds for all $S, R \subseteq A$ with $|S| = |R| = r$ with high probability. 

    \paragraph{Estimating $\USRgood$}
    To estimate $\USRgood$, we afford a union bound over all $S, R \subseteq A$ with $|S| = |R| = r$. To this end, first consider an arbitrary fixed pair of sets $S, R \subseteq A$ with $|S| = |R| = r$. We note that 
    \begin{align*}
        \Expectednop{\USRgood} = \sum_{v \in B}  \Expected{\mathds{1} (M_v \cap \mathcal{B}_{S \cup R} = \emptyset) \cdot \mathds{1}(v \in N_G(T)) \cdot u_{S,p}(v)u_{R,p}(v)}.
    \end{align*}
    To estimate the above expectation, we fix a $v \in B$ and condition on the event that $M_v \cap \mathcal{B}_{S \cup R} = \emptyset$ since the case $M_v \cap \mathcal{B}_{S \cup R} \neq \emptyset$ is covered by our bounds on $\USRbad$. In this conditional probability space $\Omega$, we first draw the labels in $\mathcal{D}_{S \cup R \cup T}$ for all vertices in $B$, and afterwards draw all the remaining labels for vertices in $B$. Regarding the first step, denote by $\mathcal{E}_{(0)}, \mathcal{E}_{(1)}$, and $\mathcal{E}_{(> 1)}$ the event that $v$ draws zero, exactly one, or strictly more than one label from $\mathcal{D}_{S \cup R \cup T}$, respectively. Using the law of total expectation, we bound \begin{align*}
        \Expectedsubnop{\Omega}{\mathds{1}(v \in N_G(T)) \cdot u_{S,p}(v)u_{R,p}(v) } &\le \maxbal{r} \Prnop{\mathcal{E}_{(>1)}} \\
        & \hspace{1cm}+ \Prnop{\mathcal{E}_{(1)}}\Expectedsubnop{\Omega}{ u_{S,p}(v)u_{R,p}(v) \mid \mathcal{E}_{(1)} \cap \{ v \in N_G(T) \} }\\
        &\hspace{1cm}+ \Expectedsubnop{\Omega}{ u_{S,p}(v)u_{R,p}(v) \mid \mathcal{E}_{(0)} \cap \{ v \in N_G(T) \} }.
    \end{align*} Since $|\mathcal{D}_{S \cup R \cup T}| \le 2\log(n)^2$ by assumption 4, we get that 
    $$
        \Prnop{\mathcal{E}_{(>1)}}  \le \binom{|\mathcal{D}_{S \cup R \cup T}|}{2} \delta^2 \le O\left( \log(n)^4 \delta^2 \right)  \text{ and } \Prnop{\mathcal{E}_{(1)}} = O\left( \log(n)^2\delta \right).
    $$
    Since $n\Prnop{\mathcal{E}_{(>1)}} = o(k)$, that is, smaller than our desired $\Delta_s$, we can essentially restrict our attention to the other two cases and assume that at most one label in $\mathcal{D}_{S \cup R \cup T}$ and no label in $\mathcal{B}_{S \cup R}$ was chosen by $v$. By construction, all labels not in $\mathcal{D}_{S \cup R \cup T}$ are chosen by at most one vertex in $S \cup R \cup T$, so after revealing the labels in $\mathcal{D}_{S \cup R \cup T}$, the random variables $\mathds{1}(v \in N_G(T))$ and $u_{S,p}(v)u_{R,p}(v)$ are independent, and every edge between $v$ and $S \cup T$ appears independently. By \Cref{lem:pconcentration}, for each such edge $e = \{u, v\}$, this happens with probability $p \pm \gamma$, for some $\gamma = O(\log(n) \delta \sqrt{\delta d})$. Hence, $$
        \Expected{H_p(u, v)} \le \sqrt{\frac{p}{1-p}} (1 - p +\gamma) - \sqrt{\frac{1 - p}{p}}(p - \gamma) \le 2\gamma \sqrt{\maxbal{}}.
    $$ Now, assume that $\mathcal{E}_{(0)}$ occurs after drawing the labels in $\mathcal{D}_{S \cup R \cup T}$. Since $|S \triangle R| \ge 3$, there are at least three edges between $S \triangle R$ and $v$, whose contribution to $\Expectedsubnop{\Omega}{ u_{S,p}(v)u_{R,p}(v) \mid \mathcal{E}_{(0)} \cap \{ v \in N_G(T) \} }$ is $O\bigl(\log(n) \delta\sqrt{\delta d} \bigr)^3 \sqrt{\maxbalcomp{3}}$. All remaining edges have a contribution of at most $\sqrt{\maxbalcomp{2r - 3}}$. In total, 
    \begin{align*}
        &\Expectedsubnop{\Omega}{ u_{S,p}(v)u_{R,p}(v) \mid \mathcal{E}_{(0)} \cap \{ v \in N_G(T) \} } \\&\hspace{3cm} \le O\bigl(\log(n) \delta\sqrt{\delta d} \bigr)^3 \sqrt{\maxbal{3}} \sqrt{\maxbal{2r - 3}} 
        \\&\hspace{3cm}= O \bigl( \log(n)^3 \delta^3 (\delta d)^{3/2} \bigr) \maxbal{r}.
    \end{align*}
    Moreover, note that if $\mathcal{E}_{(1)}$ occurs, this determines at most two edges between $v$ and $S \cup R$. Since $|S \triangle R| \ge 3$, after this happens, there is still some ``randomness'' left for at least one further edge in $S \triangle R$, which is now (i.e. after revealing the labels of $v$ in $\mathcal{D}_{S \cup R \cup T}$) independent of $\mathds{1}(v \in N_G(T))$. Therefore, using the same argument as before,
    \begin{align*}
        &\Expectedsubnop{\Omega}{ u_{S,p}(v)u_{R,p}(v) \mid \mathcal{E}_{(1)} \cap \{ v \in N_G(T) \} }
        \\&\hspace{3cm} \le O\bigl(\log(n) \delta\sqrt{\delta d} \bigr) \sqrt{\maxbal{}} \sqrt{\maxbal{2r-1}} \\&\hspace{3cm} = O\bigl(\log(n) \delta\sqrt{\delta d} \bigr) \maxbal{r}. 
    \end{align*}
    In total, this yields that 
    \begin{align*}
        \Expectednop{\USRgood} \le \maxbal{r} \left( O\left( \log(n)^4 n\delta^2 \right)  + O\bigl( \log(n)^3 n\delta^2\sqrt{\delta d} \bigr) + O\bigl(\log(n)^3 n\delta^3 (\delta d)^{3/2}\bigr) \right).
    \end{align*}
    Using \Cref{lem:sizeofdelta} and the assumption $d \ge n^{\alpha}$, it follows that $\Expectednop{\USRgood} \le o(k) ((1-p)/p)^r$.
    Now, using a Bernstein bound as in \Cref{lem:bernstein} and setting $z = C \sqrt{n\log(n)}$ for $C$ sufficiently large, we get that \begin{align*}
        &\Pr{|\USRgood| - \Expectednop{\USRgood}| \ge z } \le \\
        &\hspace{2cm} 2 \exp \left( -z^2 / \left(8\maxbal{2r}n + \frac{4}{3} \maxbal{r}z \right) \right)\le n^{-2t - 2r}
    \end{align*}
    Using a union bound over all $S, R \subseteq A$ with $r$ vertices yields that for all such $S, R$, we have \begin{align*}
        \USRgood = \bigl( o(k) + O(\sqrt{n\log(n)}) \bigr) \maxbal{r} = o(k) \maxbal{r}
    \end{align*} since $k \gg \sqrt{n\log(n)}$, and this holds with probability $1 - n^{-2t}$.

    \paragraph{Plugging it together} In total, we have shown that \begin{align*}
        |\USR| \le |\USRgood| + |\USRbad| &\le o(k) \maxbal{r} + (1+o(1))b \maxbal{r} \cdot  p^{t-\eta}k + o(k) \\
        &\le (1+o(1))b \maxbal{r} \cdot  p^{t-\eta}k + o(k)
    \end{align*} holds for all $S, R \subseteq A$ with $|S| = |R| = r$ and $|S \triangle R| \ge 3$ with probability $1 - n^{-2t}$, as desired.
\end{proof}

\subsection{When do balanced subgraphs appear?}

\newcommand{\Gi}{G[A_i \uplus (B_i \cap \neigh{T})]}

After having proved \Cref{lem:balancedness-core}, it remains to give a criterion for when balanced subgraphs actually appear in a $G \sim \RIGone$. We formalize this in \Cref{lem:general-balancedness-generalized} which on a high level states the following. If we fix a $T \in \binom{[n]}{t}$ and use a procedure that constructs a bipartite sub-graph $G' = G[A \uplus B]$ of $G$ while only having access to the information of whether a vertex has chosen a label in a set $L \subseteq [d]$ or not, then $G[A \uplus (B \cap \neigh{T})]$ has good balancedness properties if no vertex in $A$ has chosen a label in $L$, and not too many vertices in $T$ have chosen any given $\ell \in [d] \setminus L$.

As a natural example consider the case where $L$ only contains a single label $\ell$. Then we can think of the procedure described above as some algorithm that outputs a bipartition $A, B$ such that $A \cap S_\ell = \emptyset$. If every label $\ell ' \in [d] \setminus \{\ell\}$ now appears in only a limited number of vertices in $T$, then $G[A \uplus (B \cap \neigh{T})]$ has good balancedness. 

Formally, we represent a procedure for choosing a such bipartition as an algorithm $\textsc{Alg}$ that takes as input two disjoint sets ($S_\ell$, and $[n] \setminus S_\ell$) and outputs a bipartition $A \uplus B$ such that $S_\ell \cap A = \emptyset$ or $S_\ell \cap B = \emptyset$. We write $\textsc{Alg}(S_\ell, [n] \setminus S_\ell)$ to denote this bipartition $A \uplus B$.

\begin{definition}[Procedures for Choosing a Bipartition]\label{def:procedure-bipartition}
    Let $m \in \mathbb{N}$, fix a subset $M \subseteq [n]$, and let $\{ \textsc{Alg}_i \}_{i \in [m]}$ be any collection of (possibly randomized) algorithms that each on the input of a set $S \subseteq [n]$ and $[n]$ output a bipartition of $M$. Denote the bipartition computed by  $\textsc{Alg}_i$ by $A_i \uplus B_i$ or alternatively $\textsc{Alg}_i(S, [n])$.
    If $A_i \cap S = \emptyset$ for all $i\in [m]$, we call the collection $\{ \textsc{Alg}_i \}_{i \in [m]}$ a \emph{valid} collection of procedures for choosing a bipartiton.
\end{definition}

The following lemma states that given any fixed collection of such algorithms $\{ \textsc{Alg}_i \}_{i \in [m]}$ and some assignment from $t$-tuples $T \in \binom{[n]}{t}$ to sets $L(T) \subseteq [d]$, it holds w.h.p. over the draw of $G \sim \RIGone$, that the bipartite graphs $\{G[A_i \uplus (B_i \cap \neigh{T})]\}_{i \in [m]}$ are all simultaneously balanced for all $t$-tuples $T$ such that any label $\ell \in [d] \setminus L(T)$ only appears in at most $\eta$ vertices of $T$.

\begin{lemma}[General Criterion for Balancedness of sub-graphs of $G \sim \RIGone$]\label{lem:general-balancedness-generalized}
    Let $t, \eta \in \mathbb{N}, t \ge 3$ be fixed constants and let further $\{ \textsc{Alg}_i \}_{i \in [m]}$ be a valid collection of procedures for choosing a bipartition according to \Cref{def:procedure-bipartition} for some $m \le n$. Moreover, associate to each $T \in \binom{[n]}{t}$ a set $L(T) \subseteq [d]$, and denote by $S(T) \coloneqq \bigcup_{\ell \in L(T)} S_\ell$. For a given $T$ and $i \in [m]$, denote by $A_i \uplus B_i$ the bipartition output by $\textsc{Alg}(S(T), [n])$.
    
    Then, with high probability over the draw of $G \sim \RIGone$ and the randomness in $\{ \textsc{Alg}_i \}_{i \in [m]}$, the following holds. For all $t$-tuples $T \in \binom{[n]}{t}$ such that no label in $[d] \setminus L(T)$ appears in more than $\eta$ vertices of $T$, the bipartite graphs $\Gi$ have $r$-fold balancedness 
    $$
        \frac{16p^{-\eta} }{\alpha} \maxbal{r}  \cdot tp^{t}k + o(k)
    $$
    for all $i \in [m]$ and $3 \le r \le t$. 
\end{lemma}
\begin{proof}
    First, fix some $T \in \binom{[n]}{t}$ (we shall later afford a union bound over all $T$) and think of sampling $G \sim \RIGone$ in 3 phases as follows.
    \begin{enumerate}
        \setlength\itemsep{.01em}
        \item \textbf{Phase 1:}
            Reveal the labels in $T$. This uncovers the set $L(T)$.
        \item \textbf{Phase 2:}
            Reveal for all vertices $v \in [n]$ the labels in $L(T)$. This uncovers the set $S(T)$.
        \item \textbf{Phase 3:} Reveal all the remaining labels.
    \end{enumerate}

    Conditioned on the event that no label in $[d]$ appears in more than $\eta$ vertices in $T$ after Phase~1, we now show that all the $\Gi$ are balanced after Phase~3 w.h.p.
    
    We will later use a union bound over all $m$ algorithms in $\{ \textsc{Alg}_i \}_{i \in [m]}$, so it suffices that the desired balancedness properties hold for any fixed $\textsc{Alg}$ in the collection $\{ \textsc{Alg}_i \}_{i \in [m]}$ after Phase~3 with sufficiently high probability. To this end, let $\textsc{Alg}$ be one fixed such algorithm and split Phase~3 into three sub-phases as follows.
    \begin{enumerate}
        \setlength\itemsep{.01em}
        \item \textbf{Phase 3.1:}
            Compute the bipartition $A \uplus B = \textsc{Alg}(S(T), [n])$.
        \item \textbf{Phase 3.2:}
            Reveal all the remaining labels for vertices in $A \setminus T$.
        \item \textbf{Phase 3.3:}
            Reveal all the remaining labels for vertices in $B\setminus T$.
    \end{enumerate}
    To show that $G[A \uplus (B \cap \neigh{T})]$ is balanced, after these phases, we wish to use \Cref{lem:balancedness-core}, and proceed by showing that the prerequisites for this are met after Phase~3.2 with probability $1 - O(n^{-2t})$.
    \begin{enumerate}
        \setlength\itemsep{.01em}
        \item For condition one of \Cref{lem:balancedness-core}, we note that after Phase~3.2 all labels in $A \cup T$ are revealed, while in $B$ only the labels in $L(T)$ have been uncovered. Thus, prerequisite one is met for $L = L(T)$.
        \item Regarding the second condition, note that no vertex in $A$ has chosen label $L(T)$, by the fact that $\textsc{Alg}$ is valid according to \Cref{def:procedure-bipartition}.
        \item The third condition is met since we condition on the event that that no label in $[d]$ appears in more than $\eta$ vertices in $T$ after Phase~1.
        \item For condition four, we get from \Cref{lem:E-is-likely} that $\mathcal{E}(A, t, \frac{8t}{\alpha})$ and $\mathcal{E}(T, t, \frac{8t}{\alpha})$ are met with  probability $1 - O(n^{-2t})$.
        \item For condition five, we note that by the occurrence of $\mathcal{E}(T, t, \frac{8t}{\alpha})$, the number of duplicate labels in $T$ is already at most $\log(n)^2$, and the total number of labels is $\le 2t\delta d$. Hence, for every fixed $S, R\subseteq A$ with $|S| = |R| = r$, the expected number of additional duplicate labels in $S \cup R \cup T$ is at most $2t\delta^2d + |S \cup R|^2\delta^2d = O(1)$. Using a Chernoff bound like in the proof of \Cref{lem:E-is-likely}, we get that the probability that there are at least $\log(n)^2$ additional duplicates is $n^{-\omega(1)}$.
        \item Finally, for condition 6, we note that a simple Chernoff bound with subsequent union bound in Phase 2 yields that for all labels $\ell \in L(T)$, the number of vertices that choose $\ell$ in step 2 is at most $2k$ with probability $1 - n^{-2t}$. 
    \end{enumerate}

    Hence, applying \Cref{lem:balancedness-core} after Phase~3.2 (over the randomness in Phase~3.3), we get that $\Gi$ has $r$-fold balancedness \begin{align*}
        (1+o(1)) \frac{8t}{\alpha} p^{-\eta} \maxbal{r} \cdot tp^{t}k + o(k) \le \frac{16p^{-\eta} }{\alpha} \maxbal{r}  \cdot tp^{t}k + o(k)
    \end{align*} for all $3 \le r \le t$ with probability $1 - n^{-\omega(1)}$. Taking a union bound over all $m \le n$ elements in $\{ \textsc{Alg}_i \}_{i \in [m]}$, shows that all the $\Gi$ have the desired balancedness properties after Phase~3 given that no label in $[d] \setminus L(T)$ appears in more than $\eta$ vertices of $T$, and this holds with probability $1 - O(m n^{-2t})$.
    A final union bound over all $T$ yields that our lemma holds with probability $1 - O(n^{-t+1})$, as desired.
\end{proof}

A useful implication of the above used multiple times when analysing our algorithms is the following statement telling us that for \emph{most} $t$ tuples in any $S_\ell$, the graph $G[([n] \setminus S_\ell) \uplus ([n] \cap \neigh{T})]$ has good balancedness.

\begin{lemma}[Balancedness for most $T \subseteq S_\ell$]\label{lem:general-balancedness}
    Let $t \in \mathbb{N}, t \ge 6$ be any fixed constant and let further $\{ \textsc{Alg}_i \}_{i \in [m]}$ be a valid collection of procedures for choosing a bipartition according to \Cref{def:procedure-bipartition}. For a given $\ell \in [d]$ and $T \in \binom{S_\ell}{t}$ and $i \in [m]$, denote by $A_{i,\ell} \uplus B_{i, \ell}$ the bipartition output by $\textsc{Alg}(S_\ell, [n])$. 
    
    Then, it holds with probability $1-o(1)$ over the randomness of $G \sim \RIGone$ and $\{ \textsc{Alg}_i \}_{i \in [m]}$ that for all labels $\ell \in [d]$, a $1-o(1)$ fraction of all $t$-tuples $T \in \binom{S_\ell}{t}$ is such that $G[A_{i,\ell} \uplus (B_{i, \ell}]$ has $r$-fold balancedness 
    $$
        \frac{16p^{-6} }{\alpha} \maxbal{r}  \cdot tp^{t}k + o(k)
    $$
    for all $i \in [m]$ and $3 \le r \le t$. 
\end{lemma}
\begin{proof}
    First of all, applying \Cref{lem:general-balancedness-generalized} with $\eta = 6$ and $L(T) = \{\ell\}$ and $S(T) = [n] \setminus \{\ell\}$ and union bounding over all $\ell$ yields that our conclusion holds with probability $1 - O(dn^{-t + 1})$ for all $\ell \in [d]$ and $T \in \binom{S_\ell}{t}$ such that at most $6$ vertices in $T$ have chosen any $\ell' \in [d] \setminus \{\ell\}$.

    Hence, it only remains to show that a $(1-o(1))$ fraction of all $t$-tuples have this property. To this end, we simply note that when uncovering the labels of any given $T \subseteq S_\ell$, the probability that there is a label in $\ell' \in [d] \setminus \{\ell\}$ that is chosen by at least $6$ vertices in $T$ is at most 
    $$
        d \binom{t}{6} \delta^6 = O\left( \frac{ dt^6}{\sqrt{d}^{6}} \right) = O\left( \frac{1}{d^{2}} \right).
    $$
    Hence, the expected fraction of $t$ tuples in $S_\ell$ for which this bad event occurs is $O(1/d^2)$. By Markov's inequality, we get that with probability $ 1 - O(1/(d\sqrt{d}))$, at most a $1/(d\sqrt{d}) = o(1)$ fraction of such $t$-tuples exists in $S_\ell$. Union bounding over all $\ell \in [d]$ yields that this holds for all $\ell$ with probability $1 - O(1/\sqrt{d}) = 1 - o(1)$, as desired.
\end{proof}

\subsection{Cliques Intersecting other cliques}\label{sec:cliqueintersections}

Another crucial set of statements used in multiple places throughout this work concerns showing that cliques do (with high probability) not occur in certain parts of a $G \sim \RIGone$. 


To this end, we are crucially interested in the following event 
\begin{definition}\label{def:Eab}
    Define the event $\mathcal{E}_{a,b}$ as the event that there exist two sets $S, T \subseteq [n]$ in a $G \sim \RIGone$ with the following properties (see also \Cref{fig:bad-event}). 
    \begin{enumerate}
        \setlength\itemsep{.01em}
        \item $S \cap T = \emptyset$, $|S| = a, |T| = b$
        \item There is a label $\ell \in [d]$ such that all vertices in $S$ have chosen $\ell$, while \emph{none} of the vertices in $T$ has chosen $\ell$.
        \item All edges in $S \times T$ are part of $G \sim \RIGone$.
    \end{enumerate}
\end{definition}

We show that $\mathcal{E}_{a,b}$ is extremely unlikely when choosing some small $b = \text{polylog}(n)$ and $a \ge \frac{n^{\varepsilon} k}{\sqrt{d}}$ for arbitrarily small $\varepsilon > 0$. We remark that this order of magnitude is almost optimal for $a$, since  for all $\ell_1, \ell_2 \in [d]$, among all $\approx k$ vertices that have chosen label $\ell_1$, there will be a set of $\approx k/\sqrt{d}$ vertices which has also chosen $\ell_2$. Accordingly, choosing $S$ as this subset and $T$ as an arbitrary set of vertices which have chosen $\ell_2$ but not $\ell_1$ contradicts $\overline{\mathcal{E}_{a,b}}$ for $a \approx k/\sqrt{d}$ and all $b \ll k$ w.h.p. 

\begin{figure}[t]
\centering
\includegraphics[height=0.2\textheight]{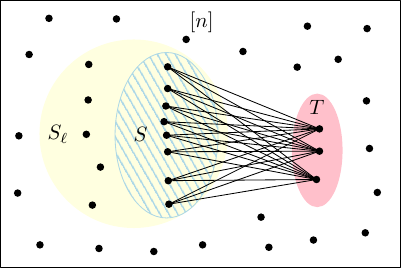}\hfill
\includegraphics[height=0.2\textheight]{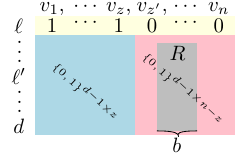}

\caption{(a) Sketch of event $\mathcal{E}_{a,b}$. Vertices $S_\ell$ (yellow area) is the set of vertices with label $\ell$. The set of vertices $S$ (hatched area) is a subset of $S_{\ell}$, while no vertex of $T$ (red area) has label $\ell$. Any vertex of $T$ has all possible edges the set $S$. (b) Illustration of revealing three label revealing phases in \Cref{lem:bad-event} 1 (yellow), 2 (red) and 3 (blue). The grey area represents a set $R$ with $b$ vertices with bounded duplicate labels in \Cref{claim:duplicates}.}
    \label{fig:bad-event}

\end{figure}

\begin{lemma}\label{lem:bad-event}
    For any $\alpha \in (0,1)$ and $d = n^{\alpha}$, we get that $\Pr{\mathcal{E}_{a,b}} = n^{-\omega(1)}$ for all $a \ge \frac{n^{\varepsilon} k}{\sqrt{d}}$ and $b = \lfloor \log(n)^2 \rfloor$.
\end{lemma}
\begin{proof}
    For any label $\ell \in [d]$, denote by $\mathcal{E}_{a,b,\ell}$ the event that there exist sets $S \subseteq S_\ell, T \subseteq [n] \setminus S_\ell$ with $|S| = a, |T| = b$ such that all edges in $S \times T$ are part of $G$. Clearly, if we show that $\Pr{\mathcal{E}_{a,b,\ell}} \le n^{-\omega(1)}$ for all $\ell \in [d]$, then a simple union bound implies that $\Pr{\mathcal{E}_{a,b}} \le n^{-\omega(1)}$ and we are done. Hence, for the rest of the proof, we fix an arbitrary $\ell \in [d]$ and show that $\Pr{\mathcal{E}_{a,b,\ell}} \le O(1/d^2)$.
    
    Given a fixed $\ell$, uncover the labels of $G \sim \RIGone$ in the following three phases

          \begin{enumerate}\label{enum:labelreveal}
        \setlength\itemsep{.01em}
        \item\label{item:phase1} \textbf{Phase 1:} Uncover the label $\ell$ for all $v \in [n]$ and thus reveal the set $S_\ell$.
        \item\label{item:phase2} \textbf{Phase 2:} Uncover the missing labels for vertices in $[n] \setminus S_\ell$.
        \item\label{item:phase3} \textbf{Phase 3:} Uncover the missing labels for vertices in $S_\ell$.
    \end{enumerate}

    First, uncover the set $S_\ell$, then uncover the missing labels of all vertices in $[n] \setminus S_\ell$, and finally uncover the missing labels of all vertices in $S_\ell$. We show that the following event occurs w.h.p. after Phase~2 (see also \Cref{fig:bad-event}) asserting that the number of duplicate labels is small in every set of vertices of size $b$.

    \begin{claim}\label{claim:duplicates}
        There is a constant $C > 0$ such that with probability $1 - n^{-\omega(1)}$, the following occurs after Phase~2. For all sets $R \subseteq [n] \setminus S_\ell$ with $|R| = b$, there are at most $C\log(n)^4$ labels that were chosen by at least two vertices in $R$, such labels are called \emph{duplicates}. Moreover, for every $v \in [n] \setminus S_\ell$, $|M_v| = \delta d \pm \log(n)\sqrt{\delta d}$.
    \end{claim}\begin{proof}
        Given a fixed label $\ell' \in [d] \setminus \{\ell\}$, the probability that at least 2 vertices in any fixed $R$ choose $\ell'$ is at most 
        $$
            \binom{b}{2} \ \delta^2 = O\left( \frac{\log(n)^2}{\sqrt{d}} \right)^2.
        $$
        Hence the expected number of labels for which this happens is $O(\log(n)^4)$. Now, applying a Chernoff bound yields that the probability that the number of such labels is $\omega(\log(n)^4)$ is at most $n^{-\log(n)^3}$.
       
        Finally, a union bound over all $\binom{n}{b} \le n^{\log(n)^2}$ sets yields the first part of the claim. The second part follows directly by \Cref{lem:labelconcentration}.
    \end{proof}

    Given this fact, we now bound the probability that all edges between some $T \subseteq [n]\setminus S_\ell$ with $|T| = b$ and some $S \subseteq S_\ell$ with $|S| \ge a$ exist after Phase~3. For a fixed $v \in S$, the probability that all edges between $v$ and $S$ exist is at most 
    $$
        \Pr{v \sim T} \le C\log(n)^4\delta + \big(p + O(\log(n)\delta\sqrt{d\delta})\big)^{b}
    $$ where the first term accounts for the possibility that $v$ choses a duplicate in $T$. The second term accounts for the event that this does not happen, here we applied \Cref{lem:pconcentration} and noted that all edges between $v$ and $T$ are independent if no duplicate is chosen. Noting that $\log(n)\delta\sqrt{d\delta} = n^{-\Omega(1)}$ since $d \ge n^\alpha$, and that $p \le 1/2$ yields that $\Pr{v \sim T} \le C\log(n)^4\delta + (1+o(1))p^{b}$. Since (conditional on the fixed lables in $T$), the event $v \sim T$ is independent for all $v \in S_\ell$, we get that 
    \begin{align*}
         \Pr{S \sim T} \le \big( C\log(n)^4\delta + (1+o(1))p^{b}\big)^{a} = O\left(\frac{\log(n)^4}{\sqrt{d}}\right)^a
    \end{align*} since $\delta = O(1/\sqrt{d})$ and $p^{b} = o(1/\sqrt{d})$ because $b = \Omega(\log(n)^2)$. Now, noting that $|S_\ell| \le 2k$ with probability $1 - n^{-\omega(1)}$ (by using a Chernoff bound in conjunction with expected number of vertices with label $\ell$ being $\delta n = k = n^{\Omega(1)}$) after Phase~1, and employing a union bound over all $S$ and $T$ in Phase~3 yields (by the law of total probability) that 
    \begin{align*}
        \Pr{\mathcal{E}_{a, b, \ell}} - n^{-\omega(1)} &\le \sum_{m = a}^{2k} \binom{2k}{m} \binom{n}{b}  O\left(\frac{\log(n)^4}{\sqrt{d}}\right)^m \\
        &\le \sum_{m = a}^{2k} \left( \frac{4e k}{m} \right)^m  \left( \frac{2e n}{b} \right)^b  O\left(\frac{\log(n)^4}{\sqrt{d}}\right)^m \\
        &\le \sum_{m = a}^{2k} O\left(\frac{k\log(n)^4}{m\sqrt{d}}\right)^m  \left( \frac{2e n}{b} \right)^b \\
        &\le \sum_{m = a}^{2k} \exp \left( m\log\left(\frac{k}{m\sqrt{d}}\right)+ O(m \log\log(n)) + O(b\log(n)) \right).
    \end{align*}
   Since $k/(m\sqrt{d}) \le n^{-\varepsilon}$ by our assumption $a \ge n^\varepsilon k / \sqrt{d}$ and $a \gg b$, we get that the above is at most $$
        \Pr{\mathcal{E}_{a, b, \ell}} \le \sum_{m = a}^{2k} \exp(-(1-o(1)) \varepsilon m \log(n)) \le n^{-\omega(1)}.
    $$ A final union bound over all $\ell \in [d]$ then finishes the proof.
\end{proof}
We further bound the maximum degree any vertex outside of a ground-truth clique $S_\ell$ can have into $S_\ell$.

\begin{lemma}\label{lem:degreebound}
    With high probability over the draw of $G \sim \RIGone$, the following holds. For all $\ell \in [d]$ and all $v \in [n] \setminus S_\ell$, the number of neighbors of $v$ in $S_\ell$ is at most $(1+o(1))pk$. 
\end{lemma}
\begin{proof}


Fix a label $\ell \in [d]$ and reveal the labels of $[n]$ in the same three phases we used in the proof of \Cref{lem:bad-event}. Fixing a vertex $v \in [n] \setminus S_\ell$, we obtain via \Cref{lem:labelconcentration} that $|M_v|$ satisfies the requirement of \Cref{lem:pconcentration} after Phase~2 with probability $1 - n^{-\omega(1)}$. Then, making use of \Cref{lem:pconcentration}, it follows that for any $u \in S_\ell$ that $\Pr{u\sim v} \leq (1+o(1))p$ over the randomness in Phase~3. Moreover, via a Chernoff bound we obtain that $|S_\ell|\leq (1+o(1))k$ with probability $1 - n^{-\omega(1)}$. Using another Chernoff bound it then follows that $v$ has at most $(1+o(1))pk$ with probability  $1 - n^{-\omega(1)}$. In conclusion, our fixed vertex $v$ has at most $(1+o(1))pk$ neighbours in $S_\ell$ for our fixed label $\ell$ with probability $1 - n^{-\omega(1)}$ by a union bound over the complementary events of the ones we considered throughout the three phases. Hence, a union bound over all $|[n]\setminus S_\ell| \leq n$ vertices yields that the event that there exists a vertex $v \in [n] \setminus S_\ell$ where the size of the neighbourhood in $S_\ell$ exceeds $(1+o(1))pk$ has probability at most $n^{-\omega(1)}$. Finally, another union bound over all $d$ labels finishes the proof. \qedhere
\end{proof}

\section{Using Balancedness to Certify the Absence of Large (Bi)cliques} \label{sec:sos}

After having established criteria for when balanced bipartitions appear in an RIG, we go over to showing how these properties can be used to certify the absence of (bi)cliques within the SoS-proof system. 

\subsection{Preliminaries on SoS and Pseudo-Distributions}\label{sec:sosprelims}

We start with some preliminaries on SoS-proofs, pseudo-distributions and their algorithmic feasibility. We refer to the book \cite{Fleming_Kothari_Pitassi_2019} for a more comprehensive introduction.

\paragraph{Pseudo-distributions}

Given two functions $\mu, f: \{0,1\}^m \rightarrow \mathbb{R}$, we define the \emph{pseudo-expectation} of $f$ under $\mu$ as $\Expectedtildesub{\mu}{f} = \sum_{x \in \{0,1\}^m} \mu(x)f(x)$. A degree-$d$ pseudo-distribution is a function $\mu: \{0,1\}^m \rightarrow \mathbb{R}$ such that $\sum_{x \in \{0, 1\}^m} \mu(x) = 1$ and \emph{for every} polynomial of degree at most $d/2$, we have $\Expectedtildesub{\mu}{f^2} \ge 0$. Given a set $\cA = \{ f_1 \ge 0, f_2 \ge 0, \ldots, f_{m'} \ge 0\}$ of polynomial inequality constraints, we say that a degree-$d$ pseudo-distribution satisfies the system of constraints $\cA$ at degree $r$ if
for every $S \subseteq [m']$ and every polynomial $g$ such that $\deg(g^2) + \sum_{i \in S} \max\{\deg(f_i), r \} \le d$, we have $\Expectedtildesub{\mu}{ g^2 \prod_{i \in S} f_i} \ge 0$. Given some $\eta > 0$, we say that the pseudo-distribution $\mu$ satisfies the system $\cA$ $\eta$-approximately at degree $r$ if $\Expectedtildesub{\mu}{g^2 \prod_{i \in S} f_i } \ge -\eta \|g^2\|_2 \prod_{i \in S} \|f_i\|_2$ again for every polynomial $g$ such that $\deg(g^2) + \sum_{i \in S}\deg(f_i) \le d$. Here $\|f\|_2$ denotes the $\ell_2$-norm of the coefficient vector of the polynomial $f$. We say that $\mu$ ($\eta$-approximately) satisfies $\cA$ if it does so at degree $0$.

In this setting, it has been observed that there exists an efficient algorithm that given a system of polynomial constraints $\cA$ computes a degree-$d$ pseudo distribution that approximately satisfies $\cA$ in polynomial time.

\begin{theorem}[Efficient optimization over pseudo-distributions \cite{parrilo2000structured, nesterov2000squared, lasserre2001new}]\label{thm:optpseudo}
    There exists an algorithm that takes as input a system of $m$ polynomial constraints $\cA$ in $n$ variables with rational coefficients that for every $\eta > 0$ outputs a degree-$d$ pseudo-distribution that $\eta$-approximately satisfies $\cA$. It runs in time $(n + m)^{O(d)}\text{polylog}(1/\eta)$.
\end{theorem}

We further use the following standard fact about pseudo-distributions (see \cite[Theorem 2]{agahi2011holder}).

\begin{lemma}[Hölder's Inequality for pseudo-distributions]\label{lem:holder}
    Let $f,g$ be polynomial of degree at most $d$. Then, for any $t\in\mathbb{N}$ and any degree-$dt$ pseudo-distribution $\mu$, the following holds.
    $$
        \Expectedtildesub{\mu}{f^{t-1}g} \le \left(\Expectedtildesub{\mu}{f^t}\right)^{(t-1)/t} \left(\Expectedtildesub{\mu}{g^t}\right)^{1/t}.
    $$
\end{lemma}

\paragraph{Sum-of-Squares proofs}

Very useful for analysing pseudo-distributions are \emph{sum-of-squares proofs}. A degree-$d$ sum-of-squares (SoS) proof that the system of constraints $\cA = \{f_1 \ge 0, \ldots, f_m \ge 0\}$ implies the constraint $\{g \ge 0 \}$ is a collection $\{p_{S}\}_{S \subseteq [m]}$ of sum-of-squares polynomials (i.e. polynomials that can be represented as the sum of squared polynomials) such that
$g = \sum_{S \subseteq [m]} p_S \cdot \prod_{i\in S}f_i$, while for all $S \subseteq [m]$, the polynomial $p_S \cdot \prod_{i\in S}f_i$ has degree at most $d$. We write $\cA \sststile{x}{d} \{g \ge 0\}$ to express that there is a degree-$s$ SoS-proof in variables $x$ that $\cA$ implies $\{g \ge 0\}$.

The connection to pseudo-distributions is given by the following fact which states that whenever there is a degree-$d$ SoS-proof that $\cA$ implies $\{g \ge 0\}$, then every degree-$d$ pseudo distribution that satisfies $\cA$ also satisfies $\{g \ge 0\}$. 
\begin{fact}[Soundness]
    If $\mu$ is a degree-$d$ pseudo distribution that (approximately) satisfies the system of constraints $\cA$ and if $\cA \sststile{x}{r} \{g \ge 0\}$, then $\mu$ satisfies $\{g \ge 0\}$ at degree $r$.
\end{fact}
Moreover, we have to constrain the so-called \emph{bit-complexity} of a SoS-proof, in order to ensure that the above soundness is meaningfully reflected in pseudo-distributions that can be computed in polynomial time. 
\begin{definition}[Bit-complexity of SoS-proofs]
    We say that the system of constraints $\cA = \{ f_1 \ge 0, f_2 \ge 0, \ldots, f_m \ge 0 \}$ derives $\{g \ge 0\}$ in degree $d$ and bit-complexity $B$ if there is a degree-$d$ sum-of-squares proof $g = \sum_{S \subseteq [m]} p_S \cdot \prod_{i\in S}f_i$ such that the total number of bits required to describe all the coefficients of the polynomials $f_i, g, p_S$ is at most $B$.
\end{definition}

We further introduce some useful SoS-proofs that we shall use later.
The following is taken from \cite[Lemma 3.13]{Kothari-STOC-2023} with minor modifications.
\begin{lemma}[SoS Cancellation Inequalities]\label{lem:cancellation}
    Let $f$ be a polynomial of degree $d$ in indeterminate $x$. Then,
    \begin{align*}
        \{ f^2 \le C f \} &\sststile{x}{2d} \left\{ f^2 \le C^2 \right\}\\
        \{ f^2 \le C  \} &\sststile{x}{2d} \{ f \le \sqrt{C} \}
    \end{align*}
\end{lemma}
\begin{proof}
    For the first part, note that $$
         \{ f^2 \le C f \} \sststile{x}{2d} \left\{ f^2 \le f^2 + (f - C)^2 =  2f^2 - 2fC + C^2 = C^2 + 2(f^2 -Cf) \le C^2 \right\}.
    $$ Moreover, we note that $\sststile{x}{2d} \{ (f + \sqrt{C})^2 \le 2f^2 + 2C \}$, which allows us to conclude that
    \begin{align*}
        \{ f^2 \le C \} \sststile{x}{2d} \left\{ f = \frac{1}{4\sqrt{C}} (f + \sqrt{C})^2 - \frac{1}{4\sqrt{C}}(\sqrt{C} - f)^2 \le \frac{1}{2\sqrt{C}}f^2 + \frac{1}{2}\sqrt{C} \le \sqrt{C} \right\}.
    \end{align*}
\end{proof}

\begin{lemma}[SoS Almost Triangle Inequality, Fact 3.9 in \cite{Kothari-STOC-2023}]\label{lem:sos-almost-triangle}
    Let $f_1, f_2, \ldots, f_s$ be indeterminates for $1 \le i \le s$ and some $s \in \mathbb{N}$. For any $t \in \mathbb{N}$, 
    \begin{align*}
        \sststile{f_1, f_2, \ldots, f_s}{2t} \left \{ \left( \sum_{i=1}^s f_i \right)^{2t} \le s^{2t - 1} \sum_{i=1}^s f_i^{2t} \right \}.
    \end{align*}
\end{lemma}
We further use the following lemma that relates the sum over all variables $x_S$ for $S \subseteq A$ with $|S| = r$ to $|x|^r$ from 
\cite{Kothari-STOC-2023}.
\begin{lemma}[Lemma 4.10 in \cite{Kothari-STOC-2023}]\label{lem:sos-factorial}
    For every $\varepsilon > 0$, $r \in \mathbb{N}$ there is an SoS-proof with coefficients of bit complexity $\text{poly}(|A|, \log(1/\varepsilon))$ for the following statement. 
    $$
        \left\{ x_v^2 = x_v \right\}_{v \in A} \sststile{x}{r} \left\{ \frac{1}{2^r r!} |x|^r - \frac{2r^r}{r!} - \varepsilon \le \sum_{S \subseteq A, |S| = r} x_S \le \frac{1}{r!}|x|^r + \varepsilon \right\}.
    $$
\end{lemma}
Related to this, we use the following proposition, which is occasionaly more convenient than the above. 
\begin{lemma}\label{lem:sosbinomialcoefficient} There is an SoS proof with bit complexity $\text{poly}(|A|, \log(1/\varepsilon))$ for the following statement
    \begin{align*}
        \left\{ x_v^2 = x_v \right\}_{v \in A} \sststile{x}{r} \left\{  \left(\sum_{S \subseteq A, |S| = r} x_S\right) = \frac{1}{r!} \prod_{i=1}^r \left( \left( \sum_{v \in A} x_v \right) - (i-1)\right) \right\}.
    \end{align*}
    Consquently, for any $s \in \mathbb{N}$,
    $$
        \bigg\{ \sum_{v\in A} x_v \ge s \bigg\} \cup \left\{ x_v^2 = x_v \right\}_{v \in A} \sststile{x}{r} \left\{  \left(\sum_{S \subseteq A, |S| = r} x_S\right) \ge \binom{s}{r}\right\}.
    $$
\end{lemma}
\begin{proof}
    It is not hard to see that under the set of axioms $\left\{ x_v^2 = x_v \right\}_{v \in A} $ when evaluating the product $\left( \sum_{S \subseteq A, |S| = r} x_S \right) \left( \sum_{v \in A} x_v \right)$, we enumerate each set $S \subseteq A$ with $|S| = r + 1$ for $r+1$ times and each $S \subseteq A$ with $|S| = r$ for $r$ times. Therefore, 
    \begin{align*}
        \left\{ x_v^2 = x_v \right\}_{v \in A} \sststile{x}{r+1} \left\{  \left(\sum_{S \subseteq A, |S| = r + 1} x_S\right) = \frac{1}{r+1} \ \left( \sum_{S \subseteq A, |S| = r} x_S \right) \left( \left( \sum_{v \in A} x_v \right) - r\right) \right\}.
    \end{align*} Iterating this argument yields the first part of the lemma. The second part of the lemma then follows by definition of the binomial coefficient.
\end{proof}

\subsection{Systems of Axioms for Encoding Cliques}\label{sec:axioms}

We use the following axioms to describe cliques and bicliques that serve as starting points for our SoS proofs. 
The most commonly used axioms for encoding a clique of size $k$ at least $k$ (and at most $2k$) are as follows (see for example also \cite{Kothari-STOC-2023}).
\begin{align*}
    \mathcal{A}(G, k) = \left \{
    \begin{aligned}
      &\forall v \in V
      & x_v^2
      & =x_v \\
      &
      & k \le \sum_{v \in A} x_v  &\le 2k\\
      &\forall u,v \in V \text{ s.t. } \{u,v\} \not\in E
      & x_u x_v
      & = 0
    \end{aligned}
  \right \}\mper
\end{align*}
Above, we have one variable $x_v$ for each $v \in V$ indicating whether or not $v$ participates in a clique or not.

Very similar axioms $\mathcal{B}$ can be used in a bipartite graph $H = (A\uplus B, E_H)$ to find a biclique of size at least $k \times k$. Now, we denote the variables for vertices $v \in A$ by $x_v$ and variables for vertices $v \in B$ by $y_v$.
\begin{align*}
    \mathcal{B}(H, k) = \left\{ 
        \begin{aligned}
            &\forall v \in A & x_v^2 &= x_v \\
            &\forall v \in B & y_v^2 &= y_v \\
            &                &k \le \sum_{v \in A} x_v  &\le 2k\\
            &                &k \le \sum_{v \in A} y_v  &\le 2k\\
            &\forall u \in A, v \in B \text{ s.t. } \{u, v\} \notin E_H & x_uy_v  &= 0
        \end{aligned}
    \right\}.
\end{align*}

Removing the size constraint from the above, we end up with what we call the \emph{reduced biclique axioms} $\mathcal{R}$ which we use at some parts in our proofs. 

\begin{align*}
    \mathcal{R}(H) = \left\{ 
        \begin{aligned}
            &\forall v \in A & x_v^2 &= x_v \\
            &\forall v \in B & y_v^2 &= y_v \\
            &\forall u \in A, v \in B \text{ s.t. } \{u, v\} \notin E_H & x_uy_v  &= 0
        \end{aligned}
    \right\}.
\end{align*}

For the proofs in \Cref{sec:stronger-adv}, we further need the following constraints that specify not a clique but a dense sub-graph in which up to $\gamma$ edges incident to all vertices can be missing. These are called $\mathcal{R}(H, \gamma)$ and defined follows.

\begin{align*}
    \mathcal{R}(H, \gamma) = \left\{ 
        \begin{aligned}
            &\forall v \in A & x_v^2 &= x_v \\
            &\forall v \in B & y_v^2 &= y_v \\
            &\forall u \in A, v \in B & z_{u,v}^2 &= z_{u,v} \\
            &\forall u \in A, v \in B \text{ s.t. } \{u, v\} \notin E_H & x_uy_v(1-z_{u,v})  &= 0 \\
            &\forall v \in A & \sum_{u \in B} z_{u,v} &\le \gamma  
        \end{aligned}
    \right\}.
\end{align*}

Here, the additional variables $z_{u,v}$ encode which edges to ``add back'' into the graph in order to obtain a clique.
The constraint $\sum_{u \in A} z_{u,v} \le \gamma$ ensures that only up to $\gamma$ such edges can be added incident to every $u \in A$.

\subsection{SoS Certificates in balanced bipartite graphs}

The core of the analysis of all our algorithms is an SoS-certificate for the absence of bicliques in balanced bipartite graphs. 
Before formally describing our certificates, we introduce some notation. Recall that given a bipartite graph $G = A \uplus B$, we denote variables for vertices $v\in A$ by $x_v$ and variables for vertices $v\in B$ by $y_v$. We further denote by $x \in \mathbb{R}^{|A|}$ the vector of all variables $x_v$ for $v \in A$ and by $y \in \mathbb{R}^{|B|}$ the vector of all variables $y_v$ for $v \in B$. Given a vector $x$ of variables, we denote by $|x|$ its $L_1$-norm. Moreover, for a set $S \subseteq [n]$, we use the abbreviation $x_S \coloneqq \prod_{i \in S} x_i$.

\subsubsection{Core certificate}

With this, we introduce the main lemma upon which our biclique certificates are built. It can be seen as a modified version of \cite[Lemma 4.11]{Kothari-STOC-2023}. The difference is that we have different balancedness guarantees since we only get meaningful balancedness if $|S \triangle R| \ge 3$ instead of $|S \triangle R| \ge 1$ as in \cite{Kothari-STOC-2023}. Moreover, when using $\mathcal{R}(H, \gamma)$ for non-zero $\gamma$, we further have to account for the fact that some edges within a ``clique'' can be missing. These factors require us to modify the proof of \cite[Lemma 4.11]{Kothari-STOC-2023} and result in a slightly different SoS-certificate than in previous work. The following is the core certificate we use for exact recovery and our refutation algorithms.
\begin{lemma}\label{lem:core-sos-proof}
    Let $r \in \mathbb{N}, r \ge 2$ and $H = (A \uplus B, E)$ be a bipartite graph which has $r$-wise balancedness. Recall that variables associated to $v \in A, u \in B$ are denoted by $x_v$ and $y_u$, respectively. Then, 
    \begin{align*}
        \mathcal{R}(H) \sststile{x, y}{4r} \left\{ \begin{aligned} &\maxbalall{r} |y|  \left( \sum_{S \subseteq A, |S| = r} x_S \right)^2 \\
        & \hspace{2cm} \le \left(n + rn^2\right) \maxbal{r} \left( \sum_{S \subseteq A, |S| = r} x_S \right) + \Delta \left( \sum_{S \subseteq A, |S| = r} x_S \right)^2 \end{aligned} \right\}.
    \end{align*}
    The bit complexity of the proof is $n^{O(r)}$. 
\end{lemma}

\begin{remark}[Intuition]
    To get a sense for what the above lemma tells us, let us for now assume $p=1/2$, and let us assume that $\sum_{S \subseteq A, |S| = r} x_S \approx |x|^r$ for $|x| = \sum_{v \in A} x_v$, which is approximately true by \Cref{lem:sos-factorial}. Then, re-arranging the above, we get a statement of the form $|x|^{2r} ( |y| - \Delta ) \le (n + rn^2) |x|^r$. While for $|y| < \Delta$, this statement is trivial, it is quite powerful whenever $|y| \gtrsim \Delta$, i.e., if the ``right side'' of our biclique is larger than our balancedness $\Delta$. In this case, the statement reduces to something like $|x|^r \le n + rn^2$, which we can see as a certificate that the ``left side'' of our biclique is of size at most $O(n^{2/r})$.
\end{remark}

To prove \Cref{lem:core-sos-proof}, we start with the following lemma that allows us to arrive at the left hand side of the certificate in \Cref{lem:core-sos-proof}.
\begin{lemma}\label{lem:lhs-cert}
    Let $r \in \mathbb{N}, r \ge 2$ and $H = (A \uplus B, E)$ be a bipartite graph.
    Recall the definition of
    $H_p(u, v)$ and $u_{S, p}$ from \Cref{def:balancedness}. Then for any $S, R \subseteq A$ with $|S| = |R| = r$, 
    $$
         \mathcal{R}(H) \sststile{x, y}{4r} \left\{ x_{S\cup R} \sum_{v \in B} y_v u_{S, p}(v)u_{R, p}(v) = x_{S\cup R} \maxbalall{r} |y| \right\}.
    $$ 
\end{lemma}
\begin{proof}
    We note that the constraint $x_uy_v = 0$ for all $\{u, v\} \notin E$ implies that for any $v \in B$,
    $$
        \mathcal{R}(H) \sststile{x, y}{4r} \left\{ x_{S \cup R} y_v u_{S, p}(v) u_{R, p}(v) = \maxbalall{r} x_{S \cup R} y_v \right \}
    $$
    since $u_{S, p}(v) u_{R, p}(v) = \maxbalall{r}$ whenever all edges between $v$ and $S \cup R$ are present in $H$. Summing over $v \in B$ yields the statement.
\end{proof}

With this, we are ready to prove \Cref{lem:core-sos-proof}.

\begin{proof}[Proof of \Cref{lem:core-sos-proof}]
    Recall the definition of
    $u_{S, p}$ from \Cref{def:balancedness}. For any $v \in B$, we now define $u'_{S, p}(v) \coloneqq u_{S, p}(v)(1 - y_v)$. Using \Cref{lem:lhs-cert} and further that $\mathcal{R}(H) \sststile{x,y}{4} \{(1 - y_v)^2 = 1 - y_v\}$ for any $v \in B$, we get that for all $S, R \subseteq A$ with $|S| = |R| = r$, \begin{align*}
        \mathcal{R}(H) \sststile{x, y}{4r} \left\{ \begin{aligned} x_{S \cup R} \sum_{v \in B} u'_{S, p}(v) u'_{R, p}(v) &= x_{S \cup R} \sum_{v \in B} u_{S, p}(v)u_{R, p}(v)(1-y_v)^2 \\
        &\le x_{S \cup R} \sum_{v \in B} u_{S, p}(v) u_{R, p}(v) - \maxbalall{r} x_{S \cup R}|y| \end{aligned} \right\}.
    \end{align*} By our balancedness assumption, we get that for any $S, R \subseteq A$ with $|S| = |R| = r$ such that $|S \triangle R| \ge 3$, we have $ \sum_{v \in B} u_{S, p}(v) u_{R, p}(v) \le \Delta$. Accordingly, for such $S, R$, we have \begin{align*}
        \mathcal{R}(H) \sststile{x, y}{4r} \left\{  x_{S \cup R} \langle u'_{S, p}, u'_{R, p} \rangle \le \Delta x_{S \cup R} - \maxbalall{r} | y | x_{S \cup R} \right\},
    \end{align*} Otherwise (if $|S \triangle R| < 3$), we use the trivial bound
    $$
    \mathcal{R}(H) \sststile{x, y}{4r} \left\{  x_{S \cup R} \langle u_{S, p}, u_{R, p} \rangle \le n \maxbal{r} x_{S \cup R} \right\}.
    $$
    With this, 
    \begin{align*}
        \mathcal{R}(H) \sststile{x, y}{4r + 2} \left\{ 
            \begin{aligned}0 &\le \left\| \sum_{|S| = r} x_Su'_{S,p} \right\|_2^2 
            =  \sum_{\substack{S, R \subseteq A, |S|,|R| = r}} x_{S \cup R} \langle u'_{S, p}, u'_{R, p} \rangle \\ 
            &\hspace{.1cm} = \sum_{ \substack{S, R \subseteq A, |S|,|R| = r \\ |S \triangle R| \le 2 } } x_{S \cup R} \langle u'_{S, p}, u'_{R, p} \rangle + \sum_{ \substack{ S, R \subseteq A, |S|,|R| = r \\ |S \triangle R| \ge 3 } } x_{S \cup R} \langle u'_{S, p}, u'_{R, p} \rangle \\
            &\hspace{.1cm}= n \maxbal{r} \left(\sum_{ \substack{S, R \subseteq A , |S|,|R| = r \\ |S \triangle R| \le 2 } } x_{S \cup R} \right) - |y| \maxbalall{r} \left(\sum_{ \substack{ S, R \subseteq A, |S|,|R| = r \\ |S \triangle R| \le 2 } } x_{S \cup R} \right) \\
            &\hspace{2cm} + \Delta \left(\sum_{ \substack{ S, T \subseteq A, |S|,|R| = r \\ |S \triangle R| \ge 3 } } x_{S \cup R} \right) -  |y| \maxbalall{r} \left(\sum_{ \substack{ S, R \subseteq A, |S|,|R| = r \\ |S \triangle R| \ge 3 } } x_{S \cup R} \right) \\
            &\hspace{.1cm}= n \maxbal{r} \left(\sum_{ \substack{S, R \subseteq A, |S|,|R| = r \\ |S \triangle R| \le 2 } } x_{S \cup R} \right) +  \triangle \left(\sum_{ \substack{ S, T \subseteq A, |S|,|R| = r \\ |S \triangle R| \ge 3 } } x_{S \cup R} \right) \\
            & \hspace{6cm} - |y| \maxbalall{r} \left(\sum_{ \substack{ S, R \subseteq A, |S|,|R| = r } } x_{S \cup R} \right) \\
            &\hspace{.1cm}\le n \maxbal{r} \left(\sum_{ \substack{S, R \subseteq A, |S|,|R| = r \\ |S \triangle R| \le 2} } x_{S \cup R} \right) +  \triangle \left(\sum_{ \substack{ S \subseteq A, |S| = r } } x_{S} \right)^2 \\
            &\hspace{6cm}- |y| \maxbalall{r} \left(\sum_{ \substack{ S \subseteq A, |S| = r } } x_{S} \right)^2 
            \end{aligned}
        \right\}
    \end{align*}
    Given a fixed $R \subseteq A$ with $|R| = r$, we further get that $\mathcal{R}(H) \sststile{x,y}{2r} \left\{ x_R \le 1 \right\}$, so for any fixed $S \subseteq A$ with $|S| = r$ \begin{align*}
        \mathcal{R}(H) \sststile{x,y}{4r} \left\{ \sum_{ \substack{R \subseteq A, |R| = r \\ |S \triangle T| \le 2}}  x_R \le 1 + rn \right\}
    \end{align*} since the number of sets $R$ with the desired properties is at most $1 + rn$. Multiplying this with $x_S$ and summing yields that 
    $$
        \mathcal{R}(H) \sststile{x,y}{4r} \left\{ n \maxbal{r} \left( \sum_{ \substack{S \subseteq A, |S|,|R| = r \\ |S\triangle R| \le 2} } x_{S \cup R} \right) \le (n + rn^2) \maxbal{r} \left(\sum_{ \substack{ S \subseteq A, |S| = r } } x_{S} \right)  \right\},
    $$ so in total,
    \begin{align*}
        \mathcal{R}(H) \sststile{x, y}{4r} \left\{ \begin{aligned}
        0 &\le (n + rn^2) \maxbal{r} \left(\sum_{ S \subseteq A, |S| = r} x_{S } \right) + \Delta \left(\sum_{ \substack{ S \subseteq A,  |S| = r } } x_{S} \right)^2 \\
        &\hspace{4cm} - |y| \maxbalall{r} \left(\sum_{ \substack{ S \subseteq A,  |S| = r } } x_{S} \right)^2
        \end{aligned}
        \right\},
    \end{align*} as desired.
\end{proof}

\subsubsection{Core certificate for recovery against a bounded adversary}

At the cost of a higher degree, we can get a similar statement for the axioms $\mathcal{R}(H, \gamma)$ that we use for approximate recovery against a bounded adversary.
\begin{lemma}\label{lem:core-sos-proof-adv}
    Let $r \in \mathbb{N}, r \ge 2$ and $H = (A \uplus B, E)$ be a bipartite graph which has $r$-wise balancedness. Recall that variables associated to $v \in A, u \in B$ are denoted by $x_v$ and $y_u$, respectively. Then, 
    \begin{align*}
        \mathcal{R}(H, \gamma) \sststile{x, y}{8r} \left\{ \begin{aligned} &\maxbalall{r} (|y| - 2 \gamma r) \ \left( \sum_{S \subseteq A, |S| = r} x_S \right)^2 \\
        & \hspace{2cm} \le \left(n + rn^2\right) \maxbal{r} \left( \sum_{S \subseteq A, |S| = r} x_S \right) + \Delta \left( \sum_{S \subseteq A, |S| = r} x_S \right)^2 \end{aligned} \right\}.
    \end{align*}
    The bit complexity of the proof is $n^{O(r)}$. 
\end{lemma}

In order to establish the stronger version \Cref{lem:lhs-cert} of \Cref{lem:core-sos-proof-adv}, we use the following modification of \Cref{lem:core-sos-proof-adv} that incorporates the effect of the additional variables $z_{u,v}$ in $\mathcal{R}(H, \gamma)$ and intuitively states that SoS ``knows'' that the balancedness properties of a grah can only change by a limited amount if we restrict the number of edges being added ``back in'' incident to every vertex.

\begin{lemma}\label{lem:lhs-cert-adv}
    Let $r \in \mathbb{N}, r \ge 2$ and $H = (A \uplus B, E)$ be a bipartite graph. Recall the definition of
    $H_p(u, v)$ from \Cref{def:balancedness}. For any $S \subseteq A$ and $v \in B$, define
    $$
        u_{S, p}^{(z)}(v) \coloneqq \prod_{u \in S}  (1 - z_{u,v}) H_p(u, v).
    $$
    Then for any $S, R \subseteq A$ with $|S| = |R| = r$, 
    $$
         \mathcal{R}(H, \gamma) \sststile{x, y}{8r + 2} \left\{ x_{S\cup R} \sum_{v \in B} y_v u_{S, p}^{(z)}(v)u_{R, p}^{(z)}(v) \ge x_{S\cup R} \maxbalall{r} (|y| - 2\gamma r  ) \right\}.
    $$ 
\end{lemma}
\begin{proof}
    Due to the constraint $\{ x_u y_v(1 - z_{u, v}) = 0 \}$ for all $(u, v) \in A \times B$ such that $\{u, v\} \notin E$, we get that
    \begin{align*}
        \mathcal{R}(H, \gamma) \sststile{x, y}{8r} \left\{
        \begin{aligned}
        x_{S\cup R}y_v u_{S, p}^{(z)}(v) u_{R, p}^{(z)}(v) = \maxbalall{r} x_{S\cup R} y_v \prod_{u \in S} (1 - z_{u, v}) \prod_{u \in R} (1 - z_{u, v})  
        \end{aligned}
        \right \}.
    \end{align*}
    We wish to find a lower bound on the right hand side. We show how to obtain a such bound for $\prod_{u \in S} (1 - z_{u, v})$ in the following claim 

    \begin{claim}\label{clm:productsumlowerbound}
        For any set of vertices $M$,
        $$
            \{z_{u,v} = z_{u,v}^2\}_{u, v \in [n]} \sststile{2r}{z} \left\{  \prod_{u \in M} (1 - z_{u, v}) \ge 1 - \sum_{u \in M} z_{u,v} \right\}
        $$
    \end{claim}
    \begin{proof}
        Fix an arbitrary ordering $u_1, \ldots, u_r$ of the vertices in $M$ and note that
    $$
        \prod_{u \in M} (1 - z_{u, v}) = \prod_{i=1}^r (1 - z_{u_i, v}) = \prod_{i=1}^{r-1}(1 - z_{u_i, v}) - z_{u_1,v}\prod_{i=1}^{r-1}(1 - z_{u_i, v}).
    $$ Since $\{z_{u,v} = z_{u,v}^2\}_{u, v \in [n]} \sststile{2}{z} \{z_{u,v} > 0\}$, we get 
    $$
        \{z_{u,v} = z_{u,v}^2\}_{u, v \in [n]} \sststile{2r}{z} \left\{ \prod_{u \in M} (1 - z_{u, v}) \ge \prod_{i=1}^{r-1}(1 - z_{u_i, v}) - z_{u_1,v}  \right\}.
    $$ Iterating this argument yields the claim.
    \end{proof}

    With the above, we get
    $$
        \{z_{u,v} = z_{u,v}^2\}_{u, v \in [n]} \sststile{2r}{z} \left\{  \prod_{u \in S} (1 - z_{u, v}) \prod_{u \in R} (1 - z_{u, v}) \ge 1 - \sum_{u \in S} z_{u,v} - \sum_{u \in R} z_{u,v}  \right\}
    $$
    Thus, using the same argument for $R$, the axiom $\sum_{v \in B} z_{u,v} \le \gamma$, and summing over $v \in B$,
    $$
            \mathcal{R}(H, \gamma) \sststile{z}{4r} \left\{ \sum_{v \in B} y_v \prod_{u \in S} (1 - z_{u, v}) \prod_{u \in R} (1 - z_{u, v}) \ge |y| - 2\gamma r \right\}.
        $$ Multiplying by $x_{S \cup R}$ yields the desired statement.
\end{proof}

\begin{proof}[Proof of \Cref{lem:core-sos-proof-adv}]
    The proof is analogous to the proof of \Cref{lem:core-sos-proof}. We simply replace $u_{S, p}'(v)$ by ${u'}_{S,p}^{(z)}(v)\coloneqq u_{S, p}^{(z)}(v)(1 - y_v)$ (cf. \Cref{lem:lhs-cert-adv} for a definition of $u_{S, p}^{(z)}$). Furthermore, we replace \Cref{lem:lhs-cert} by \Cref{lem:lhs-cert-adv}.
\end{proof}

\section{Proof of identifiability and refutation: the single label clique theorem with and without SoS}

The goal of this section is to prove \Cref{thm:slct}.
\slct*

\begin{remark}[The parameter $\alpha$]
    Above, and throughout the rest of this paper, the constant $\alpha$ is to be understood as controlling our lower bound on $d$ in the sense that we always assume $d \gg n^{\alpha}$ (see also the discussion in \Cref{sec:parameterassumptions}). However, in this section, we shall also use $\alpha$ to refer to the quantity $\log(d)/\log(n)$. Concretely, instead of writing $n^{\alpha} \ll d \ll n^{1-\alpha}$ and $n^{\alpha} \ll d \ll n^{2-\alpha}$, we shall simply write $\alpha \in (0,1)$ and $\alpha \in (0,2)$, respectively. While this is somewhat of an abuse of notation, we think that it greatly simplifies readability and further makes our work easier to be compared with previous results where this notaion is common as well.
\end{remark}

The proof of \Cref{thm:slct} is split into two parts. First, in \Cref{sec:slctalpha02}, we adapt the proof of Spirakis et al. \cite{Spirakis-MFCS-20212} which applies to the case of $\alpha \in (0,2)$ but needs some more restrictive condition on $p, q$. To address this issue (at least in the case $\alpha \in (0,1)$), we present a novel proof in \Cref{lem:slctimproved} based on the notions of balancedness as developed in \Cref{sec:balancedness}. This is one of our main contributions and constitutes precisely the ``proof of identifiability'' described in \Cref{sec:techniques}. This proof is phrased entirely in constant-degree SoS and the ideas introduced there further form the foundations of our algorithms.

To see that the cases considered in \Cref{sec:slctalpha02} and \Cref{lem:slctimproved} indeed cover all cases needed for \Cref{thm:slct}, we provide a short proof. 
\begin{proof}[Proof of \Cref{thm:slct}]
    We let $\varepsilon > 0$ be a sufficiently small constant and assume that the parameters are such that $k \ge n^{\varepsilon}$. If $\alpha \in (0,1)$, $p \ge \varepsilon$, and $p - q \ge \varepsilon$, we apply \Cref{thm:strongsinglelabelclique}. If $\alpha \in (0,1)$ and one of these conditions on $p, q$ is not met, it follows that $p - q \le \varepsilon$ since $p - q \le p$. Then, \Cref{thm:singlelabelclique} is applicable provided $\varepsilon > 0$ is small enough. This finishes case (i). To handle the assumptions in (ii), we can directly apply \Cref{thm:singlelabelclique}.
\end{proof}

\subsection{Identifiability for $\alpha \in (0,2)$}\label{sec:slctalpha02}

Throughout this entire section, we asssume that either $\assumptionoldslct$ or $p - q = o(1)$.
The goal of this section is to prove the following.

\begin{theorem}\label{thm:singlelabelclique}
    Let $G \sim \RIGone$ and assume that $n^{\varepsilon} \ll d  \ll n^{2-\varepsilon}$ such that either $p-q = o(1)$ or $\assumptionoldslct$. Recall that $k \coloneqq \delta n$. Then with high probability, the following holds. For any sufficiently small $\varepsilon > 0$, every clique $K \subseteq G$ of size $\ge (1 - \varepsilon)k$ is spanned by a single label. That is, for every such $K$, there is some $\ell \in [d]$ such that $K \subseteq S_\ell$.
\end{theorem}

For the proof, we recall from \Cref{sec:cliqueintersections} that the event $\mathcal{E}_{a,b}$ occurs if and only if there is some $\ell \in [d]$ and sets $S \subseteq S_\ell, T \in [n] \setminus S_\ell$ with $|S| = a$ and $|T| = b$ such that $G[S \uplus T]$ is a biclique (\Cref{def:Eab}). Due to \Cref{lem:bad-event}, $\mathcal{E}_{a,b}$ is unlikely for suitable $a \ge \frac{n^\varepsilon k}{\sqrt{d}}$ and $b = \lfloor \log(n)^2 \rfloor$. A simple corollary of this fact is that every clique $K$ that is not spanned by a single label can contain at most $k^\gamma$ vertices of the same label for some constant $\gamma < 1$.

\begin{corollary}\label{cor:many-labels}
    With high probability over the draw of $G \sim \RIGone$, the following holds.
    Every clique $K$ of size $(1-\varepsilon)k$ in $G$ is either formed by a single label, or for every label $\ell \in [d]$, there are at most $k^\gamma$ vertices in $K$ that have chosen $\ell$, where $0 < \gamma < 1$ is a constant. 
\end{corollary}
\begin{proof}
    Consider a clique $K$ where $|K \cap S_\ell| = |K| - t$ for $1 \leq t\leq\lfloor \log(n)^2 \rfloor$. Then any $v \in K \setminus S_\ell$ has degree at least $|K| - t$ into $S_\ell$, which can be ruled out using \Cref{lem:degreebound} stating that this degree is at most $(1+o(1))pk$. Given sufficiently small $\varepsilon > 0$, this is a contradiction.

    Next, consider the case that $t > \lfloor \log(n)^2 \rfloor$. Then, \Cref{lem:bad-event} immediately implies that for every $\ell$, we have $|K \cap S_\ell| \le a$ as otherwise, this contradicts the occurence of $\mathcal{E}_{a,b}$. Since by \Cref{lem:bad-event}, we can choose $a = n^{\varepsilon}k/\sqrt{d} \le k^{\gamma}$ for some $\gamma < 1$, the statement follows. \qedhere
    
\end{proof}

We proceed by ruling out the case that $K$ is not spanned by a single label using the conclusion of the above lemma by showing that it implies the existence of the following structure in $K$. 

\begin{definition}[Duplicate-free sets]\label{def:dup-free} 
    Given a $G \sim \RIGone$, we call a set of verices $M \subseteq [n]$ \emph{duplicate-free} if there is an assignment $g: \binom{M}{2} \rightarrow [d] \cup \{\emptyset\}$ with the following properties. \begin{enumerate}
        \item For every $\ell \in [d]$, there is at most one $\{u, v\} \in \binom{M}{2}$ such that $g(\{u, v\}) = \ell$.
        \item For every $\{u, v\} \in \binom{M}{2}$ such that $g(\{u, v\}) \in [d]$, we have that $g(\{u,v\}) \in M_u \cap M_v$.
    \end{enumerate}
\end{definition}
Intuitively, the function $g$ maps every $\{u, v\} \in \binom{M}{2}$ to either $\emptyset$ or a label $\ell$ that both $u$ and $v$ have in common, while assigning each label $\ell \in [d]$ at most once. We show that every clique not formed by a single label contains a decently large duplicate-free set. 

\begin{lemma}
    Let $\varepsilon>0$ be a sufficiently small constant. Then, with high probability over the draw of $G \sim \RIGone$, it holds for every clique $K$ of size $\ge (1-\varepsilon)k$ not formed by a single label, that $K$ contains a duplicate-free set of size $k^c$ where $c > 0$ is a constant. 
\end{lemma}\begin{proof}
    We make use of the following:
    \begin{claim}\label{claim:two-items} The following event occurs with high probability over the draw of $G \sim \RIGone$.
        For every pair $\{u, v\} \in \binom{[n]}{2}$, there are at most $C \log(n)$ labels in $M_u \cap M_v$.
    \end{claim}

    Before we prove our claim, we show how this implies the existence of a $n^{\Omega(1)}$-sized duplicate-free set in $K$. To this end, consider the following simple algorithm for constructing such a set $Q \subseteq K$. \begin{enumerate}
        \setlength\itemsep{.0001em}
        \item Initialize $Q_0 = \emptyset$ and set $i = 0$.
        \item While $K$ is not empty, choose an arbitrary vertex $v \in K$, and remove from $K$ all vertices that have chosen a label in $\bigcup_{u \in Q_i} (M_u \cap M_v)$. For each $u\in Q_i$, set $g(\{u, v\})$ to an arbitrary label in $M_u \cap M_v$ if nonempty, and $g(\{u, v\}) = \emptyset$ otherwise. Finally, set $Q_{i + 1} = Q_i \cup \{v\}$ and update $i \leftarrow i + 1$.
        \item Output $Q_i$.
    \end{enumerate}
    We point out, that the above algorithm maintains the following invariant: at the beginning of each iteration described in step 2, it holds that for all $u, u' \in Q$ and all $v \in K$, we have $M_u \cap M_{u'} \cap M_v = \emptyset$. This follows since all vertices that have chosen a label in $M_u \cap M_{u'}$ were removed from $K$ in a previous iteration, and it implies that the set $Q$ together with the assignment $g$ is duplicate-free. 

    Moreover, using that with high probability $|M_u \cap M_v| \le C \log(n)$ for all $u,v$ (\Cref{claim:two-items}) in conjunction that for every label $\ell$, the number of vertices in $K$ that have chosen $\ell$ is at most $k^\gamma$ (\Cref{cor:many-labels}), the number of vertices removed in iteration $i$ is at most $C|Q_i|k^\gamma \log(n)$ with high probability. Since $|Q_i| = i$, after $j$ iterations, we have removed at most $$
        Ck^\gamma \log(n) \sum_{i = 1}^{j} |Q_i| = Ck^\gamma \log(n) \frac{j(j+1)}{2} 
    $$ vertices from the initial set $K$. While $j \le k^{\frac{1 - \gamma}{2} -\varepsilon}$, this is smaller than $(1-\varepsilon)k$, so we can run the above algorithm for at least $k^{c}$ iterations, where $c > 0$ is a constant. All that remains for proving the statement is the following.

\noindent\textit{Proof of \Cref{claim:two-items}.}
    Fix a pair of vertices $u,v$ and fix a single label $\ell$. Then reveal the label $\ell$ for the two vertices. The probability that both vertices share the fixed label is $\delta^2 = O(1/d)$ by \Cref{lem:sizeofdelta}. Hence, the expected number of shared labels among the fixed pair is $\mathbb{E}[|M_u \cap M_v|] = O(1)$ and a Chernoff bound yields that there exists a constant $C$, such that the probability that $|M_u \cap M_v|$ exceeds $C\log(n)$ is at most $n^{-3}$. Taking a union bound over all $O(n^2)$ pairs of vertices the desired claim  follows.
\end{proof}
Finally, we prove that duplicate-free sets of size $k^{c}$ do not exist with high probability.

\begin{lemma}\label{lem:noduplicatfreesets}
    For every constant $c > 0$, the following holds with high probability over the draw of $G \sim \RIGone$. No set $K \subseteq [n]$ of size at least $k^c$ is both duplicate-free and a clique.  
\end{lemma}\begin{proof}
    The lemma follows because it is very unlikely that a duplicate-free set contains many edges. To show this, let $M$ be a subset of $[n]$ such that $|M| \ge k^c$. We bound the probability that $M$ is duplicate-free. To this end, we consider all possible mappings $g: \binom{M}{2} \rightarrow [d] \cup \{\emptyset\}$ consistent with the constraints in \Cref{def:dup-free}. To enumerate them, we denote by $h = h(g)$ the number of pairs in $\binom{M}{2}$ that are mapped to $\emptyset$ by a given $h(g)$. For a fixed $h$ and $M \subseteq [n]$, the number of mappings $g$ conformal with \Cref{def:dup-free} and $h(g) = h$  are at most 
    
    \begin{align}\label{eq:spirakis}
        \binom{\binom{|M|}{2}}{h} \ d^{\binom{|M|}{2} - h}.
    \end{align} 
    
    The probability that there is a duplicate free set of size $|M|$ is thus (by the union bound) at most 
$$
         S \coloneqq n^{|M|} \sum_{h = 0}^{\binom{|M|}{2}} \binom{\binom{|M|}{2}}{h}  d^{\binom{|M|}{2} - h} \delta^{2 \left( \binom{|M|}{2} - h \right)} q^{h} = n^{|M|} \sum_{h = 0}^{\binom{|M|}{2}} \binom{\binom{|M|}{2}}{h}  (d \delta^2)^{\binom{|M|}{2} - h} q^{h},$$
where we used that $q^h$ is the probability that there are $h$ edges due to "noise". To show that $S$ is small, we now distinguish the two parameter regimes we assumed in the beginning.

        \textbf{Case 1:} $p - q \le \varepsilon$. Using symmetry of the binomial coefficient, 
        $$
            S  =n^{|M|} \sum_{h = 0}^{\binom{|M|}{2}} \binom{\binom{|M|}{2}}{h}  (d \delta^2)^{\binom{|M|}{2} - h} q^{h} = n^{|M|} \sum_{h = 0}^{\binom{|M|}{2}} \binom{\binom{|M|}{2}}{h}  (d \delta^2)^{h} q^{ \binom{|M|}{2} - h }.
        $$ We now split the sum into two parts, based on $h$. To this end, we let $\mu$ be a constant in $(0,1/2)$ to be fixed later, we define $m \coloneqq \binom{|M|}{2}$, and we split $$
            S = \underbrace{n^{|M|} \sum_{h = 0}^{\mu m } \binom{m}{h}  (d \delta^2)^{h} q^{m - h }}_{\eqqcolon S \downarrow} + \underbrace{ n^{|M|} \sum_{h =\mu m}^{m} \binom{m}{h}  (d \delta^2)^{h} q^{ m - h }}_{\eqqcolon S \uparrow}.
        $$
        Now, the crucial observation is that by \Cref{lem:sizeofdelta}, 
        $
            d \delta^2 \le (1+ o(1)) \log\big( \frac{1-q}{1-p} \big)
        $ and $\log\big( \frac{1-q}{1-p} \big)$ tends to zero as $p - q \rightarrow 0$. Hence, we can choose $\varepsilon$ small enough to ensure $\delta^2d \le \tilde{q} \coloneqq \max\{q, 1/3\}$. Then, the first sum can be bounded as 
        \begin{align*}
            S \!\! \downarrow \ \ \le n^{|M|} \tilde{q}^{m} \sum_{h = 0}^{\mu m} \binom{m}{h} \le n^{|M|}\tilde{q}^{m} \mu m \binom{m}{\mu m } &\le n^{|M|} \tilde{q}^{ m } m \frac{1}{(\mu m)!} m^{\mu m}\\
            &\le n^{|M|} \tilde{q}^{ m } m \left( \frac{2e m }{\mu m} \right)^{\mu m} = n^{|M|} m \left( \tilde{q} \left( \frac{2e }{\mu} \right)^{\mu} \right)^m.
        \end{align*} Since $q$ is bounded away from $1$, we can choose $\mu$ small enough to make the base $\tilde{q} \left( \frac{2e }{\mu} \right)^{\mu}$ strictly smaller than one such that $$
            S\!\!\downarrow \ \ \le n^{|M|} \binom{|M|}{2}^2 e^{ - \Omega\left(\binom{|M|}{2}\right)} \le n^{-\omega(1)}.
        $$
        \begin{align*}
            S \!\! \uparrow \ \ \le n^{|M|} (d\delta^2)^{\mu m} \sum_{h = 0}^m  \binom{m}{h} \le n^{|M|} 2^{m}  (d\delta^2)^{\mu m}.
        \end{align*}

        Choosing $\varepsilon$ small enough such that $(d\delta^2)^{\mu m} \le \left(\frac{1}{3} \right)^{m}$ then also yields that $S \!\! \uparrow \ \ \le n^{-\omega(1)}$, as desired.

        \textbf{Case 2:} $\assumptionoldslct$. 
        In this case, it follows that 
        $$
        n^{|M|} \sum_{h = 0}^{\binom{|M|}{2}} \binom{\binom{|M|}{2}}{h}  \left( (1 + o(1)) \log\left( \frac{1-q}{1-p} \right)  \right)^{\binom{|M|}{2} - h} q^{h} \le n^{|M|} \left( \frac{1}{2}- \varepsilon\right)^{\binom{|M|}{2}} \sum_{h = 0}^{\binom{|M|}{2}} \binom{\binom{|M|}{2}}{h}. 
        $$
 Finally, using that $\sum_{h = 0}^{\binom{|M|}{2}} \binom{\binom{|M|}{2}}{h} = 2^{\binom{|M|}{2}}$ yields
        \begin{align*}
        n^{|M|} \left( \frac{1}{2}- \varepsilon\right)^{\binom{|M|}{2}} \sum_{h = 0}^{\binom{|M|}{2}} \binom{\binom{|M|}{2}}{h} = n^{|M|} \left( \frac{1}{2}- \varepsilon\right)^{\binom{|M|}{2}}2^{\binom{|M|}{2}},
    \end{align*} 
    which is $n^{-\omega(1)}$ since $|M| \geq k^c = n^{\Omega(1)}$.
\end{proof}

\subsection{Identifiability via SoS for $\alpha \in (0,1)$}\label{lem:slctimproved}

While the previous subsection needs the assumption that either $p-q = o(1)$ or $\assumptionoldslct$, this leaves a gap in the dense case where $p$ is a sufficiently large constant. In this section, we close this gap for the case $\alpha \in (0,1)$ using an entirely different argument that relies on the notion of balancedness. Throughout, we therefore assume that $p$ and $p-q$ are both constant while $\alpha \in (0,1)$. Note that this implies in particular that $n^{\frac{1}{2} + \varepsilon} \ll k \ll n^{1-\varepsilon}$. 

With this, we show the following theorem, that together with the previous section covers all relevant cases assuming only $\alpha \in (0,1)$ and $n^{\varepsilon} \ll k \ll n^{1-\varepsilon}$.

\begin{theorem}[Strong Single Label Clique Theorem for $\alpha \in (0, 1)$]\label{thm:strongsinglelabelclique}
    Let $G \sim \RIGone$ and assume that $d = \Theta(n^{\alpha})$ for some fixed $\alpha \in (0,1)$ while both $p$ and $p - q$ are constants. Then for any sufficiently small $\varepsilon > 0$, every clique $K \subseteq G$ of size $\ge (1 - \varepsilon)k$ is spanned by a single label. That is, for every such $K$, there is some $\ell \in [d]$ such that $K \subseteq S_\ell$.
\end{theorem}

The high-level approach is as follows. Assume that in a $G \sim \RIGone$, there were a clique $K$ of size $(1-\varepsilon)k$ that is not spanned by a single label. Now, fix a constant $t$ and consider an arbitrary $t$-tuple $T \in \binom{K}{t}$ and note that in the neighbourhood $N_G(T)$ of $T$ in $G$, our clique $K$ is preserved. Hence, it suffices to rule out that there is some $T \in \tbinom{[n]}{t}$ such that a clique of size $(1-\varepsilon)k -t$ exists in $N_G(T)$.

Given a fixed $T$, we can use a balancedness-based argument to show that there is only a small number of ground-truth cliques in $\neigh{T}$ that $K$ can concentrate on. These cliques are precisely those that correspond to the labels in $\cliqueset{T}$, the set of labels that appear in at least $3$ of the vertices in $T$. We denote the set of vertices corresponding to the cliques in $\cliqueset{T}$ by $V(\cliqueset{T})$ and show that only $o(k)$ vertices of $K$ are outside of $\cliqueset{T}$. Given this, we can derive the stronger statement that $K$ can only intersect one of them by a significant amount using the balancedness \emph{between} the cliques in $\cliqueset{T}$. If we assume that at least one vertex of $K$ is outside of every ground-truth clique $S_\ell$, this implies that the size of $k$ must be $o(k)$ which yields our desired conclusion

Remarkably, all of this can be done \emph{within constant-degree SoS}, and the final proof of \Cref{thm:strongsinglelabelclique} can be phrased as a \emph{SoS-refutation}. This further yields an efficient refutation algorithm, which might be of independent interest.

We split the proof of \Cref{thm:strongsinglelabelclique} into three main parts: first, we show that all but $o(k)$ vertices of a clique containing $T$ concentrate on the vertices in $V(\cliqueset{T})$, afterwards, we show that among the cliques in $\cliqueset{T}$, the clique $K$ does not significantly intersect more than one of them. Finally, we finish the proof of \Cref{thm:strongsinglelabelclique} by constructing a set of axioms that encodes a clique of size $(1-\varepsilon)k$ that it is \emph{not} contained in any ground-truth clique, and showing how it gives rise to a low-degree SoS-proof of the inequality  ``$-1 > 0$'': a contradiction. 
 
\subsubsection{Step 1: Few vertices outside $V(\cliqueset{T})$}\label{sec:fewoutsidevcliqueset}
We start by formalizing the idea that a clique containing $T$ can only place $o(k)$ vertices outside of $V(\cliqueset{T})$.

\begin{lemma}\label{lem:concentrationtofewcliques}
    Given a $t$-tuple $T \in \binom{[n]}{t}$, denote by $\cliqueset{T} \subseteq [d]$ the set of labels that occurs in at least $3$ vertices of $T$. Denote further by $V(\cliqueset{T}) = \bigcup_{\ell \in \cliqueset{T}} S_\ell$, i.e., the set of vertices that have chosen a label in $\cliqueset{T}$. Then there are $t, r_0 \in \mathbb{N}$ large enough such that w.h.p. over the draw of $G \sim \RIGone$, the following holds for every $T \in \binom{[n]}{t}$ and $r \ge r_0$ such that $r \le t$ is a power of $2$.
    $$
        \mathcal{A}(G, (1-\varepsilon)k) \sststile{x}{O(r)} \Bigg\{ x_T \left(\sum_{v \in [n] \setminus V(\cliqueset{T})} x_v\right) \le O \left( \frac{n}{k} \left( \frac{n^2}{k} \right)^{1/r} \right)x_T   \Bigg\}.
    $$
\end{lemma}
\begin{remark}[Intuition]
    The above can be interpreted as stating that every clique $K$ in $G$ of size $(1-\varepsilon)k$ such that $T \subseteq K$ is such that there are at most $O(\frac{n}{k}(\frac{n^2}{k})^{1/r})$ vertices of $K$ that are in $[n] \setminus V(\cliqueset{T})$. Further, note in particular that this bound can be made $o(k)$ if choosing $r$ large enough since by assumption on $\alpha, p,$ and $q$ we get that there is some $\varepsilon > 0$ such that $k \ge n^{\frac{1}{2} + \varepsilon}$. Therefore, the conclusion of the lemma can indeed be seen as doing what was promised in the previous paragraph.
\end{remark}

We prove this lemma using a sequence of smaller statements. The general idea is to consider a multipartite sub-graph of $G$ and to exploit its balancedness properties. Precisely, we consider a multi-partition $
    V_1 \uplus V_2 \uplus \cdots \uplus V_m
$ where the sets $V_1, V_2, \ldots, V_m$ are a fixed partition of $[n]$ such that each $V_i$ has size at most $\lfloor \frac{k}{2} \rfloor$ and $m \leq 3n/k$.

\newcommand{\Gifive}{G[A_i \uplus (B_i \cap \neigh{T})]}
\newcommand{\ABifive}{A_i = V_i \setminus V(\cliqueset{T}) \text{ and } B_i = [n] \setminus A_i}

The first step towards establishing \Cref{lem:concentrationtofewcliques} consists in establishing balancedness of the cut given by $\Gifive$ for $\ABifive$ and all $i$, captured in the following claim.

\begin{lemma}\label{lem:multipartite-balancedness}
    Fix any constant $t \in \mathbb{N}$.
    With high probability over the choice of $G \sim \RIGone$, the following holds. For any fixed $r \ge 2$, every $T \in \binom{[n]}{t}$ and every $i \in [m]$, the bipartite graph $G[V_i \uplus ( [n] \setminus V_i) ] \cap N_G(T)$ has $r$-fold balancedness 
    $$
        \frac{16tp^{-2}}{\alpha} \maxbal{r} \cdot  p^{t}k + o(k).
    $$
\end{lemma}\begin{proof}
    We wish to apply \Cref{lem:general-balancedness-generalized} with $L(T) = \cliqueset{T}$, $S(T) = V(\cliqueset{T})$ and $\eta = 2$. With this choice, we get by construction of $L(T)$ that for every $\ell \in [d] \setminus L(T)$, only at most $2$ vertices in $T$ have chosen $\ell$. Moreover, we use the following collection $\{ \textsc{Alg}_i \}_{i \in [m]}$ of procedures for constructing a bipartition. Each $\textsc{Alg}_i(S(T), [n])$ outputs (deterministically) the bipartition $\ABifive$, which is valid in the sense of \Cref{def:procedure-bipartition} since $A_i \cap S(T) = \emptyset$. Hence, \Cref{lem:general-balancedness-generalized} is applicable and yields the desired statement.
\end{proof}

With this, we wish to bound the number of vertices outside of $V(\cliqueset{T})$ given a fixed $T$. For this, we use the SoS-proofs established in \Cref{sec:sos}. 

\begin{proof}[Proof of \Cref{lem:concentrationtofewcliques}]
    The main part of our proof can be phrased as a degree-$O(r)$ SoS proof based on the certificates in \Cref{sec:sos} and the axioms $\mathcal{A}(G, (1-\varepsilon)k)$. 
    
    Consider the cut given by $\Gifive$ and $\ABifive$ for some fixed $T \in \binom{[n]}{t}$ and $i \in [m]$. 
    To simplify notation, let for now $A \coloneqq A_i$ and $B \coloneqq B_i$. Considering any clique $K$ of size $(1-\varepsilon)k$, we denote by $x_u \in \{0,1\}$ the indicator variables indicating whether or not $u \in [n]$ is in $K$. We further define $|y| \coloneqq \sum_{v \in (B \cap \neigh{T})} x_v$. Applying the statement from \Cref{lem:core-sos-proof} to the bipartite graph $G[A \uplus (B \cap \neigh{T})]$, now yields that
    \begin{align}\label{eq:cert}
         \mathcal{A}(G, (1-\varepsilon)k) \sststile{x}{O(t)} \left\{ \left( \sum_{S \subseteq A, |S| = r} x_S \right)^2 \left( \maxbalall{r} |y| - \Delta \right) \le \left(n + rn^2\right) \maxbal{r} \left( \sum_{S \subseteq A, |S| = r} x_S \right) \right\}.
    \end{align}
    Now, recall that $|A_i| \le |V_i(T)|  \le \lfloor \frac{k}{2}\rfloor$. Using this fact, and the axiom $\{\sum_{v \in [n]}x_v \ge (1-\varepsilon)k\}$ we get that $ \mathcal{A}(G, (1-\varepsilon)k) \sststile{x}{2} \{ \sum_{v \in B} x_v \ge (1-\varepsilon)k - \frac{k}{2} \ge \frac{1-2\varepsilon}{2} k \ge \frac{k}{3} \}$. Further, since $ \mathcal{A}(G, (1-\varepsilon)k) \sststile{x}{2t+2} \{ x_T x_v = 0 \}$ for $v \notin \neigh{T}$, we get $$ \mathcal{A}(G, (1-\varepsilon)k) \sststile{x}{O(t)} \left\{  x_T  |y| = x_T\sum_{v \in B \cap \neigh{T}}x_v = x_T \sum_{v \in B} x_v \ge x_T\frac{k}{3} \right\}$$
    
    We choose $t$ as a sufficiently large constant such that $\Delta \le \maxbalallcompact{r} \frac{k}{4}$, which is possible by \Cref{lem:multipartite-balancedness}. Therefore, multiplying \eqref{eq:cert} by $x_T$, we get 
    \begin{align*}
         \mathcal{A}(G, (1-\varepsilon)k) \sststile{x}{O(t)} \left\{ x_T \left( \sum_{S \subseteq A, |S| = r} x_S \right)^2 \maxbalall{r} \frac{k}{12} \le \left(n + rn^2\right) \maxbal{r} x_T \left( \sum_{S \subseteq A, |S| = r} x_S \right) \right\}.
    \end{align*}
    using both cancellations from \Cref{lem:cancellation}, and recalling that  $ \mathcal{A}(G, (1-\varepsilon)k) \sststile{x}{2t+2} \{ x_T^2 = x_T \}$, we get
     $$
         \mathcal{A}(G, (1-\varepsilon)k) \sststile{x}{O(t)} \ \left\{ x_T \sum_{S \subseteq A, |S| = r} x_S  \le \frac{12 rn^2}{k} \left( \maxbal{}/\maxbalall{} \right)^r \right\}.
    $$
    Using further that $\maxbal{}/\maxbalall{} = \max\{1, (p/(1-p))^2 \}$ and applying \Cref{lem:sos-factorial}, we get
    $$
         \mathcal{A}(G, (1-\varepsilon)k) \sststile{x}{O(t)} \left\{ x_T \left( \sum_{v \in A_i} x_v \right)^r \le (2r)^r x_T \left( \frac{12 r n^2}{k}  \max\left\{1, \left(\frac{p}{1-p} \right)^2 \right\}^r + \frac{2r^r}{r!} + \varepsilon \right) \right\}.
    $$ Summing this statement over all $i \in [m]$ and using the SoS almost triangle inequality for any even $r$, (\Cref{lem:sos-almost-triangle}),
    \begin{align*}
         \mathcal{A}(G, (1-\varepsilon)k) \sststile{x}{O(r)} \left\{ \begin{aligned}
             &x_T \ \left( \sum_{v \in [n] \setminus V(\cliqueset{T})} x_v \right)^r = x_T \left(  \sum_{i\in[m]} \sum_{v \in A_i} x_v \right)^r \le \left(\frac{3n}{k}\right)^{r-1}  x_T \left( \sum_{v \in A_i} x_v \right)^r \\
             & \hspace{2cm} \le  (2r)^r \left(\frac{3n}{k}\right)^{r-1} x_T \left( \frac{12 r n^2}{k}  \max\left\{1, \left(\frac{p}{1-p} \right)^2 \right\}^r + \frac{2r^r}{r!} + \varepsilon \right) 
         \end{aligned} \right\},
    \end{align*} where we used that the number of sets $V_i$ is at most $m \leq 3n/k$. Now, we can choose am $r$ to be a sufficiently large power of $2$ and apply the cancellation $\{ f^2 \le C  \} \sststile{x}{2d} \{ f \le \sqrt{C} \}$ from \Cref{lem:cancellation} $\log_2(r)$ times to obtain 
    $$
        \mathcal{A}(G, (1-\varepsilon)k) \sststile{x}{O(r)} \left\{ x_T \sum_{v \in [n] \setminus V(\cliqueset{T})} x_v \le  \frac{6rn}{k} \left( \frac{24n^2}{k} \right)^{1/r} \max\left\{1, \left(\frac{p}{1-p} \right)^2 \right\} \ x_T \right\},
    $$
    as desired.
\end{proof}
This constitutes our first step towards the proof of \Cref{thm:strongsinglelabelclique} since it allows us to restrict our attention to $V(\cliqueset{T})$. We proceed by showing that no clique intersects more than one of the cliques in $\cliqueset{T}$ by a large amount.

\subsubsection{Step 2: Only one clique in $\cliqueset{T}$ has significant overlap with $K$}\label{sec:slctstep2}

Our key lemma here is a certificate analogous to the certificate obtained in \cite[Theorem 4.11]{Kothari-STOC-2023}. Before introducing it, we need to establish some further notation. Given a $G \sim \RIGone$ any integer $t \in \mathbb{N}$ and some $T \in \binom{[n]}{t}$, recall the definition of $\cliqueset{T}$. Define $\overlap{T}$ as the set of vertices in $G$ that have chosen at least two of the labels in $\cliqueset{T}$. For any $\ell \in \cliqueset{T}$, define further $\leftt = S_\ell \setminus \overlap{T}$ and $\rightt = \bigcup_{\ell' \in \cliqueset{T} \setminus \{\ell\}} S_{\ell'} \setminus \overlap{T}$. We will mainly be intereseted in the bipartite graphs $G[ \leftt \uplus \rightt]$. The key lemma proved in this section is the following.

\begin{lemma}\label{lem:limitedoverlapafterT}
    The following holds with high probability over the draw of $G$ and any $t \in \mathbb{N}$ sufficiently large. For every $T \in \binom{[n]}{t}$ and every $\ell \in \cliqueset{T}$, 
    \begin{align}
         \mathcal{A}(G, (1-\varepsilon)k) \sststile{x}{O(t)} \left\{ x_T \ \left( \sum_{v \in [n] \setminus S_\ell} x_v \right) \ \left( \sum_{S \subseteq \leftt, |S| = t} x_S \right)^2\le O\left( \frac{n^5}{k^2} \right) \ x_T \right\}.
    \end{align}
\end{lemma}

\begin{remark}[Intuition]
    Recall that $\leftt$ covers most of $S_\ell$, while $\rightt$ covers most of the rest of $V(\cliqueset{T})$. Then, the above lemma states that no clique of size $(1-\varepsilon)k$ can simultaneously have a large overlap with $S_\ell$ and $[n] \setminus S_\ell$.
\end{remark}

The proof does again rely on balancedness of $G[ \leftt \uplus\rightt]$.
This is captured in the following.

\begin{lemma}[The boundary of $S_\ell$ is balanced in $V(\cliqueset{T})$]\label{lem:boundarybalancedness}
    The following holds with high probability over the draw of $G$ and any $t \in \mathbb{N}$. For every $T \in \binom{[n]}{t}$ and every $\ell \in \cliqueset{T}$, the bipartite graph 
    $
        G[ \leftt \uplus \rightt] 
    $ has $r$-fold balancedness $o(k)$ for every $2 \le r \le t$. 
\end{lemma}
\begin{proof}
    Again, we sample $G \sim \RIGone$ in multiple phases, given a fixed $T \in \binom{[n]}{t}$. 
    \begin{enumerate}
        \setlength\itemsep{.0001em}
        \item Reveal the labels in $T$, thus uncovering $\cliqueset{T}$.
        \item For all vertices $[n] \setminus T$, reveal the labels in $\cliqueset{T}$.
        \item Reveal all remaining labels.
    \end{enumerate}
    Thus, after phase 2, we have uncovered which vertices are on the left and right of $G[ \leftt \uplus \rightt]$. Note that at this point, there are no edges between $\leftt$ and $\rightt$ yet. These only appear after phase 3 where our balancedness properties arise. 


    To show how, we first note that by \Cref{lem:E-is-likely}, the labels in $\leftt$  are such that $\mathcal{E}(\leftt, t, \frac{8t}{\alpha})$ is met, with probability $1 - O(n^{-2t})$. Then, fixing any pair $S, R \subseteq \leftt$ with $|S| = |R| =  t$ such that $|S \triangle R| \ge 3$, and any $v \in \rightt$, we get that 
    $$
        |\Expected{ u_{S,p}(v) u_{R,p}(v) }| \le O(\log(n)\delta) \maxbal{r} + O(\log(n) \delta \sqrt{d\delta})^3 \maxbal{r}
    $$ where the first term accounts for the possibility of choosing a duplicate label in $S \cup R$, and the second term for the remaing case assuming that no duplicate was chosen (cf. the proof of \Cref{lem:balancedness-core}). Summing over all $v \in \rightt$ and defining $U_{T,R,S,\ell} \coloneqq \left|\sum_{v \in \rightt} u_{S,p}(v) u_{R,p}(v)\right|$, then yields 
    $$
        \Expected{ U_{T,R,S,\ell} } \le |V(\cliqueset{T})| \ O\left( \frac{\log(n)}{\sqrt{d}} \right) \maxbal{r} \le O\left( \frac{\log(n) k}{\sqrt{d}} \right) \maxbal{r} = o(k),
    $$ 
    where we used \Cref{lem:sizeofdelta} to bound $\delta$ by $O(\sqrt{d})$ and $|V(\cliqueset{T})| \le \frac{8t}{\alpha}k = O(k)$ w.h.p. by \Cref{lem:boubdbadlabels} and a Chernoff bound on the size of all $S_\ell$ with subsequent union bound over $\ell \in [d]$.
    
    Now, applying \Cref{lem:bernstein}, we get that 
    \begin{align*}
        &\Pr{ |U_{T,R,S,\ell} - \Expected{ U_{T,R,S,\ell} }|  \ge z } \le \\ & \hspace{1cm} 2 \exp \left( -z^2 \ / \ \left(8\maxbal{2r}n + \frac{4}{3} \maxbal{r}z \right) \right)\le n^{-2t - 2r}
    \end{align*} for some $z =  \Omega(\sqrt{n \log(n)})$. If all of this occurs, then $U_{T,R, S,\ell} = o(k)$ as desired. Union bounding over all $R, S, T$ and all labels in $\cliqueset{T}$, yields the desired statement.
\end{proof}

Moreover, we need to establish that the overlap $\overlap{T}$ is small.
\begin{lemma}[Small overlap]\label{lem:smalloverlap}
    The following holds with high probability over the draw of $G$, for any integer $t \in \mathbb{N}$. For every $T \in \binom{[n]}{t}$, we have $\overlap{T} \le O(k \delta + \sqrt{n \log(n)}) = o(k)$.
\end{lemma}
\begin{proof}
    First fix a pair of labels $\ell, \ell'$. Then the expected number of vertices that have chosen $\ell, \ell'$ is $n\delta^2= k\delta$. Further, by \Cref{lem:bernstein} 
    $$
        \Pr{ \big | |S_\ell \cap S_{\ell'}| - k\delta \big| > \sqrt{n\log(n)} } \le 2 \exp \left( - \frac{n\log(n)}{2n+ \frac{2}{3}\sqrt{n\log(n)}} \right) \le n^{-\omega(1)}.
    $$ Hence, by a union bound over all $\ell, \ell'$ we have $|S_\ell \cap S_{\ell'}| = k \delta \pm \sqrt{n \log(n)}$ for all $\ell \neq \ell'$. 

    Fix a tuple $T \in \binom{[n]}{t}$ and note that by \Cref{lem:boubdbadlabels}, we have $|\cliqueset{T}| \le \frac{4t}{\alpha}$ w.h.p. for all $T$. Combined with the previous observation on $|S_\ell \cap S_{\ell'}|$ this yields $|\overlap{T}| \le |\cliqueset{T}|^2 (k \delta + \sqrt{n \log(n)}) = O(k \delta + \sqrt{n \log(n)})$, as desired.
\end{proof}

Now, we can use the established balancedness properties to obtain the following certificate asserting limited overlap between cliques in $\cliqueset{T}$. This statement will also serve an important purpose for our refutation algorithm in the next subsection.

\begin{proof}[Proof of \Cref{lem:limitedoverlapafterT}]
    To simplify notation in some places, we let $|y| \coloneqq \sum_{v \in \rightt } x_v$ and $|x| \coloneqq \sum_{v \in [n]} x_v$.
    Our starting point for the proof is the certificate from \Cref{lem:concentrationtofewcliques} stating that 
    $$
        \mathcal{A}(G, (1-\varepsilon)k) \sststile{x}{O(t)} \Bigg\{ x_T\ \bigg(\sum_{v \in [n] \setminus V(\cliqueset{T})} x_v\bigg) \le O \left( \frac{n}{k} \left( \frac{n^2}{k} \right)^{1/t} x_T\right) \Bigg\}.
    $$
    Choosing $t$ sufficiently large such that $O ( \frac{n}{k} ( \frac{n^2}{k} )^{1/t} ) = o(k)$ and noting that $\mathcal{A}(G, (1-\varepsilon)k) \sststile{x}{O(t)} \{ x_T |x| \ge x_T (1-\varepsilon)k \}$, this implies 
    \begin{align}\label{eq:mostinvt}
        \mathcal{A}(G, (1-\varepsilon)k) \sststile{x}{O(t)} \Bigg\{ x_T \sum_{v \in V(\cliqueset{T}) } x_v = x_T \  \bigg( |x| - \sum_{v \in [n] \setminus V(\cliqueset{T})  } x_v \bigg) \ge (1-2\varepsilon)k x_T \Bigg\}.
    \end{align}
    In other words, if $T$ is in the clique we are looking for, then almost all of its mass is in $V(\cliqueset{T})$.
    
    On the other hand, applying the statement from \Cref{lem:core-sos-proof} to the graph $G[\leftt \uplus \rightt]$ which by \Cref{lem:boundarybalancedness} has $t$-fold balancedness $\Delta = o(k)$, we get that
    \begin{align*}
         \mathcal{A}(G, (1-\varepsilon)k) \sststile{x}{O(t)} \left\{ \begin{aligned} &\left( \sum_{S \subseteq \leftt, |S| = t} x_S \right)^2 \left( \maxbalall{t} |y| - \Delta \right) \\ &\hspace{3cm} \le \left(n + tn^2\right) \maxbal{t} \left( \sum_{S \subseteq 
         \leftt, |S| = t} x_S \right) \end{aligned} \right\},
    \end{align*}
    or alternatively---to make the above easier to work with---we get 
    \begin{align}\label{eq:certconclusion}
         \mathcal{A}(G, (1-\varepsilon)k) \sststile{x}{O(t)} \left\{ \left( \sum_{S \subseteq \leftt, |S| = t} x_S \right)^2 \big( |y| - \deltap \big) \le \beta(n, t, p) \ \left( \sum_{S \subseteq \leftt, |S| = t} x_S \right) \right\},
    \end{align} where $\beta(n,t,p) = \left(n + tn^2\right) \maxbal{t} / \maxbalall{t}$ and $\deltap = \Delta / \maxbalall{t}$. 
    
    Intuitively, this tells us that our clique cannot simultaneously intersect both $\leftt$ and $
    \rightt$, provided that $|y| \gg \deltap$. While this is a-priori not given, we can make use of the same trick as used in \cite{Kothari-STOC-2023} to enforce a statement that is useful for us. Precisely, we multiply both sides by $x_i$ for $i \in [n] \setminus S_\ell$ and afterwards sum over $i$. The fact that the maximum degree of vertex $i$ into $S_\ell$ is bounded then leads to an ``all or nothing'' phenomenon: if only a single $x_i = 1$ for $i \in [n] \setminus S_\ell$, then this implies that in fact $\Omega(k)$ vertices in $[n] \setminus S_\ell$ are in our clique. Together with the certificate established in step 1, this also gives $|y| = \Omega(k)$. We capture all this in the following claim.

    \begin{claim}
    For all $i \in [n] \setminus S_\ell$, 
    $$
        \mathcal{A}(G, (1-\varepsilon)k) \sststile{x}{O(t)} \left\{ x_Tx_i \big( |y| - \deltap \big) \ge x_Tx_i (1 - p - 5\varepsilon)k \right\}
    $$
    \end{claim}\begin{proof}
        We use \Cref{lem:degreebound} to note that w.h.p., all $i \in [n] \setminus S_\ell$ have degree at most $(1+\varepsilon)pk$ into $S_\ell$. Therefore 
        $
            \mathcal{A}(G, (1-\varepsilon)k) \sststile{x}{O(t)} \big\{ x_i\sum_{j \in S_\ell} x_j \le (1+\varepsilon)pk x_i  \big\}.
        $ On the other hand, by \eqref{eq:mostinvt}, 
        $
            \mathcal{A}(G, (1-\varepsilon)k) \sststile{x}{O(t)} \big\{ x_T \sum_{j \in V(\cliqueset{T})} x_j \ge x_T(1-2\varepsilon)k \big\}.
        $ Moreover, since $|\overlap{T}| = o(k)$ by \Cref{lem:smalloverlap}, we get
        $$
            \mathcal{A}(G, (1-\varepsilon)k) \sststile{x}{O(t)} \left\{ x_Tx_i |y| \ge x_Tx_i \ \left( \sum_{v \in V(\cliqueset{T})} x_v - \sum_{v \in S_\ell }x_v - \sum_{v \in \overlap{T} } x_v \right) \ge x_Tx_i (1 - p - 4\varepsilon)k \right\}.
        $$ Subtracting $x_ix_T\deltap$ and noting $\deltap = o(k)$ yields the statement. 
    \end{proof}
    Using the above together with \Cref{eq:certconclusion}, we get that 
    \begin{align*}
         \mathcal{A}(G, (1-\varepsilon)k) \sststile{x}{O(t)} \left\{ x_Tx_i (1 - p - 5\varepsilon)k \ \left( \sum_{S \subseteq \leftt, |S| = t} x_S \right)^2\le x_Tx_i \ \beta(n, t, p) \ \left( \sum_{S \subseteq \leftt, |S| = t} x_S \right) \right\},
    \end{align*} which is already a lot more useful. Now, using the cancellation $\{ f^2 \le C f \} \sststile{x}{2d} \{ f^2 \le C^2 \}$, we get (after multiplying both sides by $x_Tx_i$) that
    \begin{align*}
         \mathcal{A}(G, (1-\varepsilon)k) \sststile{x}{O(t)} \left\{ x_Tx_i \ \left( \sum_{S \subseteq \leftt, |S| = t} x_S \right)^2 \le x_Tx_i \left( \frac{\beta(n, t, p)}{(1 - p - 5\varepsilon)k} \right)^2 \right\}.
    \end{align*} Summing over $i \in [n] \setminus S_\ell$ and noting that $\mathcal{A}(G, (1-\varepsilon)k) \sststile{x}{O(t)} \{ \sum_{i \in [n] \setminus S_\ell} x_i \le n \}$ finally yields that 
    \begin{align*}
         \mathcal{A}(G, (1-\varepsilon)k) \sststile{x}{O(t)} \left\{ x_T \ \left( \sum_{v \in [n] \setminus S_\ell} x_v \right) \ \left( \sum_{S \subseteq \leftt, |S| = t} x_S \right)^2\le  O\left( \frac{n^5}{k^2} \right) x_T \right\},
    \end{align*} as desired.
\end{proof}

\subsubsection{Step 3: Refutation}
With \Cref{lem:concentrationtofewcliques,lem:limitedoverlapafterT}, we show \Cref{thm:strongsinglelabelclique}. We formulate it as a SoS-proof that ends in the inequality $-1 \ge 0$ yielding a ``proof by contradiction'' that the premise (that there exists $(1-\varepsilon)k$ sized clique that is not ground-truth) is wrong. To this end, we first refomulate a set of axioms that encodes this premise and afterwards use the certificats from \Cref{sec:fewoutsidevcliqueset} and \Cref{sec:slctstep2}.

\paragraph{Axioms}
Our axioms depend on the graph $G$, the desired clique size $\adv{k}$ to be ruled out, and the ground truth cliques $S_\ell$, $\ell \in [d]$. We simply add an axiom telling us that the sum over all variables outside of each $S_\ell$ is larger than 0.
$$
    \adv{\mathcal{A}}\big(G, \adv{k}, \{S_\ell\}_{\ell \in [d]}\big) = \cA(G, \adv{k}) \cup \bigg\{
         \sum_{v \in [n] \setminus S_\ell} x_v \ge 1
    \bigg\}_{\ell \in [d]}
$$

\paragraph{Refutation}

\newcommand{\refaxioms}{\adv{\mathcal{A}}\big(G, (1-\varepsilon)k, \{S_\ell\}_{\ell \in [d]}\big)}
\newcommand{\refaxiomsc}{\adv{\mathcal{A}}\big(G, (1 - \varepsilon)k, \{S_\ell\}_{\ell \in [d]}\big)}

Our goal is to prove the following.
\begin{lemma}\label{lem:refutationforslct}
    Given a $G \sim \RIGone$ such that the assumption of \Cref{thm:strongsinglelabelclique} are met. Then there is a sufficiently small $\varepsilon > 0$ such that the following holds with high probability over the draw of $G$
    $$
        \adv{\mathcal{A}}(G, (1-\varepsilon)k, \{S_\ell\}_{\ell \in [d]}) \sststile{x}{O(t)} \{ -1 \ge 0\}
    $$ and the bit complexity of the proof is $n^{O(1)}$.
\end{lemma}
The main work towards establishing the above lies in the following lemma.
\begin{lemma}\label{lem:almostrefutation}
    Given $G \sim \RIGone$ such that the assumption of \Cref{thm:strongsinglelabelclique} are met and a sufficiently large integer $t$, the following holds with high probability over the draw of $G$. For all $T \in \binom{[n]}{t}$,
    $$
        \refaxiomsc \sststile{x}{O(t)} \{ x_T \le -x_T\}.
    $$
\end{lemma}
\begin{proof}
    We start by using the conclusion of \Cref{lem:limitedoverlapafterT} which states that for every $\ell \in \cliqueset{T}$,  
    \begin{align*}
         \refaxiomsc \sststile{x}{O(t)} \left\{ x_T \ \left( \sum_{v \in [n] \setminus S_\ell} x_v \right) \ \left( \sum_{S \subseteq \leftt, |S| = t} x_S \right)^2\le O\left( \frac{n^5}{k^2} \right) \ x_T \right\}.
    \end{align*} 
    We proceed by finding a suitable lower bound on the left hand side. To this end, we can first observe that the axiom
    $
        \{ \sum_{v \in [n] \setminus S_\ell} x_v \ge 1 \}
    $ allows us to drop the sum $\sum_{v \in [n] \setminus S_\ell} x_v$ entirely. Then, the cancellation $\{ f^2 \le C  \} \sststile{x}{2d} \{ f \le \sqrt{C} \}$ from \Cref{lem:cancellation} yields 
    \begin{align*}
         \refaxiomsc \sststile{x}{O(t)} \left\{ x_T \ \left( \sum_{S \subseteq \leftt, |S| = t} x_S \right) \le O\left( \frac{n^{5/2}}{k} \right) \ x_T \right\}.
    \end{align*}
    To deal with the remaining sum, we apply \Cref{lem:sos-factorial} to get 
    \begin{align*}
         \refaxiomsc \sststile{x}{O(t)} \left\{ \frac{x_T}{2^t t!} \  \left( \sum_{v \in \leftt} x_v \right)^{t}  \le x_T \left( \frac{2t^t}{t!} + \varepsilon + O\left( \frac{n^{5/2}}{k} \right)\right) \le O\left( \frac{n^{5/2}}{k} \right) \ x_T \right\}.
    \end{align*}
    Now, using the cancellation $\{ f^2 \le C  \} \sststile{x}{2d} \{ f \le \sqrt{C} \}$ from \Cref{lem:cancellation} $\log_2(t)$ times and choosing $t$ as power of $2$, we get that 
    $$
         \refaxiomsc \sststile{x}{O(t)} \left\{ x_T \ \left( \sum_{v \in \leftt} x_v \right)  \le  O\left( \frac{n^{5/2}}{k} \right)^{1/t} \ x_T \right\}.
    $$  Now, we can sum over $\ell \in \cliqueset{T}$ to obtain \begin{align*}
         \refaxiomsc \sststile{x}{O(t)} \left\{ x_T \ \left( \sum_{v \in V(\cliqueset{T}) \setminus \overlap{T}} x_v \right)  \le  |\cliqueset{T}| O\left( \frac{n^{5/2}}{k} \right)^{1/t} \ x_T \right\},
    \end{align*}
    since $\bigcup_{\ell \in \cliqueset{T}} \leftt = V(\cliqueset{T}) \setminus \overlap{T}$ and the $\leftt$ are disjoint by definition.
    Using the fact that $|\cliqueset{T}| \le \frac{4t}{\alpha}$ by \Cref{lem:boubdbadlabels} for all $T$ with high probability, we obtain 
    \begin{align}\label{eq:upperboundinvt}
         \refaxiomsc \sststile{x}{O(t)} \left\{ x_T \ \left( \sum_{v \in V(\cliqueset{T})} x_v \right)  \le O\left( \frac{n^{5/2}}{k} \right)^{1/t} \ x_T + |\overlap{T}| \ x_T \right\}.
    \end{align} Finally, it only remains to apply \Cref{lem:concentrationtofewcliques} stating that 
    \begin{align*}
        \refaxiomsc \sststile{x}{O(t)} \Bigg\{ x_T\ \left(\sum_{v \in [n] \setminus V(\cliqueset{T})} x_v\right) \le  O \left( \frac{n}{k} \left( \frac{n^2}{k} \right)^{1/t} \right) \ x_T \Bigg\}.
    \end{align*}
    Choosing $t$ sufficiently large, we can now achieve that $
        \max \big\{  \frac{n}{k} ( n^2/k)^{1/t}, (n^{5/2}/k )^{1/t}  \big\} = o(k)
    $. Moreover, we get $|\overlap{T}| = o(k)$ by \Cref{lem:smalloverlap} with high probability. In total, this gives 
    \begin{align*}
        \refaxiomsc \sststile{x}{O(t)} \Bigg\{ x_T\ \Bigg(\sum_{v \in [n]} x_v\Bigg) \le \frac{1}{2}k\  x_T \Bigg\}.
    \end{align*}
    On the other hand, due to the constraint $\{ \sum_v x_v \ge (1-\varepsilon)k \}$, we get 
    $$
        \refaxiomsc \sststile{x}{O(t)} \Bigg\{ (1-\varepsilon)k \ x_T  \le  x_T \Bigg(\sum_{v \in [n] } x_v\Bigg) \le \frac{1}{2}k \ x_T  \Bigg\}.
    $$ and the lemma follows after re-arranging if $\varepsilon < \frac{1}{2}$. \qedhere
    
\end{proof}

\Cref{lem:refutationforslct} then essentially follows by summing over $T$ and the fact that SoS knows that at least one $t$-tuple $T \subseteq [n]$ must be in our clique.
\begin{proof}[Proof of \Cref{lem:refutationforslct} ]
    We sum the conclusion of \Cref{lem:almostrefutation} over all $T$ to obtain (after re-arranging and dividing by $2$) that $
        \refaxiomsc \sststile{x}{O(t)} \left\{ \sum_{T \subseteq [n], |T| = t} x_T  \le 0  \right\}.
    $ The fact that we have the axiom $\{ \sum_v x_v \ge (1+\varepsilon')k \}$ implies by \Cref{lem:sosbinomialcoefficient} that 
    $$
        \refaxiomsc \sststile{x}{O(t)} \left\{\sum_{T \subseteq [n], |T| = t} x_T  \ge \binom{n}{\lfloor (1-\varepsilon)k \rfloor} \ge 1 \right\}. 
    $$ Therefore, $\refaxiomsc \left\{ -1 \ge 0 \right\}$ as desired.
\end{proof}

Now \Cref{thm:strongsinglelabelclique} is essentially a corollary of the above.
\begin{proof}[Proof of \Cref{thm:strongsinglelabelclique}]
    By \Cref{lem:refutationforslct} we know that the axioms $\refaxiomsc$ result in a contradiction. This is equivalent to \Cref{thm:strongsinglelabelclique}.
\end{proof}

\subsection{Algorithmic Refutation}\label{sec:refutation}

In addition to the fact that the ideas from \Cref{lem:slctimproved} yield a better single label clique theorem, they have algorithmic consequences. Since everything can be stated in terms of low-degree SoS, a certificate that no cliques outside of the ground-trugh exist \emph{can actually be computed in polynomial time}. 

This has a curious consequence in light of our algorithm for exact recovery from \Cref{thm:exactrecovery}: after termination, the algorithm can actually \emph{certify that it is correct} in the sense that it outputs a list $\mathcal{L}$ of cliques and a certificate proving that \emph{no other clique} of size $(1-\varepsilon)k$ exists that is not contained in one of the cliques in $\mathcal{L}$. Moreover, if we simply want a certificate that no clique of a certain size $(1 + \varepsilon)k$ exists \emph{without} prior knowledge of the ground truth cliques, this can also be computed in polynomial time. We capture these two facts in the following theorem.

\begin{theorem}[Algorithmic refutation]\label{thm:algorithmicrefutation}
    There are polynomial time algorithms with the following guarantees on input $G \sim \RIGone$ such that the assumptions of \Cref{thm:strongsinglelabelclique} are met. 
    \begin{enumerate}
        \setlength\itemsep{0.0001em}
        \item Output a list $\mathcal{L}$ of cliques in $G$ together with a polynomial time verifiable proof that every clique of size $(1-\varepsilon)k$ in $G$ is contained in one of the cliques in $\mathcal{L}$. This holds for every sufficiently small $\varepsilon$.
        \item Output a polynomial time verifiable proof that no clique of size $(1 + \varepsilon)k$ exists in $G$. This holds whenever $\varepsilon \gg \sqrt{\log(n)/k}$.
    \end{enumerate}
    Both algorithms succeed with high probability over the choice of $G$.
\end{theorem}
\begin{proof}
    It follows from standard results in the literature on convex optimization \cite{nesterov2000squared, parrilo2000structured, lasserre2001new} that whenever there exists a degree-$d$ SoS-proof of some inequality $\{g \ge 0\}$ in axioms $\cA$ of bit complexity $n^{O(1)}$, then we can approximately compute a such proof (in the sense that we prove $\{g \ge -\eta\}$ for some small $\eta$) in time $(m + n)^{O(d)}\text{polylog}(1/\eta)$ where $m$ is the number of constraints in $\cA$ and $n$ is the number of variables. This implies statement 1 of our theorem by the SoS proof from \Cref{lem:refutationforslct}. 

    For statement 2, we note that for the proof of \Cref{lem:refutationforslct}, we used the axiom $\sum_{v \in [n] \setminus S_\ell} x_v \ge 1$ for all $\ell$. While we cannot hope to use these axioms without knowing the ground-truth $S_\ell$, we can simply use the fact that by a Bernstein and union bound (cf. \Cref{lem:bernstein}), every $S_\ell$ has size at most $k + O(\sqrt{k\log(n)})$ w.h.p. By our assumption $\varepsilon \gg \sqrt{\log(n) / k}$, this implies that 
    $
        \cA(G, (1+\varepsilon)k) \sststile{x}{O(t)} \{ \sum_{v \in [n] \setminus S_\ell} x_v \ge (1 + \varepsilon)k - k + O(\sqrt{k\log(n)}) \ge 1\}.
    $ Hence, the proof from \Cref{lem:refutationforslct} carries over, and the statement follows using the same argument as used in the previous paragraph.
\end{proof}

\section{Analysis for exact recovery}\label{sec:analysisexact}

In this section, we show correctness of the algorithm from \Cref{thm:exactrecovery}, i.e., we prove that \Cref{alg:splitting-one-sided} achieves exact recovery, even under presence of a monotone adversary. The analysis is split into two parts: First, we prove that every ground truth clique $S_\ell$ is contained in the output $\mathcal{L}$, and afterwards, we show that no clique that is not ground-truth is added to the output. For convenience, we restate the algorithm here.

    
  \begin{algorithm}
    \setlength\itemsep{.0001em}
    \caption{ Recovering the labels of a dense RIG with One-Sided Noise. }
    \label[algorithm]{alg:splitting-one-sided}
    \begin{description}
    \setlength\itemsep{.0001em}
    \item[Given:] A graph $G \sim \RIGone$ on vertex set $[n]$ modified by an arbitrary monotone adversary.
    \item[Output:]
      A list $\mathcal{L}$ of $d$ cliques in $G$ such that exactly one $S \in \mathcal{L}$ agrees with $S_\ell$ for every $\ell \in [d]$.
    \item[Operation:]\mbox{}
    \begin{enumerate}
    \item Split $[n]$ into two disjoint sets $U, V$ using an independent, fair coin flip for every $v \in [n]$.
    \item Set $\gamma \coloneqq \frac{1}{\alpha}\max\{ 16p^{-6}, 1000 \} \max\{1,  (p/1-p)^6\}$, and $t \coloneqq 2\lceil \log_{1/p}(\gamma) \rceil$. Choose a list $\mathcal{T}$ of $O( n^{\alpha(1+ t/2)} )$ $t$-tuples from $\binom{[n]}{t}$ uniformly at random.
    \item For each $T \in \mathcal{T}$, compute the bipartite graph $H(T) = G[U \uplus V] \cap N_G(T)$
    \item If feasible, compute a degree-$\degreeexact$ pseudo-distribution $\mu$ in variables $w$ satisfying the biclique axioms $\bicliqueaxiomsalgocompact$ for a sufficiently small constant $\varepsilon > 0$ and $\kbip \coloneqq \frac{1-\varepsilon}{2}k$. Otherwise, proceed with the next $T \in \mathcal{T}$.
    \item Let $S_T = \{ v \in [n] \mid \Expectedtildesub{\mu}{w_v} \ge 1 - \sqrt{\varepsilon} \ \}$. 
    \item Combinatorial clean-up: remove from $S_T$ all vertices that have fewer than $(1 - 5\sqrt{\varepsilon})k$ neighbors in $S_T$. Then, add all vertices from $[n] \setminus S_T$ with at least $(1-5\sqrt{\varepsilon})k$ neighbors in $S_T$. If $G[S_T]$ is a clique of size at least $(1 - \varepsilon)k$, add $S_T$ to the output $\mathcal{L}$.
    \item Finally, for every $ S \in \mathcal{L}$, delete all cliques $S' \in \mathcal{L} \setminus \{S\}$ such that $S' \subseteq S$.
  \end{enumerate}
    \end{description}
  \end{algorithm}

\subsection{Preliminaries}

Before starting with the actual analysis, we introduce the following certificate used in multiple parts of the analysis. It states that a subset of vertices in a bipartite graph that induces a balanced sub-graph can only contain a very limited number of vertices in any clique. This works whenever we have a balancedness of roughly $ck \maxbalallcompact{r}$ (for $c > 0$ being a sufficiently small constant) since the certificates from \Cref{lem:core-sos-proof} and \Cref{lem:core-sos-proof-adv} are meaningful whenever $|y| > \Delta$
\begin{lemma}[Bounding biclique sizes in balanced bipartite graphs]\label{lem:sos-left-certificate}
    Let $H = (A \uplus B, E)$ be a bipartite graph on vertex set $[n]$ that contains some set $\Asub \subseteq A$ such that $\adv{H} \coloneqq H[\Asub \uplus B]$ has $r$-fold balancedness at most $\Delta \le \maxbalall{r} \frac{11}{12}k$ for some fixed $r \in \mathbb{N}, r \ge 2$. Then for some sufficiently small constant $\varepsilon > 0$, 
    \begin{align*}
        \mathcal{B}(H, k )  \sststile{x, y}{4r} \left\{ \bigg( \sum_{i \in \Asub} x_i \bigg)^{r} \le O\left(\frac{n^2}{k}\right) \right\}.
    \end{align*}
\end{lemma}
\begin{proof}
    We invoke \Cref{lem:core-sos-proof} for on the balanced graph $H'$ with bipartition $A' \uplus B$. This yields that 
    \begin{align*}
        \mathcal{R}(\adv{H}) \sststile{x, y}{4r} \left\{\begin{aligned} &\maxbalall{r} \left( \sum_{S \subseteq \Asub, |S| = r} x_S \right)^2 | y | \\ & \hspace{2cm} \le \left(n + rn^2\right) \maxbal{r} \left( \sum_{S \subseteq \Asub, |S| = r} x_S \right) + \Delta \hspace{.05cm}\left( \sum_{S \subseteq \Asub, |S| = r} x_S \right)^2 \end{aligned} \right\}.
    \end{align*} By the size constraint in $\mathcal{B}(H, \frac{1-\varepsilon}{2}k)$, we have that $|y| \ge k$. So using $\Delta \le \maxbalall{r} \frac{11k}{12}$ and $\mathcal{R}(H') \subseteq \mathcal{B}(H, \frac{1-\varepsilon}{2}k)$, we get that
    $$
        \mathcal{B}(H, k ) \sststile{x, y}{4r} \left\{ \maxbalall{r} \left( \sum_{S \subseteq \Asub, |S| = r} x_S \right)^2 \frac{k}{12} \le (r+1)n^2 \maxbal{r}
        \left( \sum_{S \subseteq \Asub, |S| = r} x_S \right) \right\}.
    $$ Dividing by $\maxbalall{r}$ and using the cancellations $\left\{ f^2 \le C f \right\} \sststile{f}{2r} \left\{ f^2 \le C^2 \right\}$, and $\left\{ f^2 \le C  \right\} \sststile{x}{2r} \{ f \le \sqrt{C} \}$ from \Cref{lem:cancellation}, we get that 
    $$
        \mathcal{B}(H, k )  \sststile{x, y}{4r} \left\{ \bigg( \sum_{S \subseteq \Asub, |S| = r} x_S \bigg)\le  \frac{12(r+1) n^2}{k} \max \left\{ 1, \frac{p}{1-p} \right\}^{2r} \right\},
    $$ and by \Cref{lem:sos-factorial},
    \begin{align*}
        \mathcal{B}(H, k )  \sststile{x, y}{4r} \left\{ 
            \begin{aligned}
                &\frac{1}{2^r r!} \bigg( \sum_{i \in \Asub} x_i \bigg)^{r} \le \left(\sum_{S \subseteq \Asub, |S| = r} x_S \right) +  \frac{2r^r}{r!} + \varepsilon 
            \end{aligned}
        \right\},
    \end{align*} so using that $r$ is a constant and combining the above with the previous observation, we get that 
    $$
        \mathcal{B}(H, k )   \sststile{x, y}{4r} \left\{ \bigg( \sum_{i \in \Asub} x_i \bigg)^{r} \le O\bigg(\frac{n^2}{k}\bigg) \right\},
    $$
    as desired.
\end{proof}

Now, the analysis is split into two parts. First, we show that for every $S_\ell$, the algorithm encounters a $T \in \mathcal{T}$ such that the set $S_T$ is precisely $S_\ell$ after step 6. Afterwards, we show that no set $S$ with $S \neq S_\ell$ for all $\ell \in [d]$ is ever added to the output.

\subsection{Step 1: All ground-truth cliques are found}\label{sec:exactstep1}

We crucially rely on the balancedness of the bipartite graph $H(T) \coloneqq G[U \uplus V] \cap N(T)$ for some of the $t$-tuples $T \in \mathcal{T}$ chosen in step $1$.

\begin{lemma}\label{lem:balancedness-hl}
    With probability $\ge 0.99$, the following holds after step 2 of \Cref{alg:splitting-one-sided}. For every $\ell \in [d]$, there is a $t$-tuple $T \in \mathcal{T}$ such that the bipartite graphs
$$
    \Hg{L} \coloneqq G[ (U \setminus S_\ell) \uplus V ] \cap N_G(T) \text{ and } \Hg{R} \coloneqq G[ (V \setminus S_\ell) \uplus U ] \cap N_G(T)
$$ have $3$-fold balancedness at most $\maxbalall{3}\frac{k}{4}$.
\end{lemma} \begin{proof}
    We prove the lemma only for $\Hg{L}$, the proof for $\Hg{R}$ works analogously. 
    Our statement is a consequence of \Cref{lem:general-balancedness} which states that for all $\ell \in [d]$, the neighbourhood of most $t$-tuples in $S_\ell$ contains a balanced bipartite graph whose bipartition is formed by a procedure $\textsc{Alg}$ that only has access to $S_\ell$ and $[n] \setminus S_\ell$. In our case, we choose $\textsc{Alg}(S_\ell, [n] \setminus S_\ell)$ to be the (deterministic) procedure that outputs the bipartition $(U \setminus S_\ell) \uplus V$ where $U, V$ were decided in step 1 of \Cref{alg:splitting-one-sided}. Now let $\gamma \coloneqq \frac{1}{\alpha}\max\{ 16p^{-6}, 1000 \} \max\{1,  (p/1-p)^6\}$. Since $t = 2\lceil \log_{1/p}(\gamma) \rceil$, we get $3$-fold balancedness \begin{align*}
        \Delta &\le \frac{16p^{-6}}{\alpha} \maxbal{3}  \  tp^{t}k + o(k) \\
        &\le \frac{16p^{-6}}{\alpha} \left( \frac{1}{\gamma} \right)^2 2\left( \log_{1/p}(\gamma) + 1 \right) \maxbal{3} k + o(k)\\
        &\le \frac{2}{\gamma}  \left( \log_{1/p}
        (\gamma)+ 1 \right) \max\left\{1,  \left(\frac{1-p}{p}\right)\right\}^6 \maxbal{3} k + o(k)\le \maxbalall{3} \frac{k}{4},
    \end{align*}
    where the last step follows since $\frac{2}{x}(\log_2(x) + 1) \le \frac{1}{5}$ for all $x \ge 1000$ and since $\gamma \ge 1000$. 
    
    It remains to show that \Cref{alg:splitting-one-sided} finds such a tuple for every $\ell \in [d]$ in step 1. To this end, note that \Cref{lem:general-balancedness} tells us that a randomly chosen $t$-tuple from $\binom{S_\ell}{t}$ is such that the graph $G[ (U \setminus S_\ell) \uplus V ] \cap N_G(T)$ has the desired balancedness with probability $1 - o(1)$. Hence, when sampling a $t$-tuple uniformly at random from $[n]$, the probability of finding such a tuple in $S_\ell$ is $\Omega((k/n)^t)$. By the analysis of the coupon collector process we get that after $O(d \log(d) (n/k)^t)$ samples, we found one such tuple for each $S_\ell, \ell \in [d]$ with probability $\ge 0.99$. We finish the proof by noting that $d \log(d) (n/k)^t = O(\log(d)d^{1 + \frac{t\alpha}{2}}) = O(n^{\alpha\left(1 + \frac{t}{2}\right)})$.
\end{proof}

By the above, we can assume that we are working with a balanced bipartite graph $H(T) = G[U \uplus V] \cap N_G(T)$ for some $T \subseteq S_\ell$ as of now. We show that in such a graph, a degree $\degreeexact$ pseudo-distribution satisfying the axioms $\mathcal{B}(H(T), \frac{1-\varepsilon}{2}k)$ allows us to recover the set $S_\ell$. To this end, we need the following lemma that yields a SoS certificate for the absence of a $O( \sqrt{n} ) \times \frac{1-\varepsilon}{2}k$ biclique in $\Hg{L}$ and $\Hg{R}$.

\begin{lemma}\label{cor:low-mass-outside-sell}
    Assume $T$ is a $t$-tuple $T \subseteq S_\ell$ such that $\Hg{L}$, and $\Hg{R}$ meet the balancedness criteria from \Cref{lem:balancedness-hl}, then
    $$
        \bicliqueaxiomsalgo \sststile{x, y}{\degreeexact} \bigg\{ \sum_{i \in U \setminus S_\ell} x_i, \sum_{i \in V \setminus S_\ell} y_i \le O( \sqrt{n}) \bigg\}
    $$
    and bit complexity of the proof is $n^{O(1)}$.
\end{lemma}\begin{proof}
    Note that $\Hg{L}, \Hg{R}$ are subgraphs of $H(T)$ with balancedness at most $ \maxbal{r} \frac{k}{4}$ for all $3 \le r \le t$ by \Cref{lem:balancedness-hl}.
    Thus, invoking  \Cref{lem:sos-left-certificate} with $r = 3$, we get that $$
        \bicliqueaxiomsalgo \sststile{x, y}{\degreeexact} \bigg\{ \bigg( \sum_{i \in U \setminus S_\ell} x_i \bigg)^{3} \le O\Big(\frac{n^2}{k}\Big) \bigg\}.
    $$   
    Combining this with the fact that $\bicliqueaxiomsalgo \sststile{x, y}{\degreeexact} \left\{ \sum_{i \in U \setminus S_\ell} x_i \le 2 k \right\}$, we get that 
    $
        \bicliqueaxiomsalgo \sststile{x, y}{\degreeexact} \left\{ (\sum_{i \in U \setminus S_\ell} x_i)^{4} \le O( n^2) \right\}.
    $ Applying the cancellation $\{ f^2 \le C \} \sststile{f}{2} \{ f \le \sqrt{C} \}$ from \Cref{lem:cancellation} twice, we arrive at the desired conclusion. An analogous argument for $\Hg{R}$ also yields the second part of the statement.
\end{proof}
\noindent The duality between SoS-proofs and pseudo-distributions then immediately yields that we can recover a biclique that has large intersection with $S_\ell$ and small intersection with $[n]\setminus S_\ell$ by taking $S_T$ as specified in step $5$.
\begin{lemma}\label{lem:recover-cliques}
    Assume $T$ is a $t$-tuple $T \subseteq S_\ell$ such that $\Hg{L}$, and $\Hg{R}$ meet the balancedness criteria from \Cref{lem:balancedness-hl}. Then for any degree $\degreeexact$ pseudo-distribution $\mu$ satisfying $\mathcal{B}(H(T), \frac{1-\varepsilon}{2}k)$ in variables $w = \{\{x_v\}_{v \in U}, \{y_v\}_{v \in V}\}$, we have $|S_T \triangle S_\ell| \le 4\sqrt{\varepsilon} k$ after step 5 of \Cref{alg:splitting-one-sided}.
\end{lemma} \begin{proof}
    First of all, observe that for every $S_\ell$, we get that over step 1 of the algorithm, we have $ \frac{1-\varepsilon}{2}k \le  |S_\ell \cap U|, |S_\ell \cap V| \le \frac{1+\varepsilon}{2} k$ for any fixed $\varepsilon > 0$, with high probability. This follows from standard Chernoff- and subsequent union bounds. Using \Cref{cor:low-mass-outside-sell} combined with the duality between SoS-proofs and pseudo-distributions, we get that $\sum_{i \in (U \cup V) \setminus S_\ell}\Expectedtildesubnop{\mu}{w_i} \le O(\sqrt{n})$. 

    This already implies that $| S \setminus S_\ell| = o(k)$ since at most $O(\sqrt{n})$ vertices $i$ outside of $S_\ell$ are such that $\Expectedtildesubnop{\mu}{x_i} \ge 1 - \varepsilon$. Moreover, since $\sum_{i \in U \cup V}\Expectedtildesubnop{\mu}{w_i} = 2\frac{1-\varepsilon}{2}k \ge (1-\varepsilon) k$, we get that $\sum_{i \in S_\ell} \Expectedtildesubnop{\mu}{w_i} \ge (1-\varepsilon) k - o(k)$. This implies that on average over $i \in S_\ell$, $\Expectedtildesubnop{\mu}{w_i} \ge 1-2\varepsilon$. Hence, by Markov's inequality, there are at most $2\sqrt{\varepsilon} k$ vertices in $S_\ell$ with $\Expectedtildesubnop{\mu}{w_i} < 1 - \sqrt{\varepsilon}$. With this, it follows that $|S_\ell \triangle S_T| \le 4\sqrt{\varepsilon} k$, as desired.
\end{proof}

What is remarkable about the above is that the conclusion still holds after allowing a monotone adversary to modify $G$. This is despite the fact that the resulting modified graph $G'$ and the associated bipartite graphs considered by \Cref{alg:splitting-one-sided} might not be balanced anymore. The conclusions hold nonetheless since the computed pseudo-distributions satisfy a useful monotonicity property.

\begin{lemma}\label{lem:monotoneadversary}
    Given $G \sim \RIGone$, let $\widetilde{G}$ denote $G$ after an arbitrary subset of edges $\{u ,v\} \in E(G)$ with $M_u \cap M_v = \emptyset$ has been deleted. Furthermore, analogously to $H(T)$, define 
    \begin{align*}
        \widetilde{H}(T) &\coloneqq \widetilde{G}[U \uplus V] \cap N_{\widetilde{G}}(T).
    \end{align*}
    Then \Cref{lem:recover-cliques} still holds when replacing $H(T)$ by $\widetilde{H}(T)$. 
\end{lemma} \begin{proof}
    \Cref{lem:recover-cliques} holds for all pseudo-distributions $\mu \sdtstile{x, y}{\degreeexact} \mathcal{B}(\Hg{L}, \frac{1-\varepsilon}{2}k)$. Since $\widetilde{H}(T) \subseteq H(T)$ if $T \subseteq S_\ell$ for any $\ell \in [d]$, we get that every degree $\ge \degreeexact$ pseudo-distribution $\widetilde{\mu}$ satisfying $\mathcal{B}(\widetilde{H}(T), \frac{1-\varepsilon}{2}k)$ is also a pseudo-distribution satisfying $\mathcal{B}(H, \frac{1-\varepsilon}{2}k)$. Moreover, since all edges within $S_\ell$ remain unchanged when passing from $G$ to $\widetilde{G}$, we get that $S_\ell \in \widetilde{H}$, so there further exists at least one pseudo-distribution $\widetilde{\mu}$ satisfying $\mathcal{B}(\widetilde{H}, \frac{1-\varepsilon}{2}k)$ (take for example any actual probability distribution that puts all its mass on $\frac{1-\varepsilon}{2}k$ vertices in $U\cap S_\ell$ and $V\cap S_\ell$, respectively).
\end{proof}

Finally, it only remains to argue that the combinatorial clean-up in step 6 exactly recovers $S_\ell$. 
\begin{lemma}\label{lem:equalaftercleanup}
    Assume $T$ is a $t$-tuple $T \subseteq S_\ell$ such that $\Hg{L}$, and $\Hg{R}$ meet the balancedness criteria from \Cref{lem:balancedness-hl}. Then, after step 6, we have $S_T = S_\ell$. The lemma continues to hold if we replace $G$ by $\widetilde{G}$ as specified in \Cref{lem:monotoneadversary}.
\end{lemma} \begin{proof}
    We note that by \Cref{lem:recover-cliques}, we find a set $S_T$ with $|S_\ell \triangle S_T| \le 4\sqrt{\varepsilon} k$. Since $|S_\ell| \ge (1 - o(1))k$ w.h.p., we get that every $v \in S_\ell \cap S_T$ has at least $(1 - 5\sqrt{\varepsilon})k$ neighbours in $S_T$, so these are not deleted. Moreover, by \Cref{lem:degreebound}, each $v \in S_T\setminus S_\ell$ has at most $(1+o(1))pk + 4\sqrt{\varepsilon} k < (1 - 5\sqrt{\varepsilon})k$ neighbours in $S_T$, so they are removed (if $\varepsilon$ is sufficiently small). Finally, vertices in $S_\ell \setminus S_T$ have at least $(1 - 5\sqrt{\varepsilon})k$ neighbours in $S_T$, so these vertices are added. Hence, $S_T = S_\ell$ after step 6.
    The same argument applies to $\adv{G}$ by \Cref{lem:monotoneadversary}. 
\end{proof}

\subsection{Step 2: No other cliques are found}

We proceed by showing that the algorithm never finds a clique outside of the ground-truth.
\begin{lemma}\label{lem:step2analysisexactprime}
    Let $T \in \binom{[n]}{t}$ be any $t$-tuple in $G \sim \RIGone$. Then, after step 6 of \Cref{alg:splitting-one-sided}, the set $S_T$ is either discarded or a subset of $S_\ell$ for some $\ell \in [d]$.
\end{lemma}
With the above lemma and \Cref{lem:equalaftercleanup} it immediately follows that the clean-up in step $7$ yields a list $\mathcal{L}$ that contains exactly the ground-truth cliques $S_\ell$ for all $\ell \in [d]$, so \Cref{thm:exactrecovery} follows.

To prove \Cref{lem:step2analysisexactprime}, we recall that given a $T \in \binom{[n]}{t}$, the set  $\cliqueset{T} \subseteq [d]$ denotes the set of labels that have been chosen by at least three vertices in $T$. Moreover, $V(\cliqueset{T})$ is the set of vertices in $[n]$ that have chosen a label in $\cliqueset{T}$. The key observation we use (reminiscent of the ideas in \Cref{sec:fewoutsidevcliqueset}) is that for \emph{all} $T \in \binom{[n]}{t}$, the axioms $\bicliqueaxiomsalgocompact$ are sufficient to show that not many vertices can be in $V(\cliqueset{T})$.
\begin{lemma}\label{cor:low-mass-outside-sell-2}
    With high probability over the draw of $G \sim \RIGone$, the following holds.
    For all $T \subseteq \binom{[n]}{t}$,
    $$
        \bicliqueaxiomsalgo \sststile{x, y}{\degreeexact} \bigg\{ \sum_{i \in U \setminus V(\cliqueset{T})} x_i, \sum_{i \in V \setminus V(\cliqueset{T})} y_i \le O( \sqrt{n}) \bigg\}.
    $$
\end{lemma}
This has the implication that the pseudo expectation on $[n] \setminus V(\cliqueset{T})$ is $O(\sqrt{n})$. Since the number of labels in $\cliqueset{T}$ is bounded w.h.p., this allows us to use a simple argument invoking the combinatorial insights from \Cref{sec:cliqueintersections} to show that exact recovery  succeeds. 

\begin{proof}[Proof of \Cref{lem:step2analysisexactprime}] 

    By \Cref{cor:low-mass-outside-sell-2}, we get that for all $T \in \binom{[n]}{t}$, the set $S_T$ is such that $|S_T \cap ([n] \setminus V(\cliqueset{T}))| \le O(\sqrt{n})$. Now, whenever $|S_T| \ge (1 - \varepsilon) k$ after the clean up in step $6$, then 
    we get from the fact that $|\cliqueset{T}| \le \frac{4t}{\alpha}$ by \Cref{lem:boubdbadlabels},  that there is some $\ell \in \cliqueset{T}$ such that 
    $$
        |S_T \cap S_\ell| \ge \frac{\alpha}{8t} k.
    $$
    Since we can further assume $S_T$ to induce a clique in $G$, we can now rely on the combinatorial properties regarding cliques intersecting ground-truth cliques from \Cref{sec:cliqueintersections}. 
    Precisely, by \Cref{lem:bad-event}, we get that with high probability, the event $\mathcal{E}_{a, b}$ for $b = \lfloor \log(n)^2 \rfloor$, any $a \ge n^{\varepsilon}k/\sqrt{d}$ and any $\varepsilon > 0$ does not occur. Hence, for sufficiently small $\varepsilon$, we get $a \ll \frac{\alpha}{8t} k$, so the non-occurrence of $\mathcal{E}_{a,b}$ implies that $|S_T \setminus S_\ell| \le \lfloor \log(n)^2 \rfloor$. Since $|S_T| \ge (1-\varepsilon)k$, this implies that all $v \in S_T \setminus S_\ell$ have degree at least $(1-2\varepsilon)k$ into $S_\ell$. For sufficiently small $\varepsilon$, this contradicts the degree bound from \Cref{lem:degreebound}, so the only way that $S_T$ is not discarded at the end of step $6$ is that is is a subset of $S_\ell$. 
\end{proof}

It only remains to prove \Cref{cor:low-mass-outside-sell-2}. Since the ideas are quite similar to arguments already employed in \Cref{sec:slctstep2} and \Cref{sec:exactstep1}, we only sketch the proof to avoid repetition.
\begin{proof}[Proof of \Cref{cor:low-mass-outside-sell-2}]
    Similar to the definition of $\Hg{L}$ in \Cref{lem:balancedness-hl}, we define $
        F_L(T) \coloneqq G[( U \setminus V(\cliqueset{T}) ) \uplus T ] \cap N_G(T).
    $ 
    Similarly as in the proof of \Cref{lem:multipartite-balancedness}, we get from \Cref{lem:general-balancedness-generalized} with $\eta = 2$, $L(T) = \cliqueset{T}$, and $S(T) = V(\cliqueset{T})$, that $F_L(T)$ has $3$-fold balancedness 
    $$
        \frac{16tp^{-2}}{\alpha} \maxbal{3} \cdot  p^{t}k + o(k),
    $$ and w.h.p., this holds for all $T$.
    Using the same calculation as in \Cref{lem:balancedness-hl}, we then get that for our choice of $t$, $F_L(T)$ has $3$-fold balancedness at most $\maxbalallcompact{3}\frac{k}{4}$. Then, employing the same reasoning as in \Cref{cor:low-mass-outside-sell} yields the desired statement. \qedhere
\end{proof}

\section{Analysis for approximate recovery}\label{sec:stronger-adv}
    \addtocounter{algorithm}{-0} 
    
  \begin{algorithm}
    \setlength\itemsep{.0001em}
    \caption{ Recovering the labels of an RIG under a bounded adversary. }
    \label[algorithm]{alg:recovery-approx}
    \begin{description}
    \setlength\itemsep{.0001em}
    \item[Given:] A graph $G \sim \RIGone$ on vertex set $[n]$ modified by an adversary that alter up to $\epsedge k$ edges incident to every $v \in [n]$ and arbitrarily alter the neighbourhood of up to $\epsnode k$ vertices. In addition to that, a monotone adversary can delete an arbitrary amount of edges $\{u, v \} \in E(G)$ such that $M_v \cap M_u = \emptyset$. Denote the modified graph obtained in this way by $\adv{G}$,
    \item[Output:]
      A list $\mathcal{L}$ of $d$ sets of vertices in $[n]$ such that for every $\ell \in [d]$, there is exactly one $S \in \mathcal{L}$ such that $|S_\ell \triangle S| \le O(\varepsilon k)$.
    \item[Operation:]\mbox{}
    \begin{enumerate}
    \item Split $[n]$ into two disjoint sets $U, V$ using an independent, fair coin flip for every $v \in [n]$.
    \item Set $\gamma \coloneqq \frac{1}{\alpha}\max\{ 16p^{-6}, 1000 \} \max\{1,  (p/1-p)^6\}$, and $t = 2\lceil \log_{1/p}(\gamma) \rceil$. Choose a list $\mathcal{T}$ of $O( n^{\alpha(1+ t/2)} )$ $t$-tuples from $\binom{[n]}{t}$ uniformly at random.
    \item Let $\ktil = \kapprox$. For each $T \in \mathcal{T}$, compute a degree-$O(t)$ pseudo-distribution $\mu$ on variables $w$ satisfying the axioms $\bicliqueaxiomsalgoapproxconcrete$ (see \eqref{eq:bicliqueaxiomsapprox} for a definition) that maximizes $\Expectedtildesub{\mu}{|w|}$.
    \item Let $S_T = \{ v \in [n] \mid \Expectedtildesub{\mu}{w_v} \ge 1 - \smallc \}$ for some small constant $\smallc \le \frac{1}{8t}$. 
    \item If $|S_T| < \kfoundsize k$, discard $S_T$ and continue with the next $T \in \mathcal{T}$. Otherwise, add $S_T$ to the output $\mathcal{L}$. 
    \item Pruning: For every $S \in \mathcal{L}$, delete from $\mathcal{L}$ all sets $S' \neq S$ such that $|S \triangle S'| \le 10 \epsedge / \rho$.
  \end{enumerate}
    \end{description}
  \end{algorithm}

We show that \Cref{alg:recovery-approx} is robust against a stronger adversary.

\subsection{Preliminaries}

Before starting with the analysis, we introduce some preliminaries.

\paragraph{Notation}
Throughout the rest of this section, we denote by $G$ a typical sample from $\RIGone$, and by $\adv{G}$ the resulting graph after an adversary has added or deleted up to $\varepsilon k^2$ edges in $G$. Given a vertex $v \in [n]$, we denote by $\adv{\deg}(v)$ the number of edges the adversary added or deleted \emph{incident to} $v$. Moreover, for any $\modified \in \mathbb{R}$, we denote by $M_{\le \modified}$ the set of vertices $v \in [n]$ such that $\adv{\deg}(v) \le \modified$. Similarly, $\Madv$ is set of vertices $v \in [n]$ such that $\adv{\deg}(v) \ge \modified$.

\paragraph{Adversarial Capabilities}

In addition to a monotone adversary, we allow further edge corruptions as follows. 
We fix some small constant $\epsedge = \epsedge$ and allow an adversary to arbitrarily modify edges with the constraint that $\adv{\deg}(v) \le \epsedge k$ for every $v \in [n]$. On top of that, we fix a second constant $\epsnode$ and allow the adversary to modify up to $\epsnode k$ vertices $v \in [n]$ with \emph{unbounded} adversarial degree $\adv{\deg}(v)$.

Note that this is more general than simply allowing node corruptions. It further covers an adversary that can modify $\varepsilon k^2$ edges arbitrarily. For some $\varepsilon = \varepsilon(\epsedge, \epsnode)$. To see this note that allowing for $\varepsilon k^2$ edge corruptions means that there are at most $ \frac{2\varepsilon}{\epsedge} k$ vertices with $\adv{\deg}(v) \ge \epsedge k$. Thus, if we can handle our more general adversary with parameters $\epsedge, \epsnode$, then we can handle $\varepsilon k^2$ edge corruptions for $\varepsilon = 2\epsedge \epsnode$. 
The main theorem of this section can then be stated as follows.
\begin{theorem}[Approximate recovery]\label{thm:approxrecovery2}
    There exists a polynomial time algorithm that on input $G \sim \RIGone$ possibly modified by a monotone adversary and the adversary described above, achieves $O(\epsnode k)$-approximate recovery. This holds with high probability over the choice of $G$ and for all sufficiently small constants $\epsedge, \epsnode$. 
\end{theorem}
\begin{remark}
    Notice that the above implies that for $\epsnode = 0$, we obtain \emph{exact recovery}, even under an adversary, which is allowed to corrupt up to $\epsedge k$ edges incident to every vertex. 
\end{remark}

With \Cref{thm:approxrecovery2}, \Cref{thm:approxrecovery} follows as a special case as shown by the following argument.
\begin{proof}[Proof of \Cref{thm:approxrecovery} with \Cref{thm:approxrecovery2}]
    \Cref{thm:approxrecovery2} implies \Cref{thm:approxrecovery} as follows. To handle $\varepsilon k^2$ edge corruptions, we set $\epsedge$ to a fixed, sufficiently small constant as in \Cref{thm:approxrecovery2}. Then, set $\epsnode = \frac{\varepsilon}{2 \epsedge}$ and choose $\varepsilon$ small enough such that $\epsnode$ is small enough such that \Cref{thm:approxrecovery2} applies. With these parameters, we claim that $\varepsilon k^2$ edge corruptions are a special case of the more general adversarial model introduced above. To see this, we simply note that the number of vertices with $\adv{\deg}(v) \ge \epsedge k$ is at most $\frac{2\varepsilon}{\epsedge} k$ as otherwise we contradict the fact that there are at most $\varepsilon k^2$ edge corruptions. The recovery guarantee we obtain from \Cref{thm:approxrecovery2} is then $O(\epsnode k) = O(\varepsilon k)$, as desired.
\end{proof}
The rest of this section is dedicated to proving \Cref{thm:approxrecovery2}.

\paragraph{Axioms}
Our approach towards achieving robustness requires a couple more ideas, which we formally describe in this section. Specifically, instead of the axioms $\bicliqueaxiomsalgo$ for exact recovery, we use the following system of axioms given a graph $G$, a bipartite sub-graph $H = (A \uplus B)$ of $G$ with $A \cup B = [n]$, a set $T \subseteq [n]$, and parameters $k, \gamma$
\begin{align}\label{eq:bicliqueaxiomsapprox}
    \bicliqueaxiomsalgoapprox = \left\{ 
        \begin{aligned}
            &\forall v \in [n] & w_v^2 &= w_v \\
            &                & 2k \ge \sum_{v \in A} w_v &\ge k\\
            &                & 2k \ge \sum_{v \in B} w_v &\ge k\\
            &\forall u, v \in [n], u \neq v  & z_{u,v}^2 &= z_{u,v}\\
            &\forall u, v \in [n], u \neq v \text{ s.t. } \{u, v\} \notin E(G) & w_uw_v (1 - z_{u,v})  &= 0\\
            &\forall u \in A                  & \sum_{v \in [n] \setminus \{u\}} z_{u,v} &\le \gamma
        \end{aligned}
    \right\}.
\end{align}
While the first four constraints are standard to encode a $k\times k$ biclique, the rest is specifically tailored towards ensuring robustness to our adversary. The idea is as follows.
\begin{enumerate}
    \setlength\itemsep{0.0001em}
    \item The variables $z_{u, v}$ are used to encode ``adding back'' edges that were deleted by the adversary. The objective  $w_uw_v (1 - z_{u,v})  = 0$ states that the $w$ should form a biclique after adding back the edges represented by $z$.
    \item The conditions $\sum_{v \in [n] \setminus \{u\}} z_{u,v} \le \gamma$ encode the constraint that incident to any vertex $v$, we only add back up to $\gamma$ edges.
\end{enumerate}

\subsection{Key SoS-certificates}

Before starting with the analysis, we the most important certificates used during the analysis. The first is an analogue of \Cref{lem:sos-left-certificate}. The difference is that we are additionally given a $t$-tuple $T$, and deal with a different bipartite sub-graph of $G$.
\begin{lemma}\label{lem:sos-left-certificate-approx}
    Fix some integer $t$. Assume we are given a graph $G$ on vertex set $[n]$, a bipartite sub-graph $H = (A \uplus B, E_H)$ of $G$, a set $T \subseteq [n]$ with $|T| = t$, and parameters $k, \gamma$.
    Assume further that there is some  $\Asub \subseteq A$ such that $G[\Asub \uplus (\neigh{T} \cap B)]$ has $r$-fold balancedness at most $\Delta \coloneqq \maxbalallcompact{r} \frac{10k}{12}$ for some fixed $r \in \mathbb{N}, r \ge 2$. Then, there is some $\varepsilon' > 0$ such that for all $\gamma \le \varepsilon' k$,
    \begin{align*}
        \bicliqueaxiomsalgoapprox \sststile{w}{\degreeapprox} \left\{ w_T \ \bigg( \sum_{v \in \Asub} w_v \bigg)^{r} \le O\left(\frac{n^2}{k}\right) \ w_T \right\}.
    \end{align*}
\end{lemma}

Moreover, we need the following certificate, which states that forcing a tuple $T$ to be in our clique implies that there are only a few ground-truth cliques that the clique we are searching for can overlap with.
It holds whenever we force our clique to contain any given $t$-tuple $T$ such that $\adv{\deg}(v) \le \modified $ for all $v \in T$. We call such $t$-tuples $\modified$-\emph{good}. The certificate is then stated as follows.
\begin{lemma}\label{lem:lowmassoutsidecliquesetstep2-approx}
    Let $T \in \binom{[n]}{t}$. Let further $\cliqueset{T} \subseteq [d]$ be the set of labels that appears in at least $3$ vertices of $T$, and let $V(\cliqueset{T}) = \bigcup_{\ell \in \cliqueset{T}} S_\ell$ be the set of vertices that have chosen at least one of these labels. Moreover, let $\adv{V}(\cliqueset{T}) \coloneqq V(\cliqueset{T}) \cup M_{\ge \modified}$ and let $\epsedge$ be small enough. 
    Then, with high probability over the draw of $G \sim \RIGone$, the following holds for all $\modified$-good $t$-tuples $T \in \binom{[n]}{t}$.
    $$
        \bicliqueaxiomsalgoapproxconcretereduced \sststile{w}{\degreeapprox} \bigg\{ w_T \ \bigg( \sum_{i \in [n] \setminus \adv{V}(\cliqueset{T}) }  w_i \bigg) \le O( \sqrt{n}) \ w_T \bigg\}.
    $$ 
\end{lemma}
\begin{remark}
    What is important to note about the above is that we explicitly do not add the axiom $\{x_T= 1\}$. The above holds for any $\modified$-good $T \in \binom{[n]}{t}$ \emph{without} the need to give $T$ as input to some algorithm. This also applies to the following \Cref{lem:overlapeight-approx}.
\end{remark}
After having established  ``concentration'' to the cliques in $\cliqueset{T}$ (with the certificate above), the following asserts that our clique cannot have large overlap with both $S_\ell$ and $[n] \setminus S_\ell$ for all labels $\ell \in \cliqueset{T}$. It states that forcing (in addition to $T$) a single vertex $i \in [n] \setminus S_\ell$ to be in our clique, the mass on $S_\ell$ is $o(k)$.
\begin{lemma}\label{lem:overlapeight-approx}
    With high probability, the following holds over the draw of $G \sim \RIGone$ and all sufficiently small $\epsedge$. For all $T \in \binom{[n]}{t}$ that are $\modified$-good, every $\ell \in \cliqueset{T}$, and all $i \in [n] \setminus (S_\ell \cup \Madv )$,
    \begin{align*}
        \bicliqueaxiomsalgoapproxconcretereduced \sststile{w}{\degreeapprox} \Bigg\{ w_Tw_i \ \bigg( \sum_{v \in S_\ell \setminus \Madv } w_v \bigg) \le  \beta(k) \ w_T w_i\ \Bigg\}.
    \end{align*} for some $\beta(k) = o(k)$.
\end{lemma}

\subsection{Pseudo-Concentration}

The key tool in our analysis (and one of the main difference to the previous section) is a \emph{concentration property} for pseudo-distributions satisfying the axioms in $\bicliqueaxiomsalgoapproxconcretereduced$. It is key to showing that we \emph{never} find a set of vertices that does not have large overlap with some ground-truth clique.
\begin{lemma}[Concentration lemma for pseudo-distributions]\label{lem:pseudoconentrationapprox}
    Let $\mu$ be a degree-$O(t)$ pseudo-distribution satisfying $\bicliqueaxiomsalgoapproxconcretereduced$. Fix a  constant $\smallc$ and let $S \coloneqq \big\{i \in [n] \mid \Expectedtildesubnop{\mu}{w_i} \ge 1 - \smallc \big\}$. Let further $\Madv$ be the set of vertices with $\adv{\deg}(v) \ge \modified$.
    
    Then, whenever $\epsedge, \epsnode$ are sufficiently small and $\smallc \le \frac{1}{8t}$, the following holds with high probability over the draw of $G \sim \RIGone$. For every degree-$O(t)$ pseudo-distribution $\mu$ satisfying $\bicliqueaxiomsalgoapproxconcretereduced$ such that $|S| \ge \Ssize$, there is some $\ell \in [d]$ such that $S \subseteq S_\ell \cup \Madv$.
\end{lemma}

The idea behind establishing the above is as follows. Since $\Madv\le \epsnode k$ is sufficiently small and $|S|$ is sufficiently large, there must be a $t$-tuple $T$ in $S \setminus \Madv$ such that $\Expectedtildesub{\mu}{w_T}= \Omega(1)$. Then, by the certificate in \Cref{lem:lowmassoutsidecliquesetstep2-approx}, this already implies that only a small part of $S$ is not in $\adv{V}(\cliqueset{T})$. 
This is already a concentration statement, which now can be ``bootstrapped'' into an even stronger one. Concretely, by the bounded number of labels in $\cliqueset{T}$, we can infer that there is some $\ell \in \cliqueset{T}$ such that $|S \cap S_\ell| = \Omega(k)$. Using the certificate from \Cref{lem:overlapeight-approx}, we can show that this implies that $S$ must actually be a subset of $S_\ell \cup \Madv$. 

\paragraph{Establishing pseudo-concentration}

For the proof, we need the following proposition relating the pseudo-expectation of a product to the product of pseudo-expectations, provided the  ``variance'' is small enough.

\begin{lemma}\label{lem:variance}
    Let $\mu$ be a degree-$2d$ pseudo-distribution. For any polynomial $f$, define 
    $$
        \Vartildesub{\mu}{f} \coloneqq \Expectedtildesub{\mu}{f^2} - \Expectedtildesub{\mu}{f}^2.
    $$ Then, for any degree two polynomials $f, g$ of degree $d$, \begin{align*}
        \Expectedtildesub{\mu}{fg} \ge \Expectedtildesub{\mu}{f}\Expectedtildesub{\mu}{g} - \frac{1}{2} \left( \Vartildesub{\mu}{f} + \Vartildesub{\mu}{g} \right).
    \end{align*}
\end{lemma}
\begin{proof}
    By  Hölder's inequality \Cref{lem:holder}, 
    \begin{align*}
        0 &\le \Expectedtildesub{\mu}{(f + g)^2} - \Expectedtildesub{\mu}{(f + g)}^2 \\
        &\hspace{2cm}= (\Expectedtildesub{\mu}{f^2} - \Expectedtildesub{\mu}{f}^2) + (\Expectedtildesub{\mu}{g^2} - \Expectedtildesub{\mu}{g}^2) + 2\Expectedtildesub{\mu}{fg} - 2\Expectedtildesub{\mu}{f}\Expectedtildesub{\mu}{g}. 
    \end{align*}
    Re-arranging and applying the definition of $ \Vartildesub{\mu}{f}$ then yields the statement.
\end{proof}

We start by showing that there is indeed some $t$-tuple $T$ in $S \setminus \Madv$ such that $\Expectedtildesub{\mu}{w_T}= \Omega(1)$.
\begin{lemma}[Constant expectation on some $T$]\label{lem:expectationonT}
    Assume the setting from \Cref{lem:pseudoconentrationapprox} and let further $\adv{S} \coloneqq S \setminus \Madv$. Then for all $0 < \smallc < 1$ and all $\epsnode \le \frac{1}{3}$, there is some $T \in \adv{S}$ such that $\Expectedtildesub{\mu}{w_T} \ge (1-\smallc)^t - o(1)$.
\end{lemma}
\begin{proof}
    By construction of $S$, we have $
        \Expectedtildesubnop{\mu}{\sum_{v \in \adv{S}} w_v} \ge (1-\varepsilon)|\adv{S}|.
    $ Hence, by Hölder's inequality \Cref{lem:holder},
    \begin{align*}
        (1-\smallc)^t |\adv{S}|^t \le \Expectedtildesubtwo{\mu}{\sum_{v \in \adv{S}} w_v}^t &\le \Expectedtildesubtwo{\mu}{\bigg(\sum_{v \in \adv{S}} w_v\bigg)^t} = t!  \sum_{s=1}^t \sum_{T \subseteq \adv{S}, |T| = s} \Expectedtildesub{\mu}{w_T} \le \\
        & \hspace{1cm} t! \binom{|\adv{S}|}{t} \max_{T \subseteq \adv{S}, |T| = s} \Expectedtildesub{\mu}{w_T}  + t!|\adv{S}|^{t-1} \le |\adv{S}|^{t} \max_{T \subseteq \adv{S}, |T| = s} \Expectedtildesub{\mu}{w_T}  + t!|\adv{S}|^{t-1}.
    \end{align*}
    Dividing by $|\adv{S}|^t$, we get 
    $$
        (1-\smallc)^t \le \max_{T \subseteq \adv{S}, |T| = s} \Expectedtildesub{\mu}{w_T} \ + t!/|\adv{S}|.
    $$
    Since $|S| \ge \Ssize$ and $|\Madv| \le \frac{1}{3}k$, we get $\adv{S} \ge \frac{1}{6}k$, so in fact $t!/|\adv{S}| = o(1)$ and $$
        \max_{T \subseteq \adv{S}, |T| = s} \Expectedtildesub{\mu}{w_T} \ge (1-\smallc)^t - o(1),
    $$ as desired. 
\end{proof}

Next, we use the certificate from \Cref{lem:lowmassoutsidecliquesetstep2-approx} to show the following preliminary concentration result stating that only few vertices in $S$ are outside of $\adv{V}(\cliqueset{T})$ for some $T$.

\begin{lemma}\label{lem:pseudoconcentrationstep1}
    Assume the setting from \Cref{lem:pseudoconentrationapprox} and let further $\adv{S} \coloneqq S \setminus \Madv$. Then whenever $\ \epsedge$ is sufficiently small and $\smallc \le \frac{1}{4t}$, there is some $T \subseteq \adv{S}$ with $|T| = t$ such that $$
        \big| S \setminus \adv{V}(\cliqueset{T}) \big| \le O(\sqrt{n}).
    $$
\end{lemma}
\begin{proof}
    Let $T \subseteq \adv{S}$ with $|T| = t$ such that $\Expectedtildesub{\mu}{w_T} \ge (1-\smallc)^t - o(1)$, which exists due to \Cref{lem:expectationonT}.
    The certificate from \Cref{lem:lowmassoutsidecliquesetstep2-approx} now tells us that \begin{align}\label{eq:pseudostuff}
        \sum_{ i \in S\setminus \adv{V}(\cliqueset{T}) } \frac{\Expectedtildesub{\mu}{w_Tw_i}}{\Expectedtildesub{\mu}{w_T}} \le  \sum_{ i \in [n] \setminus \adv{V}(\cliqueset{T}) } \frac{\Expectedtildesub{\mu}{w_Tw_i}}{\Expectedtildesub{\mu}{w_T}} \le  O(\sqrt{n}).
    \end{align}
    We now wish to find a lower bound on the quantity $\frac{\Expectedtildesub{\mu}{w_Tw_i}}{\Expectedtildesub{\mu}{w_T}}$ in the leftmost sum. 

    Towards this end, we use \Cref{lem:variance} to obtain 
    $$
        \frac{\Expectedtildesub{\mu}{w_Tw_i}}{\Expectedtildesub{\mu}{w_T}} \ge \Expectedtildesub{\mu}{w_i} - \frac{1}{2}\frac{\Vartildesub{\mu}{w_T} + \Vartildesub{\mu}{w_i}}{\Expectedtildesub{\mu}{w_T}}.
    $$ We use that
    $
        \Vartildesub{\mu}{w_T} = \Expectedtildesub{\mu}{w_T} ( 1 - \Expectedtildesub{\mu}{w_T} )
    $ because $\Expectedtildesub{\mu}{w_T^2} = \Expectedtildesub{\mu}{w_T}$, and that $
        \Vartildesub{\mu}{w_v} = \Expectedtildesub{\mu}{w_i^2} - \Expectedtildesub{\mu}{w_i}^2 \le 2\smallc 
    $ because $\Expectedtildesub{\mu}{w_i^2} \le 1$ and $\Expectedtildesub{\mu}{w_i}^2 = (1-\smallc)^2 \ge 1- 2\smallc$. Plugging these bounds in yields that 
    $
        \Expectedtildesub{\mu}{w_Tw_i}/\Expectedtildesub{\mu}{w_T} > 0
    $ for all $t \ge 1$ and $\smallc \le \frac{1}{4t}$. Since for this choice of parameters, this holds uniformly for all $i \in S \setminus \adv{V}(\cliqueset{T})$, we get from \eqref{eq:pseudostuff} that there is some constant $c > 0$ such that 
    $$
        c \ |S \setminus \adv{V}(\cliqueset{T})| \le \sum_{ i \in S\setminus \adv{V}(\cliqueset{T}) } \frac{\Expectedtildesub{\mu}{w_Tw_i}}{\Expectedtildesub{\mu}{w_T}} \le  O(\sqrt{n})
    $$ and the lemma follows.
\end{proof}

Now comes the ``bootstrapping'' step: since by the above most of $S$ concentrates on $\adv{V}(\cliqueset{T})$, the balancedness properties among the cliques with labels in $\cliqueset{T}$ (captured in the certificate from \Cref{lem:overlapeight-approx}) yield a significantly improved concentration.

\begin{lemma}[Bootstrapping the intersection of $S$ and $S_\ell$]\label{lem:pseudoconentrationapprox-2}
    Assume the setting from \Cref{lem:pseudoconentrationapprox}.
    Assume further that there is some constant $C > 0$ and $\ell \in $ such that $\ell \in \cliqueset{T}$ such that $|S \cap (S_\ell \setminus M_{\ge \modified})| \ge Ck$. 
    Then, whenever $\epsedge$ is sufficiently small and $\smallc \le \frac{1}{8t}$, we have $S \subseteq S_\ell \cup \Madv$.
\end{lemma}
\begin{proof}
    We want to show that $S \setminus (S_\ell \cup M_{\ge \modified}) = \emptyset$. To this end assume for the sake of contradiction that there is a vertex $i \in S \setminus (S_\ell \cup M_{\ge \modified})$. Let further $T \subseteq S \setminus \Madv$ be a $t$-tuple with $\Expectedtildesub{\mu}{w_T} \ge (1-\smallc)^t - o(1)$, which exists due to \Cref{lem:expectationonT}.

    By duality between SoS-proof and pseudo-distributions, the certificate from \Cref{lem:overlapeight-approx} then yields that 
    $$
         \sum_{v \in S \cap (S_\ell \setminus \Madv)} \Expectedtildesub{\mu}{ w_Tw_i w_v } \le \sum_{v \in S_\ell \setminus \Madv} \Expectedtildesub{\mu}{ w_Tw_i w_v } \le \beta(k) \Expectedtildesubnop{\mu}{w_Tw_i} \le \beta(k)\Expectedtildesubnop{\mu}{w_T}.
    $$
    Similarly as in the proof of \Cref{lem:pseudoconcentrationstep1}, we get 
    $$
        \frac{\Expectedtildesub{\mu}{w_Tw_iw_v}}{\Expectedtildesub{\mu}{w_T}} \ge \Expectedtildesub{\mu}{w_iw_v} - \frac{1}{2}\frac{\Vartildesub{\mu}{w_T} + \Vartildesub{\mu}{w_iw_v}}{\Expectedtildesub{\mu}{w_T}}.
    $$ Again, we use $
        \Vartildesub{\mu}{w_T} = \Expectedtildesub{\mu}{w_T} ( 1 - \Expectedtildesub{\mu}{w_T} ) \le \Expectedtildesub{\mu}{w_T} ( 1 - (1 - \smallc)^t + o(1) ).
    $ Moreover, 
    $$
        \Expectedtildesub{\mu}{w_iw_v} \ge \Expectedtildesub{\mu}{w_i}\Expectedtildesub{\mu}{w_v} - \frac{1}{2}(\Vartildesub{\mu}{w_i} + \Vartildesub{\mu}{w_v}) \ge 1 - 2\smallc -\frac{1}{2}(4\smallc) \ge 1 - 4\smallc.
    $$ Accordingly,
    $$
        \Vartildesub{\mu}{w_iw_v} = \Expectedtildesub{\mu}{w_iw_v}(1 - \Expectedtildesub{\mu}{w_iw_v}) \le 4\smallc.
    $$ Plugging these bounds in now yields that 
    $
        \Expectedtildesub{\mu}{w_Tw_iw_v}/\Expectedtildesub{\mu}{w_T} > 0
    $ for all $t \ge 1$ and $\smallc \ge \frac{1}{8t}$. Thus, since this holds uniformly for all $v$, going back to the beginning yields 
    $$
         c \ |S \cap (S_\ell \setminus \Madv)| \le \sum_{v \in S \cap (S_\ell \setminus \Madv)} \frac{\Expectedtildesub{\mu}{ w_Tw_i w_v }}{\Expectedtildesub{\mu}{w_T}} \le  \beta(k).
    $$
    Since by assumption $|S \cap (S_\ell \setminus \Madv)|\ge Ck $, this is a contradiction since the left side is $\Omega(k)$ and $\beta(k) = o(k)$.
\end{proof}

With this, we prove \Cref{lem:pseudoconentrationapprox}.
\begin{proof}[Proof of \Cref{lem:pseudoconentrationapprox}]
    
    From \Cref{lem:pseudoconcentrationstep1}, we get that there is a $t$-tuple $T \subseteq S \setminus \Madv$ such that $\big| S \setminus \adv{V}(\cliqueset{T}) \big| \le O(\sqrt{n})$ and $\Expectedtildesub{\mu}{w_T} \ge (1-\smallc)^t - o(1)$.
    By \Cref{lem:boubdbadlabels}, we get that w.h.p., $|\cliqueset{T}| \le \frac{4t}{\alpha}$ for all $T$. Hence for some $\ell \in \cliqueset{T}$, we have 
    \begin{align}\label{eq:overlap}
        |S \cap (S_\ell \setminus M_{\ge \modified})| \ge \big(\frac{1}{3} k - |\Madv|\big) \frac{\alpha}{4t} \ge \frac{\alpha}{16 t} k
    \end{align} for sufficiently small $\epsnode$. Now, it follows from \Cref{lem:pseudoconentrationapprox-2} that $S \subseteq S \cup \Madv$, as desired.
\end{proof}

\subsection{Analysis}

\subsubsection{Step 1: Ground-truth cliques are found}

Again, we rely on balancedness properties. However, clearly we cannot hope to have the same balancedness guarantees as before since an adversary can easily destroy the balancedness properties of a given bipartite graph $F = A \uplus B$ already after altering the neighbourhood of only a constant number of vertices. However, it is not hard to see that after removing a small subset of vertices $M \subseteq A$ (small in the sense that $|M| \le \epsedge k$), the graph $(A \setminus M) \uplus B$ still has sufficiently good balancedness properties. We capture this in the following lemma. 

\begin{lemma}[Balancedness after a bounded adversary]\label{lem:balancedness-adversary}
    Let $F = (A \uplus B, E)$ be a bipartite graph and fix some integer $r \in \mathbb{N}, r \ge 3$ such that $F$ has $r$-fold balancedness $\Delta$ (cf. \Cref{def:balancedness}). Denote by $\adv{F}$ the graph $F$ after the adversarial modifications and recall that $M_{\ge \modified}$ denotes the set of vertices with $\adv{\deg}(v) \ge \modified$. 
    Then, given any $\epsedge > 0$, the graph $F[(A \setminus M_{\ge \modified}) \uplus B]$ has $r$-fold balancedness at most  
    $$\Delta + 4r\epsedge k \maxbal{r}. $$
\end{lemma}
\begin{proof}
    Note that for every $S, R \subseteq A \setminus M_{\ge \modified}$ with $|S| = |R| = r$, for all but at most $2r \modified$ vertices $v \in B$, the edges between $v$ and $S \cup R$ remain unchanged because $\adv{\deg}(u) \le \modified$ for $u \in A \setminus M_{\ge \modified}$. For every $v \in B$ for which this is \emph{not} the case, the quantity $u_{S,p}(v)u_{R,p}(v)$ can change by at most $2 \maxbal{r}$. Accordingly, 
    $$
        \left|\sum_{v \in B} u_{S,p}(v)u_{R,p}(v)\right| \le \Delta + 4r\epsedge k \maxbal{r}.
    $$
\end{proof}

We continue start by showing that after step 2 of our algorithm, we again find at least one $t$-tuple $T \in \mathcal{T}, T \subseteq S_\ell$ for each $\ell \in [d]$ such that most of $S_\ell$ is preserved in $\neigh{T}$ while the graphs $\Hgtil{L}, \Hgtil{R}$ (defined analogously to $\Hg{L}, \Hg{R}$ in the following) are still sufficiently well balanced. These will later serve as a witness for \Cref{lem:sos-left-certificate-approx} which in turn yields desirable behaviour of our pseudo-distribution.

\begin{lemma}\label{lem:balancedness-core-adv}
  Given $\modified \in \mathbb{N}$ and any $T \in \mathcal{T}$ such that $T \subseteq S_\ell$, define $\widetilde{S} \coloneqq S_\ell \cup M_{\ge \modified}$
    \begin{align*}
        \Hgtil{L} &\coloneqq \adv{G}[ (U \setminus \widetilde{S} ) \uplus (V \cap N_{\adv{G}}(T)) ] \\ \text{ and } \Hgtil{R} &\coloneqq \adv{G}[ (V \setminus \widetilde{S}) \uplus (U\cap N_{\adv{G}}{T}) ] 
    \end{align*}
    With high probability, the following holds after step 2 of \Cref{alg:recovery-approx}.
    For sufficiently small $\epsnode$, every $\epsedge$, and every $\ell \in [d]$, there is some $T \in \mathcal{T}, T \subseteq S_\ell$ such that $\Hgtil{L}$ and $\Hgtil{R}$ have $3$-fold balancedness at most $$
        \maxbalall{3} \frac{k}{4} + (2t\epsedge + 4 r \epsedge ) k \maxbal{3}.
    $$ Furthermore, for all $v \in T$, we have $\adv{\deg}(v) \le \modified$.
\end{lemma}
\begin{proof}
    First of all, we can use \Cref{lem:general-balancedness} to derive that for all $\ell \in [d]$, a $1 - o(1)$ fraction of $t$-tuples $T \in \binom{S_\ell}{t}$ are such that 
    \begin{align*}
        \Hgtilprime{L} &\coloneqq G[ (U \setminus \widetilde{S} ) \uplus (V \cap N_{G}(T)) ] \\ \text{ and } \Hgtilprime{R} &\coloneqq G[ (V \setminus \widetilde{S}) \uplus (U\cap N_{G}(T)) ] 
    \end{align*}
    have $3$-fold balancedness $\Delta \coloneqq \maxbalallcompact{3} \frac{k}{4}$ (like in the proof of \Cref{lem:balancedness-hl}). Notice the subtle difference that we have replaced $\adv{G}$ by $G$ in $\neigh{T}$ above.

    The rest of  the proof is dedicated to addressing this difference. To this end, note that for sufficiently small $\epsnode$, the fraction of $t$-tuples $T$ in $S_\ell$ such that for all $v \in T$, we have $v \notin \Madv$ while at the same time $\Hgtilprime{L}, \Hgtilprime{R}$ have the desired $r$-fold balancedness is still at least $\frac{1}{2}$. Thus, by the same analysis as in \Cref{lem:balancedness-hl}, with high probability, we find a such tuple in $\mathcal{T}$ for each $\ell \in [d]$ after step 2.
    
    
    Accordingly, with high probability after step 2, $\mathcal{T}$ contains a tuple $T \subseteq S_\ell$ for all $\ell \in [d]$ such that $\adv{\deg}(v) \le \modified$ for all $v\in T$, while $\Hgtilprime{L}, \Hgtilprime{R}$ have the same balancedness guarantees as in \Cref{lem:balancedness-hl}. To lift this result to $\Hgtil{L}, \Hgtil{R}$, we first consider the intermediate graphs
    \begin{align*}
        \Hgtilprimeprime{L} &\coloneqq G[ (U \setminus \widetilde{S} ) \uplus (V \cap N_{\adv{G}}(T)) ] \\ \text{ and } \Hgtilprimeprime{R} &\coloneqq \adv{G}[ (V \setminus \widetilde{S}) \uplus (U\cap N_{G}{T}) ].
    \end{align*}
    Here, by the fact that $\adv{\deg}(v) \le \modified$ for $v \in T$, we get that the right hand side of $\Hgtilprimeprime{L}$ differs by at most $t\modified$ vertices compared to $\Hg{L}$. Therefore, the $3$-fold balancedness of $\Hgtilprimeprime{L}$ is at most 
    $$
        \Delta + 2t \epsedge k \maxbal{3}.
    $$ Now, the desired balancedness guarantees for $\Hgtil{L}$ directly follow after applying \Cref{lem:balancedness-adversary} to $\Hgtilprimeprime{L}$. The same argument applies to $\Hgtil{R}$.
\end{proof}

Now, we can use our SoS certificates from \Cref{sec:sos} to conclude that any pseudo-distribution can not put a lot of ``mass'' on the left hand sided vertices of $\Hg{L}$ and $\Hg{R}$, respectively. This is captured in the following 
analogue of \Cref{cor:low-mass-outside-sell}.

\begin{lemma}\label{lem:low-mass-adv}
    Assume that $T$ is a $t$-tuple $T \subseteq S_\ell$ such that $\Hgtil{L}$, and $\Hgtil{R}$ (cf. \Cref{lem:balancedness-core-adv}) meet the balancedness criteria from \Cref{lem:balancedness-core-adv}. Then for sufficiently small $\epsedge$ ,
    $$
        \bicliqueaxiomsalgoapproxconcrete \sststile{w}{\degreeapprox} \bigg\{ \sum_{i \in U \setminus \widetilde{S}} w_i, \sum_{i \in V \setminus \widetilde{S} } w_i \le O( \sqrt{n}) \bigg\}
    $$
    and the bit complexity of the proof is $n^{O(1)}$. 
\end{lemma}
\begin{proof}
    Similarly as in the proof of \Cref{cor:low-mass-outside-sell}, we get from \Cref{lem:balancedness-core-adv} that 
    $\Hgtil{L}, \Hgtil{R}$ are subgraphs of $H(T)$ with $r$-fold balancedness at most $ \maxbalallcompact{3} \frac{k}{3}$ if choosing $\epsedge$ small enough (this works since we assume $p$ to be bounded away from $1$).
    Thus, invoking \Cref{lem:sos-left-certificate-approx} with $r = 3$ to $\Hgtil{L}$, we get that $$
        \bicliqueaxiomsalgoapproxconcrete \sststile{x, y}{\degreeapprox} \Bigg\{ \bigg( \sum_{i \in U \setminus \widetilde{S}} w_i \bigg)^{3} \le O\bigg(\frac{n^2}{k}\bigg) \Bigg\},
    $$
    Again, combining this with the fact that $\bicliqueaxiomsalgoapproxconcrete \sststile{x,y}{\degreeapprox} \{ \sum_{i \in U \setminus \widetilde{S}} x_i \le 2k \}$, we get that 
    $$
        \bicliqueaxiomsalgoapproxconcrete \sststile{x, y}{\degreeapprox} \left\{ \bigg(\sum_{i \in U \setminus \widetilde{S} } w_i\bigg)^{4} \le O( n^2 ) \right\}.
    $$ Applying the cancellation $\{ f^2 \le C  \} \sststile{f}{2} \{ f \le \sqrt{C} \}$ from \Cref{lem:cancellation} twice, we arrive at the desired conclusion. An analogous argument for $\Hgtil{R}$ also yields the second part of the statement.
\end{proof}

\newcommand{\Sellgood}{\hat{S}}



Finally, we can combine all of the above to show that we can approximately recover all $S_\ell$. 

\begin{lemma}\label{lem:recover-cliques2}
    Assume $T$ is a $t$-tuple $T \subseteq S_\ell$ such that $\Hgtil{L}$, and $\Hgtil{R}$ meet the balancedness criteria from \Cref{lem:balancedness-hl}. Then, after step 5 and for all sufficiently small $\epsedge, \epsnode$ and all $\smallc$, we have $|S_T \triangle S_\ell| \le (\mkfoundsize) k $. Moreover, $S_T \subseteq S_\ell \cup \Madv$.
\end{lemma} 
\begin{remark}
    Notice that the above implies that for $\epsnode = 0$, there is some $\epsedge > 0$ such that $|S_T \triangle S_\ell| = 0$. This yields \emph{exact recovery}, even under the adversary.
\end{remark}
\begin{proof}
    Like in the proof of \Cref{lem:recover-cliques}, we observe that for every $S_\ell$ after step 1 of the algorithm, we have $ \frac{1-\varepsilon}{2}k \le  |S_\ell \cap U|, |S_\ell \cap V| \le \frac{1+\varepsilon}{2} k$ for every fixed $\varepsilon > 0$, with high probability. 

    Hence, there is at least one pseudo-distribution $\mu$ in axioms $\bicliqueaxiomsalgoapproxconcrete$, set for example $w_i = 1$ if and only if $i \in S_\ell \setminus M_{\ge \modified}$, and $z_{u, v} = 1$ if and only if $x_u = 1, y_v = 1$, and $\{u, v\} \notin E(\adv{G}[U \uplus V])$. This pseudo-distribution achieves objective $ \Expectedtildesub{\mu}{|w|}  = |S_\ell \setminus \Madv|$.

    
    
    Using \Cref{lem:low-mass-adv} and recalling that $\widetilde{S} = S_\ell \cup M_{\ge \modified}$, we get that $\sum_{i \in [n] \setminus \widetilde{S}} \Expectedtildesubnop{\mu}{w_i} \le O(\sqrt{n})$ if $\mu$ is a degree-$\degreeapprox$ pseudo-distribution in axioms $\bicliqueaxiomsalgoapproxconcrete$. Accordingly, 
    $$
        \frac{1}{|S_\ell \cup \Madv|}\sum_{i \in S_\ell \cup M_{\ge \lambda} } \Expectedtildesubnop{\mu}{w_i} \ge \frac{|S_\ell \setminus \Madv| - O(\sqrt{n})}{|S_\ell \cup \Madv|} \ge \frac{|S_\ell| - \epsnode k - O(\sqrt{n})}{|S_\ell| + \epsnode k} 
    $$ 
    Choosing $\epsnode$ small enough, the right hand side is at least $1 - 4\epsnode$. Thus, by Markov's inequality there are at most $4\epsnode k / \rho$ vertices $v$ in $S_\ell \cup \Madv$ such that $\Expectedtildesubnop{\mu}{w_i} \le 1 - \rho$. Hence, for sufficiently small $\epsnode$ the overlap $|S_T \cap (S_\ell \setminus \Madv)| = \Omega(k)$, so by \Cref{lem:pseudoconentrationapprox-2}, we have $S_T \subseteq S_\ell \cup \Madv$.
    This implies further that  the set $S_T$ satisfies $|S_\ell \triangle S_T| \le \epsnode k + (4 \epsnode / \rho) k$, as desired.  \qedhere

\end{proof}

\subsubsection{Step 2: No other cliques are found} 

We proceed by showing that the the algorithm never finds a set $S_T$ that does not have large overlap with some ground-truth clique.
\begin{lemma}\label{lem:step2analysisexact}
    Let $T \in \binom{[n]}{t}$ be any $t$-tuple in $G \sim \RIGone$. Then, after step 5 of \Cref{alg:recovery-approx}, the set $S_T$ is either discarded or such that $|S_T \triangle S_\ell| \le (\mkfoundsize) k$ for some $\ell \in [d]$. 
\end{lemma}
The challenge for the above is that--even for $k \gg n^{\frac{1}{2} + \varepsilon}$--the single label clique theorem is not applicable anymore due to the adversary we want to allow. We further want our algorithm to work all the way down to $k \gg \sqrt{n \log(n)}$, so a different approach is needed.
Luckily, our the pseudo-concentration statement from \Cref{lem:pseudoconentrationapprox} comes in handy and yields a short proof of \Cref{lem:step2analysisexact}

\begin{proof}[Proof of \Cref{lem:step2analysisexact}]
    We check if the $S_T$ found after step 5 has size at least $\kfoundsizetwo k$. If so, \Cref{lem:pseudoconentrationapprox} tells us that there is some $\ell \in \cliqueset{T}$ such that $S_T \subseteq S_\ell \cap \Madv$. Hence $|S_T \triangle S_\ell| \le (\mkfoundsizetwo) k + \epsnode k + o(k) \le (\mkfoundsize) k$, as desired. 
\end{proof}

Finally, we analyze the pruning in step 6.
\begin{lemma}
    There is a constant $C$ such that with high probability, the following holds after step 6 of \Cref{alg:recovery-approx}. We have $|\mathcal{L}| = d$ and for every $\ell \in \ell$, there is exactly one $S \in \mathcal{L}$ such that $|S \triangle S_\ell| \le (\mkfoundsize) k$. Moreover, if $|S \triangle S_\ell| \le (\mkfoundsize) \epsnode k$, then further $S \subseteq S_\ell \cup \Madv$.
\end{lemma}
\begin{proof}
    We get from \Cref{lem:recover-cliques2} and \Cref{lem:balancedness-core-adv} that with high probability, for every $\ell \in [d]$, there is some $S \in \mathcal{L}$ such that $|S \triangle S_\ell| \le (\mkfoundsize) k$ and $S \subseteq S_\ell \cup \Madv$. Moreover, by \Cref{lem:step2analysisexact}, the same guarantees hold for every $S \in \mathcal{L}$ for some $\ell \in [d]$.

    It therefore only remains to show that $|\mathcal{L}| = d$. To this end, we note that with high probability over the draw of $G$, we have that for all $\ell, \ell' \in [d]$ that $|S_\ell \cap S_{\ell'}| \le (1 + o(1))\delta k$. This implies that for sufficiently small $\smallc$, exactly one $S \in \mathcal{L}$ with $|S \triangle S_\ell| \le \kfoundsize k$ is kept. 
\end{proof}

\subsection{Deferred proofs}

We re-visit the SoS certificates whose prove was deferred above.

\subsubsection{Proof of \Cref{lem:sos-left-certificate-approx}}

For establishing \Cref{lem:sos-left-certificate-approx} we need the following proposition telling us that imposing the constraint $x_T = 1$ implies that only a small number of vertices is outside of $\neigh{T}$ as a function of $\gamma$.
\begin{lemma}\label{lem:sosknowsfewoutside}
    Given a biparite graph $H = (A \uplus B, E)$ and any $T \in A$ with $|T| = t$, 
    $$
        \bicliqueaxiomsalgoapprox \sststile{w}{O(t)} \Bigg\{ w_T \sum_{v \in B \setminus \neigh{T}} w_v  \le \gamma t \  w_T \Bigg\}.
    $$
\end{lemma}
\begin{proof}
    Observe that by the axiom $w_uw_v(1-z_{u,v}) = 0$,
    $$
        \bicliqueaxiomsalgoapprox \sststile{w,z}{\degreeapprox} \left\{ 0 = w_T\sum_{v \in B \setminus \neigh{T}} w_v\prod_{u \in T}(1-z_{u,v}) \right\}.
    $$
    However, using the bound from \Cref{clm:productsumlowerbound}, we get 
    $$
        \bicliqueaxiomsalgoapprox \sststile{w,z}{\degreeapprox} \left\{ \prod_{u \in T}(1-z_{u,v})  \ge 1 - \sum_{u \in T} z_{u,v} \right\}.
    $$ Thus, 
    \begin{align*}
        \bicliqueaxiomsalgoapprox \sststile{w,z}{\degreeapprox} \left\{\begin{aligned} 0 &= w_T\sum_{v \in B \setminus \neigh{T}} w_v\prod_{u \in T}(1-z_{u,v})\\ &\hspace{2cm}\ge w_T \ \left( \sum_{v \in B \setminus \neigh{T}} w_v  - \sum_{v \in B \setminus \neigh{T}}w_v\sum_{u \in T} z_{u,v} \right) \\
        &\hspace{2cm}\ge w_T \ \left( \sum_{v \in B \setminus \neigh{T}} w_v  - \sum_{u \in T} \sum_{v \in B \setminus \neigh{T}}  z_{u,v} \right) \end{aligned}\right\}.
    \end{align*}
    Due to the constraint $\sum_{v \in B} z_{u, v} \le \gamma$, we get  
    \begin{align*}
        \bicliqueaxiomsalgoapprox \sststile{x, y}{\degreeapprox} \left\{  0 \ge w_T \left( \sum_{v \in B \setminus \neigh{T}} w_v  - \gamma t \right) \right\}.
    \end{align*} Re-arranging then finishes the proof.
\end{proof}

\begin{proof}[Proof of \Cref{lem:sos-left-certificate-approx}]
    We denote by $|y_{\neigh{T}}| = \sum_{y \in B \cap \neigh{T}} y_v$ and 
    invoke \Cref{lem:core-sos-proof-adv} for on $H[A_{\text{(sub)}} \uplus B]$ to get
    \begin{align*}
        \bicliqueaxiomsalgoapprox \sststile{w,z}{8r} \left\{\begin{aligned} &\maxbalall{r} \left( \sum_{S \subseteq \Asub, |S| = r} w_S \right)^2 (| y_{\neigh{T}} | - 2\gamma r) \\ & \hspace{1cm} \le \left(n + rn^2\right) \maxbal{r} \left( \sum_{S \subseteq \Asub, |S| = r} w_S \right) + \Delta \hspace{.05cm}\left( \sum_{S \subseteq \Asub, |S| = r} w_S \right)^2 \end{aligned} \right\}.
    \end{align*}
    Now, we rely on the fact that by \cref{lem:sosknowsfewoutside}, $
        \bicliqueaxiomsalgoapprox \sststile{w,z}{O(t)} \left\{ w_T \sum_{v \in B \setminus \neigh{T}} w_v  \le \gamma t w_T \right\}.
    $ and thus, (recalling that $\gamma \le \varepsilon' k$),
    $$
        \bicliqueaxiomsalgoapprox \sststile{w,z}{O(t)} \left\{ w_T |y_{\neigh{T}}| \ge (k - \gamma t)w_T \ge (1- \varepsilon' t ) k \ w_T \ge (2 \gamma r + 11k/12 ) \ w_T \right\}
    $$ for sufficiently small $\varepsilon' > 0$. 
    Using this and our bound on $\Delta$ ($\Delta \le \maxbalallcompact{r}10k/12$), we get that
    $$
        \mathcal{B}(H, k ) \sststile{w,z}{\degreeapprox} \left\{ \maxbalall{r} w_T \ \left( \sum_{S \subseteq \Asub, |S| = r} w_S \right)^2 \frac{k}{12} \le (r+1)n^2 \maxbal{r}
        w_T \ \left( \sum_{S \subseteq \Asub, |S| = r} w_S \right) \right\}.
    $$ almost like in the proof of \Cref{lem:sos-left-certificate}. The rest of the proof then follows exactly the proof of \Cref{lem:sos-left-certificate} (notice that applying the cancellation inequalities remains valid since $w_T = w_T^2$).
\end{proof}

\subsubsection{Proof of \Cref{lem:lowmassoutsidecliquesetstep2-approx}}

\begin{proof}[Proof of \Cref{lem:lowmassoutsidecliquesetstep2-approx}]
        We define the auxiliary graph $$
            F_L(T) \coloneqq G[(U \setminus \adv{V}(\cliqueset{T}) \uplus (V \cap N_{G}(T)s))].
        $$
        Then, like in the proof of \Cref{lem:multipartite-balancedness} and \Cref{lem:step2analysisexact}, we get from \Cref{lem:general-balancedness-generalized} with $L(T) = \cliqueset{T}$ and $S(T) = V(\cliqueset{T})$ that the graph $F_L(T)$ has $3$-fold balancedness
        $$
            \frac{16tp^{-2}}{\alpha} \maxbal{3} \cdot  p^{t}k + o(k),
        $$
        for all $T$. By our choice of $t$ and the calculation in \Cref{lem:balancedness-hl}, this is at most
        $ \maxbalallcompact{3}\frac{k}{4}$. Switching from $F_L(T)$ to 
        $$
            \adv{F_L}(T) \coloneqq \adv{G}[(U \setminus \adv{V}(\cliqueset{T}) \uplus (V \cap N_{\adv{G}}(T)))].
        $$ and using the same argument as in \Cref{lem:balancedness-core-adv} (here we rely on the $\lambda$-goodness of $T$), we obtain a $3$-fold balancedness in $\adv{F_L}(T)$ of at most 
        $$
            \maxbalall{3} \frac{k}{4} + (2t\epsedge + 4 r \epsedge ) k \maxbal{3}.
        $$ 
        Now, we can use the same argument as in \Cref{lem:low-mass-adv} to get to our conclusion, provided that $\epsedge$ is small enough.
\end{proof}

\subsubsection{Proof of \Cref{lem:overlapeight-approx}}

\begin{proof}[Proof of \Cref{lem:overlapeight-approx}]
    For this lemma, we rely on balancedness properties among the cliques in $V(\cliqueset{T})$ like in \Cref{sec:slctstep2} and therefore re-introduce some notation from this section. 

    First, recall that  $\overlap{T}$ is the set of vertices in $G$ that have chosen at least two of the labels in $\cliqueset{T}$. For any $\ell \in \cliqueset{T}$, recall further that $\leftt = S_\ell \setminus \overlap{T}$ and $\rightt = \bigcup_{\ell' \in \cliqueset{T} \setminus \{\ell\}} S_{\ell'} \setminus \overlap{T}$.
    Now, balancedness properties hold in the bipartite graphs $H_\ell \coloneqq G[ \leftt \uplus \rightt]$ for all $\ell \in \cliqueset{T}$. Specifically, it follows from \Cref{lem:boundarybalancedness} that all the $H_\ell$ have $3$-fold balancedness $o(k)$.

    We lift this balancedness guarantee to $\adv{G}$ by considering the following bipartite graph for every $\ell \in \cliqueset{T}$
    \begin{align*}
        \adv{H}_\ell &\coloneqq G[ \lefttadv \uplus \righttadv] \\ \text{ with } \lefttadv &\coloneqq \leftt \setminus \Madv \text{ and } \righttadv \coloneqq \rightt \setminus \Madv.
    \end{align*}
    Noting that $\adv{\deg}(v) \le \modified$ for all $v \in V(\adv{H}_\ell)$, and applying \Cref{lem:balancedness-adversary}, we get that $\adv{H}_\ell$ has $3$-fold balancedness at most 
    \begin{align*}
        \Delta \coloneqq o(k) + \left(2t\epsedge + 12\epsedge\right) k \maxbal{3}.
    \end{align*}
    Then, we define $|y| \coloneqq \sum_{v \in \righttadv} w_v$ and apply the certificate from \Cref{lem:core-sos-proof-adv} to obtain (using $r = 3$) that
    \begin{align*}
         \bicliqueaxiomsalgoapproxconcretereduced \sststile{w}{O(t)} \left\{ \begin{aligned} &\left( \sum_{S \subseteq \lefttadv, |S| = 3} w_S \right)^2 \left( \maxbalall{3} |y| - \Delta - 12 \epsedge k  \right) \\ &\hspace{3cm} \le \left(n + 3n^2\right) \maxbal{3} \left( \sum_{S \subseteq 
         \lefttadv, |S| = 3} w_S \right) \end{aligned} \right\}.
    \end{align*}
    To get a lower bound on the left hand side, we use $|w| \ge k$ and the certificate from \Cref{lem:lowmassoutsidecliquesetstep2-approx} to obtain
    \begin{align*}
        \bicliqueaxiomsalgoapproxconcretereduced \sststile{w}{O(t)} \left\{ \begin{aligned} w_Tw_i|y| &\ge w_Tw_i \bigg( 2\ktil - \sum_{v \in S_\ell}w_v - \sum_{v \in [n] \setminus \adv{V}(\cliqueset{T})}w_v - \sum_{v \in \overlap{T}}w_v - \sum_{v \in \Madv}w_v  \bigg)\\
        &\ge w_Tw_i \bigg( 2\ktil - \sum_{v \in S_\ell}w_v - O(\sqrt{n}) - o(k) - \epsnode k  \bigg)\end{aligned} \right\}.
    \end{align*}
    For the remaining sum above, we use the fact that by \Cref{lem:degreebound}, we get that the maximum degree of any $i \in [n] \setminus (S_\ell \cup \Madv)$ into $S_\ell \setminus \Madv$ is at most $(1 + \varepsilon)pk$ w.h.p. Since $p$ is bounded away from $1$, we get that for any sufficiently small $\epsnode$, there is a constant $c > 0$, such that  $
        \bicliqueaxiomsalgoapproxconcretereduced \sststile{w}{O(t)} \left\{ w_Tw_i|y| \ge w_Tw_i c k  \right\}.
    $ Hence, due to our bound on $\Delta$, we get that for sufficiently small $\epsedge$,
    $$
        \bicliqueaxiomsalgoapproxconcretereduced \sststile{w}{O(t)} \left\{ w_Tw_i\ \left(\maxbalall{3} |y| - \Delta - 12 \epsedge k\right) \ge c'k \ w_Tw_i \right\},
    $$ for some constant $c' > 0$. Hence,
    \begin{align*}
         \bicliqueaxiomsalgoapproxconcretereduced \sststile{w}{O(t)} \left\{ \begin{aligned} & w_Tw_i \ \left( \sum_{S \subseteq \lefttadv, |S| = 3} w_S \right)^2 \le O\bigg(\frac{n^2}{k}\bigg)  \ w_T w_i \left( \sum_{S \subseteq 
         \lefttadv, |S| = 3} w_S \right) \end{aligned} \right\}.
    \end{align*}
    Using the cancellation $\{ f^2 \le C  \} \sststile{x}{2d} \{ f \le \sqrt{C} \}$, we get
    \begin{align*}
         \bicliqueaxiomsalgoapproxconcretereduced \sststile{w}{O(t)} \left\{ \begin{aligned} & w_Tw_i \ \left( \sum_{S \subseteq \lefttadv, |S| = 3} w_S \right)^2 \le O\bigg(\frac{n^4}{k^2}\bigg) \ w_T w_i \end{aligned} \right\}.
    \end{align*} Using \Cref{lem:sos-factorial} and multiplying by $\{ (\sum_{v \in [n]} w_v)^2 \le 8k^2 \}$, we get $
         \bicliqueaxiomsalgoapproxconcretereduced \sststile{w}{O(t)} \{  w_Tw_i \ ( \sum_{v \in \lefttadv} w_v )^8 \le O(n^4) w_T w_i \}.
    $ Applying the cancellation the cancellation $\{ f^2 \le C  \} \sststile{x}{2d} \{ f \le \sqrt{C} \}$ three times, we get 
    $$
         \bicliqueaxiomsalgoapproxconcretereduced \sststile{w}{O(t)} \left\{ w_Tw_i \  \left( \sum_{v \in \lefttadv} w_v\right) \le O(\sqrt{n}) w_T w_i \right\}
    $$  and the proof is almost finished. It only remains to note that 
    \begin{align*}
         \bicliqueaxiomsalgoapproxconcretereduced \sststile{w}{O(t)} \bigg\{ \sum_{v \in S_\ell \setminus \Madv} w_v \le \sum_{v \in \lefttadv} w_v + \sum_{v \in \overlap{T}} w_v \le \sum_{v \in \lefttadv} w_v + o(k) \bigg\}
    \end{align*} since $\overlap{T} = o(k)$ w.h.p. by \Cref{lem:smalloverlap}. With this, our statement follows.
\end{proof}

\section{ Open questions and future work }\label{sec:discussion}

We leave two main open problems that we consider to be exciting directions for future work.

\paragraph{Algorithms for RIGs with ``two-sided'' noise}

A natural extension of our (noisy) model of RIGs would be to consider a model where we replace the ground-truth cliques with dense sub-graphs.

\begin{definition}[Random Intersection graphs with two-sided noise]
    Define a random graph $\textsc{RIG}_2(n, d, p, \qplus, \qminus)$ with $\qminus < \qplus$ on vertex set $[n]$ as follows. Every vertex $v$ draws a set $M_v \subseteq [d]$ by including each $i \in [d]$ independently with probability $\delta = \delta(d, p, \qplus, \qminus)$. Then, for every pair $u, v \in V$, we include an edge with probability \begin{align*}
        \Pr{u \sim v} = \begin{cases}
            \qplus & \text{ if } M_v \cap M_u \neq \emptyset\\
            \qminus & \text{ otherwise.}
        \end{cases}
    \end{align*} The probability $\delta = \delta(d, p, \qplus, \qminus)$ is chosen such that $\Pr{u \sim v} = p$ for all $u,v\in [n], u \neq v$.
\end{definition} 
Here, edges within some ground truth community form with probability $\qplus$ while edges outside form with probability $\qminus < \qplus$. Can our algorithms be extended to this case? The techniques presented in \Cref{sec:stronger-adv} partially cover the case where up to a constant fraction of edges (that are adversarially chosen) are missing in the ground-truth $S_\ell$, but they do not directly extend to the two-sided case. Here, a significantly larger number of missing edges has to be handled, which, however, is \emph{not} adversarially chosen but random. We believe our techniques to be a good starting point for achieving recovery in $\textsc{RIG}_2(n, d, p, \qplus, \qminus)$, but some further ideas will be required. Moreover, it would be interesting to see if the algorithms on RIGs with two-sided noise can still be made robust to a \emph{monotone} adversary that is now not only allowed to delete edges outside of the ground-truth communities, but also to \emph{add back} edges within the communities. We believe this to be the case, however it is known that the capabilities of a such adversary \emph{do} shift the information-theoretic threshold for recoverability in the stochastic block model \cite{Moitra_Perry_Wein_2016}, so this question stil deserves a formal investigation.

\paragraph{Hardness of recovery in dense RIGs without noise}

Moreover, it would be very interesting to better understand the \emph{computational-statistical trade-offs} for recovering communities in RIGs. We show in \Cref{sec:hardness} that recovery is hard in the \emph{noisy model} if $k \ll \sqrt{n}$ like it is the case for classical planted clique. Together with the information-theoretic feasibility proved in \Cref{sec:slctalpha02}, this shows that there is a computational-statistical gap whenever $n^\varepsilon \ll k \ll \sqrt{n}$ and $p, q$ are a sufficiently small constants. However, our reduction from planted clique in \Cref{thm:hardness} crucially relies on the fact that we have independent edges outside of the communities in our model $\RIGone$. 

It would be very interesting to find out if this hardness carries over to the case where $p$ is a constant and $q = 0$, i.e., the \emph{dense, noiseless} case. We strongly believe this to be the case, however, finding evidence for a such computational statistical gap seems to be quite challenging. While there are multiple ways of giving evidence for such gaps exist (e.g. low-degree polynomials, SoS lower bounds, the overlap gap property, statistical query algorithms, reductions from other statistical problems), none of the techniques known to us seems to easily handle dense, noiseless RIGs. Finding such techniques would likely be interesting in its own right and might further shed light on other computational aspects of other average-case problems.

\bibliography{literature}

\appendix
\section{Lower bound on the spectral norm of the centered adjacency matrix}\label{sec:spectrallowerbound}

\renewcommand{\matA}{\ensuremath\mathbf{A}}
\newcommand{\vecx}{\ensuremath\mathbf{x}}
\newcommand{\vecy}{\ensuremath\mathbf{y}}
\renewcommand{\ind}{\ensuremath{\mathds{1}}}
\newcommand{\HighDegV}{\ensuremath{L}}

In this section we prove the lower-bound on the spectral norm of the centered adjacency matrix of a random intersection graph. 
For this, we will use the Berry--Esseen theorem (see e.g. \cite{klenke2008probability}).

\begin{theorem}[Berry--Esseen]\label{thm:BE}
Let, for $n\in\mathbb N$, $X_1,X_2\dots,X_n$ be i.i.d. random variables with $\Expectednop{X_1}=0$, $\Expectednop{X_1^2}=\sigma^2\in(0,\infty)$ and $\Expectednop{X_1^3}<\infty$. Let $X=\frac{1}{\sqrt{n\sigma^2}}(X_1+\dots+X_n)$ and let $\Phi$ be the distribution function of the standard normal distribution. Then,
\[
\sup\left|\Pr{S_n\leq x}-\Phi(x)\right|\leq\frac{0.8\Expectednop{X_1^3}}{\sigma^3\sqrt n}.
\]
\end{theorem}

\SpectralLB*
\begin{proof}
To prove the lemma we show that $\matA$ has sufficiently many rows that have sufficiently many ``$+1$'' entries. In order to do that, we first calculate the probability that, for each $v$, we have sufficiently many elements in $M_v$. For $i\in[d]$ define $Y_i=\ind(i\in M_v)-\delta$, and observe that $\Expectednop{Y_i^2}=\delta(1-\delta)$ and $\Expectednop{|Y_i|^3}=\delta(1-\delta)(\delta^2+(1-\delta)^2)$. Now define 
\[X=\frac{1}{\sqrt{d\delta(1-\delta)}}\sum_{i\in[d]}Y_i=\frac{1}{\sqrt{d\delta(1-\delta)}}(|M_v|-\delta d).\] 
We can apply \cref{thm:BE} on $X$ and obtain that,

\begin{equation}\label{eq:BE-application}
\sup\left|\Pr{X\leq x}-\Phi(x)\right|\leq\frac{0.8\delta(1-\delta)(\delta^2+(1-\delta)^2)}{(\delta(1-\delta))^{3/2}\sqrt d}=\frac{0.8(\delta^2+(1-\delta)^2)}{\sqrt{d\delta(1-\delta)}}\in o(1),
\end{equation}
since $\delta\in O(d^{-1/2})$ for our choice of parameters and \Cref{lem:sizeofdelta}.

We now have that,
\[
\Pr{|M_v|\geq d\delta+\sqrt{d\delta}}=\Pr{X\geq\frac{1}{\sqrt{1-\delta}}}\geq 1-\Phi\left(\frac1{\sqrt{1-\delta}}\right)-o(1),
\]
where the last inequality follows from \eqref{eq:BE-application}. Observe that $\Phi(\frac1{\sqrt{1-\delta}})$ can be bounded above by a constant, yielding that, for all $v\in V$, $|M_v|\geq d\delta+\sqrt{d\delta}$ with constant probability. 

Let $\HighDegV$ be the set of vertices $v$ such that $|M_v|\geq d\delta+\sqrt{d\delta}$. Since each vertex has at least constant probability of being in $\HighDegV$, using Chernoff's bound on the set of vertices shows that w.h.p. over the draw of $G\sim RIG(n,d,1/2,0)$, $|\HighDegV|=\Theta(n)$. 

Next, we bound the number of neighbors of any $u \in \HighDegV$ into $\HighDegV$. To this end, we fix a vertex $u \in \HighDegV$ and consider a vertex $v \in \HighDegV \setminus \{u\}$. Then, revealing the labels of $v$ we obtain
$$
    \Pr{u \sim v \mid v \in \HighDegV} \ge \Pr{u \sim v} \geq 1 - (1-\delta)^{d\delta + \sqrt{d\delta}} \geq 1 - e^{-(d\delta^2 + \delta \sqrt{d \delta})}.
$$
Using that $\exp(\delta\sqrt{d\delta})=1+\Theta(\delta\sqrt{d\delta})$ in conjunction with $\delta = \sqrt{1 - 2^{1/d}}$ (for our choices $p = 1/2$ and $q=0$) yields that for any $v\in [n] \setminus \{u\}$, $$\Pr{u\sim v}=1/2+\Omega(\delta\sqrt{d\delta}).$$
Using a standard Chernoff bound in combination with a union bound, it now follows that w.h.p. for all $u \in \HighDegV$ it holds that $u$ has $(1/2 + \Omega(\delta\sqrt{d\delta}))|\HighDegV|$ neighbors in $L$.
 
 

Now, let $\vecy$ be the vector assigning 
\[
y_i=\begin{cases} |\HighDegV|^{-1/2} &\textrm{ if } i\in \HighDegV \\ 0 & \textrm{otherwise.} \end{cases}
\]
Clearly $\|\vecy\|=1$. For $i,j\in[n]$, let $A_{i,j}$ be the $i,j$-th entry of $\matA$ and let $r_i=\sum_{j\in \HighDegV }A_{i,j}$. As shown in the previous paragraph, we have $r_i = \Theta(n\delta\sqrt{d\delta})$ for all $i\in \HighDegV$. With this, we can bound 
\begin{align*}
\|\matA\| \geq \vecy^\top \matA\vecy \ge
\frac{1}{|\HighDegV|} \sum_{i \in \HighDegV} r_i = \Theta( n \delta \sqrt{\delta d}  ) = \Theta(  n  d^{-1/4} ),
\end{align*}
which yields the statement.
\end{proof}

\begin{remark}
    Note that the above proof equally applies to a model of planting equal-sized cliques: instead of each vertex $v$ drawing elements in $M_v$ from $[d]$ independently and with probability $\delta$, we can also consider the model where each element in $[d]$ is assigned to $k$ vertices of $G$ uniformly at random. Then, the application of the Berry--Esseen theorem holds also in this setting, and the rest of the proof follows with minor modifications.
\end{remark}

\section{Hardness for $k \ll \sqrt{n}$}\label{sec:hardness}

In this section, we provide a short proof that exact recovery in $\RIGone$ for $p, q = \Omega(1)$ is hard whenever $k \ll \sqrt{n}$, assuming the planted clique hypothesis. Our proof uses a simple average-case reduction, similar to the proof of Theorem 2.7 in \cite{kothari2023planted}

\begin{definition}[Planted Clique]\label{def:planted-clique}
For $p = \Theta(1)$, let $G' \sim \Gnp$, choose any subset of vertices $K \subseteq V(G')$ of size $|K| = k$ and let $G=(V(G'), E(G') \cup \binom{K}{2})$. Then the \emph{planted clique} problem is to recover $K$ from $G$.    
\end{definition}

\begin{conjecture}[Planted Clique Hardness]\label{con:planted-clique-hardness} Let $\eps > 0$ be any constant. Then for $k \leq n^{1/2 - \eps}$ the planted clique problem~(\Cref{def:planted-clique}) cannot be solved in polynomial time.
\end{conjecture}

\begin{theorem}[Lower bound]\label{thm:hardness}
 Assume planted clique hardness \Cref{con:planted-clique-hardness} and let $k = \delta n = n^{1/2 - \Omega(1)}$, $q = \Theta(1)$. Then there is no polynomial time algorithm that given a graph $G \sim \RIGone$, outputs a set of $d$ cliques that agrees with $\{S_\ell\}_{\ell \in [d]}$ w.h.p. 
\end{theorem}

\begin{proof}
   \textbf{Planting the clique.} Consider a graph $G' \sim \Gnp{q}$ and initialize $K_1 = \emptyset$. Then for each vertex $v \in V(G')$ flip a coin such that with probability $\delta$, $v$ is included in $K_1$. Notice that by assumption $\E{|K_1|} = \delta n \leq n^{1/2 - \Omega(1)}$ and since $|K_1|$ is a binomial random variable it follows by a Chernoff bound that $|K_1| \leq n^{1/2 - \Omega(1)}$ with probability $1-o(1)$. Thus, "planting" $K_1$ as a clique into $G'$ we obtain $G_1=(V(G'), E(G') \cup \binom{K_1}{2})$ and thus, assuming \Cref{con:planted-clique-hardness}, there is no polynomial time algorithm to recover $K_1$ from $G_1$ w.h.p. since $q$ is non-vanishing.

   \textbf{Algorithm.} Next, assume there is a polynomial time algorithm \textsc{Alg} that, given a graph $G \sim \RIGone$, outputs a set of $d$ cliques that agrees with $\{S_\ell\}_{\ell \in [d]}$ w.h.p. We show that this would contradict the planted clique conjecture~\Cref{con:planted-clique-hardness} as we could recover $K_1$ from $G_1$ w.h.p. by which our desired statement is proven.

   \textbf{Reduction.} To this end, recalling that $G_1$ is the graph $G' \sim \Gnp{q}$ with the planted clique $K_1$, we add edges to $G_1$ such that the resulting graph follows the distribution of $G\sim\RIGone$. We achieve this as follows: For all $\ell \in [2\ldots d]$, initialize $K_\ell = \emptyset$ and for rounds $\ell = 1, \ldots, d-1$, we flip in each round $\ell$ a coin for each vertex $v \in V(G')$. With probability $\delta$ include $v$ in $K_{\ell+1}$ and at the end of each round $\ell$, "plant" clique $K_{\ell+1}$ into $G_{\ell}$, i.e., $G_{\ell+1} := (V(G_{\ell}), E(G_{\ell}) \cup \binom{K_{\ell+1}}{2})$. The resulting graph $G_d$ after $d-1$ rounds is $G_1$ with additional $d-1$ cliques. Notice that the this procedure requires $O(d\cdot n)$ time and thus, if \textsc{Alg} solves planted clique on $G_d$ in polynomial time, we have a polynomial time algorithm for planted clique on $G_1$.

   \textbf{Solving planted clique via \textsc{Alg}.} Note that after the reduction outlined in the paragraph above, we obtain a graph $G_{d} \sim \RIGone$ where each clique $K_\ell$ represents the clique with label $\ell \in d$ and thus, $G_{d} \stackrel{d}{=} G \sim \RIGone$. Moreover, we observe that the initial planted clique $K_1$ in $G_1$ is also included in the set $\{S_\ell\}_{\ell \in [d]}$ of $G_d$. Thus, if \textsc{Alg} outputs a set of $d$ cliques that agrees with $\{S_\ell\}_{\ell \in [d]}$ w.h.p., it also outputs $K_1$ w.h.p. We then exploit this as follows: Fix any $\ell \in d$ and check if clique $K_\ell$ is included in $G_1$. We can do this in time $|E(K_\ell)| = O(n^2)$ by checking if all edges of $K_\ell$ exist in $G_1$. Moreover, since $d \in o(n)$ we can check for all $\ell \leq d$ cliques if $K_\ell$ is included in $G_1$ in time $O(d \cdot n^2) = o(n^{5/2})$. If the answer is "yes" then $G_1$ contains w.h.p. a planted clique, i.e., the clique $K_1$. Otherwise, the answer is "no" and $G_1 = G'$ as no clique was planted. Since we recover $K_1$ in $G_d$ in polynomial time via \textsc{Alg}, this finishes the proof. 
\end{proof}

\section{RIGs with small $p$}\label{sec:I-have-a-small-p}

\subsection{Coupling sparse and dense RIGs}\label{sec:coupling}

The following statement tells us that any $G \sim \RIGone$ can be coupled with a $G' \sim \text{RIG}(n, d, p', q')$ for any $p' > p$ and some $q' > q$ which is a function of $p, p', q$ in such a way that $G \subseteq G'$ while every ground-truth clique $S_\ell$ is identical in $G, G'$.

The idea is to simply ramp up the noise-levels in $G'$ as compared to $G$ while keeping $\delta$ intact. This is captured in the following lemma.
\begin{lemma}\label{lem:coupling}
    Consider $G \sim \RIGone$ and $G' \sim RIG(n,d, p', q')$ with $p < p', q \in [0,p)$ and $q' = 1 - \frac{(1-p')(1-q)}{1-p}$. Then $\delta(d,p,q) = \delta(d,p',q')$ and $q < q'$.
\end{lemma}

\begin{proof}
    To show that $\delta(d,p,q) = \delta(d,p',q')$ we get in similar fashion to \eqref{eq:delta}
    \begin{align*}
        \delta(d,p',q') = \sqrt{1 - \exp\left( -\log\left(\frac{1-q'}{1-p'}\right)\frac{1}{d} \right)}.
    \end{align*}
    By our choice of $q'$ this yields

        \begin{align*}
        \delta(d,p',q') = \sqrt{1 - \exp\left( -\log\left(\frac{(1-p')(1-q)}{(1-p')(1-p)}\right)\frac{1}{d} \right)}.
    \end{align*}
    Then, again invoking \eqref{eq:delta}, it follows that
    \begin{align*}
         \delta(d,p',q') = \sqrt{1 - \exp\left( -\log\left(\frac{1-q}{1-p}\right)\frac{1}{d} \right)} =  \delta(d,p,q),
    \end{align*}
    as desired.
    Next, to see that $q < q'$, note that due to our choice of $q'$ it holds that
    $$
    (1-p)(1-q') = (1-p')(1-q).
    $$
    Moreover, by our assumption $p' > p$ it also holds $1 -p > 1 - p'$ and thus,
    $$
        (1 - p)(1-q') < (1-p)(1-q) \text{ implying } 1-q' < 1 - q, 
    $$
    that is, $q < q'$.
    
\end{proof}
With this at hand, our coupling follows.
\CouplingSmallP*
\begin{proof}
    Consider $G' \sim RIG(n,d, p', q')$ where  $q' = 1 - \frac{(1-p')(1-q)}{1-p}$. Then, by our choice of $q'$ and $p'$, we have $q' > q$ and $\delta(d,p,q) = \delta(d,p',q') =: \delta$ using \Cref{lem:coupling}. Note that, intuitively, $G$ can be considered a sub-graph of $G'$, where all ground truth cliques that are present in $G'$ are also present in $G$, while edges that exist due to noise in $G'$ only have a certain probability to carry over to $G$ and a non-edge in $G'$ also implies a non-edge in $G$. To make this more formal, draw $G$ and $G'$ as follows:
    \begin{itemize}
        \setlength\itemsep{.0001em}
        \item \textbf{1. Labels:} The labels underlying $G, G'$ are identical. This is a valid coupling due to \Cref{lem:coupling} and yields that all ground truth $S_\ell$ in $G, G'$ are identical too.
        \item \textbf{2. Noise for $G$:} For each pair of vertices $u, v \in \binom{[n]}{2}$ that are not adjacent after step 1, draw a number $r(u,v)$ from $(0,1)$ uniformly at random. If $r(u,v) \leq q$ then add the edge $\{u,v\}$.
        \item \textbf{3. Noise for $G'$:} For every pair of vertices $u, v \in \binom{[n]}{2}$ that does not have an edge after step 2 add the edge $\{u,v\}$ if $r(u,v) \leq q'$.
    \end{itemize}
    Observe that the resulting graph after step 2 follows the same distribution as $G$ and the graph after step 3 that of $G'$. Moreover, $G$ and $G'$ are identical after step 2. Finally, after step 3, $G \subseteq G'$ since we only add further edges to $G'$. 
\end{proof}



    
  

\subsection{Simple algorithm for exact recovery in sparse $\RIGone$ for $p \le n^{-\varepsilon}$}\label{sec:recoverysparse}
We show that exact recovery in the sparse case ($p \le n^{-\varepsilon}$) is not a very challenging problem (even for $n^{\varepsilon} \ll k \ll \sqrt{n}$) and can be accomplished by a simple, polynomial-time combinatorial algorithm. This phenomenon is in line with existing work on other variants of planted clique: it was observed e.g. in \cite{Feige_Grinberg_2024} that this also applies to ordinary planted clique. This suggests that the dense case is actually the only challenging regime. However, we remark that due to the fact that our algorithms work in the semirandom model (i.e. under presence of a monotone adversary), they also achieve recovery in sparse graphs that are arbitrarily ``sandwiched'' between the noiseless case and the dense, noisy case, and not captured by the simple algorithm we present in the following. 

Our algorithm essentially choses a random $t$-tuple, computes a neighbourhood reduction and checks if the entire neighbourhood is a clique. Very similar ideas were presented a long time ago in \cite{Behrisch_Taraz_2006}, however, the analysis there only covers a much more restrictive range of parameters, mainly because the authors did not have a single label clique theorem at hand. For the sake of completeness, we decided to add this section here, showing that the simple greedy algorithm succeeds whenever $p \le n^{-\varepsilon}$ for arbitrarily small $\varepsilon > 0$

We focus on the following small-$p$-regime to ensure that $k$ is still of size polynomial in $n$ as captured by the following lemma.
\begin{lemma}\label{lem:small-cliques}
    Let $\alpha \in (0,2), \varepsilon' \in (0, 2\alpha)$ and $k = \delta n$. Let  $G \sim \RIGone$ with parameters $d = n^{\alpha}, p = n^{-\varepsilon'}$ and  $q \in [0,p)$. Then, for all $\epsilon > 0$ any $\ell \in [d]$ it holds $(1-\varepsilon)k<|S_{\ell}|< (1+\varepsilon)k$ w.h.p. and $k = n^{\Omega(1)}$.
\end{lemma}
\begin{proof}
    By our assumption that $q < p$ it holds via \Cref{lem:sizeofdelta} that $\delta \geq \frac{(1-o(1))}{\sqrt{d}}$. Making use of this we obtain for a fixed label $\ell \in [d]$ $$\mathbb{E}[|S_{\ell}|] = k = \delta n = n^{\Omega(1)},$$
    
    since $\alpha \in (0,2)$ and $\varepsilon' \in (0, 2\alpha)$. Two Chernoff bounds yield both desired concentrations for $|S_\ell|$ with probability $1 - n^{\-\omega(1)}$ since $k$ is of polynomial size. A union bound over all $d$ labels finishes the proof. 
\end{proof}
The following lemma tells us that for each ground-truth  $S_\ell$, there exists a $t$-tuple $T \subset [n]$ such that the neighbourhood of $T$ is exactly one of our cliques. 
\begin{lemma}\label{lem:wittness}
    Let $\alpha \in (0,2), \varepsilon' \in (0, 2\alpha)$ and $t = \lceil \frac{1 + 2\alpha}{\varepsilon'}\rceil$. Let  $G \sim \RIGone$ with parameters $d = n^{\alpha}, p = n^{-\varepsilon'}$ and  $q \in [0,p)$. Then, for all $\ell \in [d]$ there exists a $T \subset [n]$ of size $t$ such that $N_G(T) = S_{\ell}$ w.h.p.
\end{lemma}

\begin{proof}
    Fix a label $\ell \in [d]$. First, reveal the randomness of label $\ell$. Consider the clique $S_\ell$ which by \Cref{lem:small-cliques} has polynomial size w.h.p. Partition $S_\ell$ into $m = \lfloor|S_\ell|/t\rfloor$ disjoint subsets $T_1,..,T_m$ each of size $t$. Next, reveal all the missing labels $[d]\setminus \ell$ for the set of vertices $[n]\setminus S_{\ell}$. Using \Cref{lem:labelconcentration}, it holds for all $v \in [n]\setminus S_{\ell}$ that the number of labels is bounded by $|M_v| \leq (1+o(1))(d\cdot \delta)$ w.h.p. This implies by \Cref{lem:pconcentration} that for any pair of vertices $u \in T, v \in [n]\setminus S_{\ell}$ that $\Pr{u \sim v} \leq (1+o(1))p$ over the randomness of revealing the rest of labels for $T$. Thus, the probability that a vertex $v \in [n]\setminus S_{\ell}$ has an edge to all vertices in $T$ is 
    $$
    \Pr{v \sim T} \leq ((1+o(1))p)^t = O(n^{-(1+2\alpha)}),
    $$
    using that $t = \lceil\frac{1 +2\alpha}{\varepsilon'}\rceil$ and $p \leq n^{-\epsilon}$. Using a union bound over all vertices $v \in [n]\setminus S_{\ell}$ yields that $N_G(T_i) = S_{\ell}$ with probability $1 - O(n^{-2\alpha})$. Another union bound over all $d = n^{\alpha}$ labels gives the desired result, that is, for any label $\ell \in [d]$ there exists w.h.p. a $t$-tuple $T$ of size $t$ such that $N_G(T) = S_{\ell}$.
\end{proof}

We now use these $t$-tuples as witnesses to find all ground truth cliques in polynomial time.
\begin{proposition}\label{lem:small-cliques-algorithm}
    Let $\alpha \in (0,2)$ and $\varepsilon' \in (0, 2\alpha)$. Let  $G \sim \RIGone$ with parameters $d = n^{\alpha}, p = n^{-\varepsilon'}$ and  $q \in [0,p)$. Then, there is a polynomial-time algorithm for exact recovery that succeeds w.h.p. over the draw of $\RIGone$.
\end{proposition}
\begin{proof}
   Let $t = \lceil \frac{1 + 2\alpha}{\varepsilon'}\rceil$, $S = \emptyset$ and for each $T \subset [n]$ of size $|T| = t$ consider the following procedure:
   
          If $0.99 \cdot k \leq |N_G(T)| \leq 1.01 \cdot k$ and $N_G(T)$ is a complete graph do $S := S \cup N_G(T)$.  

We claim that with high probability the set $S$ contains all ground truth cliques while not containing any other set and that the above procedure takes polynomial time.

  \textbf{Runtime} Since procedure runs through $\binom{n}{t} \leq n^t$ $t$-tuples and we have to check $O(k^2)$ edges for each iteration, the runtime is $O(n^t \cdot k^2)$. Since $t = O(1)$ this is polynomial in $n$.

  \textbf{Correctness} First, we show that every ground truth clique is contained in $S_{\ell}$. By our choice of $t$ and using \Cref{lem:wittness} there indeed exists a $T$ such that $N_G(T) = S_{\ell}$ w.h.p. Moreover, using \Cref{lem:small-cliques} the size of $S_{\ell}$ is  $0.99 \cdot k \leq |S_{\ell}| \leq 1.01 \cdot k$ w.h.p. and thus, every $S_{\ell}$ is included in $S$ by our procedure. All that is left to show is that w.h.p. there is no set of vertices contained in $S$ that is not a ground truth clique. Note that in this case for a non-empty set $U \subseteq [n]\setminus S_\ell$ it holds $N_G(T) = S_{\ell} \cup U$ such that $S_{\ell} \cup U$ is a complete graph. However, using \Cref{cor:many-labels} no such  clique exists w.h.p. 
\end{proof}

\end{document}